\documentclass[11pt]{article}

\usepackage[a4paper,margin=1.15in]{geometry}
\usepackage{amsmath,amssymb,amsthm,mathtools}
\usepackage{microtype}
\usepackage[numbers,sort&compress]{natbib}
\usepackage[hidelinks]{hyperref}
\newtheorem{theorem}{Theorem}[section]
\newtheorem{proposition}[theorem]{Proposition}
\newtheorem{lemma}[theorem]{Lemma}
\newtheorem{corollary}[theorem]{Corollary}
\newtheorem{remark}[theorem]{Remark}
\newtheorem{definition}[theorem]{Definition}

\newcommand{\eps}{\varepsilon}
\newcommand{\R}{\mathbb R}
\newcommand{\Sph}{\mathbb S}
\newcommand{\dd}{\,\mathrm d}
\newcommand{\Hh}{\mathcal H}
\newcommand{\one}{\mathbf 1}
\newcommand{\supp}{\operatorname{spt}}

\title{Heterogeneous Anisotropic Kac Limits\\on Rectifiable Resolved Cut Spaces}
\author{Sai Peng\\
School of Mathematics and Computational Science, Xiangtan University\\
Xiangtan, Hunan, China\\
\texttt{pscfd@xtu.edu.cn}}
\date{}

\hypersetup{
  pdftitle={Heterogeneous Anisotropic Kac Limits on Rectifiable Resolved Cut Spaces},
  pdfauthor={Sai Peng},
  pdfsubject={Sharp-interface limits for heterogeneous anisotropic Kac energies on rectifiable resolved cut spaces},
  pdfkeywords={heterogeneous Kac kernels, anisotropic perimeter, rectifiable cuts, resolved derivative, long-range kernels}
}

\begin{document}
\maketitle

\begin{abstract}
Ambient convergence does not determine a nonlocal energy on a cut body: a cut
has zero volume, but a microscopic bond retains which labelled bank its path
crosses.  We give a direct measure-theoretic resolution of this information
and prove the corresponding sharp-interface Gamma-limit.  The cut may be a
finite or countable family of Borel, countably rectifiable hypersurfaces with
finite total surface mass, measurable approximate normals, and measurable binary bank
data.  No closedness, positive reach, finite atlas, curvature bound, angle
condition, or regularity of the branch set is assumed.  For heterogeneous Kac
interactions generated by a uniformly spanning family of possibly
noncommuting flows, the limit is the resolved anisotropic variation plus the
outer boundary contact and one coefficient-one contact on each bank.  The
Gamma-functional contains no compulsory lower-dimensional term where a crack
approaches the boundary; constant recovery creates no concentration there.  The
proof combines an exact spacetime area formula for flow crossings, a
one-dimensional split-extension lower bound, and ordered-flow compactness.
We also remove compact radial support.  A finite first moment yields the same
local Griffith functional; at the power-law threshold the limits separate
into local, logarithmically renormalized local, and genuinely fractional
regimes.  Smooth slit and junction detectors are recovered as uniformly
equivalent realizations of the direct crossing rule.
\end{abstract}



\paragraph{Convention on equicoercivity.}
Whenever a family indexed by \(\eps\) is called equicoercive, the assertion is
sequential at the singular scale: for every \(\eps_n\downarrow0\), every
sequence with uniformly bounded energy has a subsequence converging in the
stated topology.  No uniform assertion for \(\eps\) bounded away from zero is
intended.

\section{Introduction}
Sharp-interface limits turn nonlocal interactions into local geometric
energies.  Their usual formulation on an ambient \(L^1\) space rests on a
tacit premise: once the interaction length vanishes, the almost-everywhere
phase and the first-order distribution of bond directions retain all
information relevant at surface scale.  This premise fails on a cut body.  A
cut has zero Lebesgue measure, but a microscopic path can still record which
labelled bank it crosses.  Endpoint fields with the same ambient \(L^1\) limit
can therefore carry different contact multiplicities and different limiting
energies.  Neither ambient convergence nor the kernel alone determines the
Gamma-limit.

The obstruction is a structural mismatch between topology and measure.
Measure-theoretic state spaces identify changes on null sets, whereas paths
through the complement can distinguish the two sides of a codimension-one
null set.  A smooth cut can be resolved by a tubular collar or by duplicating
its banks, but those constructions cease to be intrinsic for a merely
rectifiable family with accumulating sheets, irregular branch sets, or
boundary contact.  One must retain bank multiplicity without adding physical
volume, define crossing incidence measurably, and determine whether
surface-null strata force an additional term.  Thus the central question is
not only which surface density emerges, but which state space and path rule
make a local variational limit well posed.

We answer this question directly on a rectifiable cut, without first replacing
its geometry by a smooth or polyhedral carrier.  Let
\(\mathfrak K=(K_j,\nu_j,g_j^+,g_j^-)_j\) be a finite or countable labelled
family of countably \((d-1)\)-rectifiable sets with
\(\sum_j\Hh^{d-1}(K_j)<\infty\) and surface-null label overlaps.  The
measure-theoretic resolved derivative removes from \(Du\) precisely the jump
already carried by the prescribed cut,
\[
 D_{\mathfrak K}u
 :=Du-\sum_j(T_j^+u-T_j^-u)\nu_j
       \Hh^{d-1}\mathbin{\llcorner}K_j.
\]
At finite scale, a flow arc with no crossing uses the physical endpoint
difference, while an arc with one labelled crossing is split into two
independent endpoint-to-bank contacts.  The value on multiply crossed arcs is
irrelevant at surface order.

Theorem~\ref{thm:rectifiable-cut-master-gamma} proves that these direct
crossing-bond energies Gamma-converge, on every bounded Lipschitz body, to
\begin{align}
 &\int_\Omega h_{\rm fl}(x,\sigma_{D_{\mathfrak K}u})
      \dd|D_{\mathfrak K}u|
 +\int_{\partial\Omega}h_{\rm fl}(x,\nu_\Omega)|T_\partial u-g_\partial|
    \dd\Hh^{d-1}\notag\\
 &\qquad+\sum_j\int_{K_j}h_{\rm fl}(x,\nu_j)
   \bigl(|T_j^-u-g_j^-|+|T_j^+u-g_j^+|\bigr)\dd\Hh^{d-1}.
 \label{eq:introduction-master-limit}
\end{align}
The result requires no closedness, positive reach, local graph
representation, finite atlas, angle gap, or regularity of the branch and tip
sets.  It also permits the cut to approach or terminate on the outer boundary:
the combined crossing set is sliced at once, and the Gamma-limit contains no
compulsory term on \(\overline K\cap\partial\Omega\).  Constant recovery has
no concentration on that set.  This statement concerns the limiting
functional; arbitrary bounded-energy sequences may still carry a nonnegative
excess defect measure.  Thus finite smooth slit and fan complexes are
examples, not the geometric scope of the theorem.

This conclusion supplies a state-space principle at the interface of
geometric measure theory and free-discontinuity analysis.  Codimension-one
trace multiplicity can survive a nonlocal-to-local limit even though its
carrier is invisible to ambient volume, while codimension-two incidence may
remain geometrically present without producing a compulsory term in the
Gamma-functional.  The proof therefore provides a template for nonlocal
approximations on resolved rectifiable and stratified spaces: the limiting
variation, the trace space, and the finite-scale path rule must be specified
together.  The principle is not restricted to the particular slit and fan
charts used for its regular realizations, although it does not classify
arbitrary path-lifting algorithms.

The connection with the Griffith criterion is variational rather than
evolutionary.  In a quasistatic fracture model, the functional in
\eqref{eq:introduction-master-limit} supplies the resolved surface and contact
part associated with a prescribed cut
\citep{francfort1998revisiting,dalmaso2002model}.  Coupling it to a bulk stored-energy
functional would produce a total energy of the form
\(\mathcal E_{\rm bulk}+\mathcal H_{0,\Omega,\mathfrak K}^{g}\).  Only after
specifying irreversible variations of the cut can one define a configurational
energy-release rate and formulate the local Griffith inequality.  The present
paper identifies the anisotropic cost against which that release rate would be
compared; it does not derive a crack evolution law or an energy-balance
equality at propagation.

The second main result identifies the radial threshold hidden by the usual
finite-range assumption.  For a radial law \(\lambda\), a finite first moment
\(m_1=\int_0^\infty t\lambda(\dd t)\) is the exact integrability needed by
the surface-order truncation argument.  Under one nondegenerate spreading
interval, Theorem~\ref{thm:long-range-first-moment-main} gives the same local
functional multiplied by \(m_1\).  For
\(\lambda_s(\dd t)=a_s\one_{[1,\infty)}t^{-1-2s}\dd t\),
Theorem~\ref{thm:power-tail-main} gives the sharp trichotomy
\[
 \begin{array}{c@{\qquad}l}
 s>\tfrac12 & \text{local resolved Griffith limit},\\
 s=\tfrac12 & \text{local limit after }|\log\eps|^{-1}
                     \text{ renormalization},\\
 0<s<\tfrac12 & \text{flow-fractional resolved-cut limit}.
 \end{array}
\]
The last functional remembers crossings at every physical scale and therefore
cannot be represented by a local surface density.

The homogeneous Kac-to-perimeter theory begins with the work of Alberti and
Bellettini, including the identification of planar profiles
\citep{alberti1998nonlocal,alberti1998optimal}.  General nonlocal-to-local
compactness and integral representation are developed in
\citep{ponce2004approach,alicandro2023variational}.  An unpublished preprint
of Brusca, Donati, Scalabrino, Trifone, and Voglino
(arXiv:2603.24192) develops convolution-type approximations for free
discontinuity problems.  These results act on the ambient domain and therefore
do not determine bank multiplicity or junction cells.  Metric-measure \(BV\)
theory provides an
intrinsic variation on more general spaces
\citep{miranda2003functions,ambrosio2014equivalent}, but an abstract
variation measure does not specify how a microscopic bond is lifted across a
labelled cut.  The present limit requires both ingredients: an intrinsic
state space and the finite-scale path rule that selects its local cells.

The heterogeneous interaction introduces another obstruction.  In an
unpublished preprint (arXiv:2410.19624), Caldwell treats
translation-invariant kernels that may be singular at the origin and have
noncompact tails.  Another unpublished preprint, by Davoli and Tasso
(arXiv:2412.11756), places the spatial heterogeneity in a potential with moving
wells and constructs regular parameter-dependent profiles.  Here
the wells are fixed, the finite-range kernel depends on the base point, and
the bond directions may rotate through noncommuting flows.  No continuous
selection of near-minimizing frozen profiles is assumed, and no common flow
chart is available.  Thus neither of those theories contains the present
one.  The finite-first-moment and power-tail results below concern
flow-generated resolved bonds and do not subsume their full kernel classes or
moving-well regimes.  Long-range surface interactions and nonlocal phase
transitions provide the nearest variational precedents
\citep{bach2020longrange,berendsen2019nonlocal,savin2012gamma}.

Four analytic mechanisms drive the proof.  First,
Lebesgue differentiation and Vitali selection allocate finitely many frozen
profiles on polyhedral faces.  A shrinking codimension-two replacement removes
their seams, after which strict polyhedral approximation and Reshetnyak
continuity give recovery for every finite-perimeter interface.  A positive
core and pointwise relative lower continuity yield the matching compactness
and blow-up lower bound.

Second, radially spread bonds with rotating directions are represented by
their \(C^{1,1}\) flows.  Since the flows need not commute, iterated averaging
cannot be reduced to convolution in a single coordinate system.  The
ordered-flow smoothing lemma constructs local coordinates from a uniformly
spanning ordered frame, controls commutator and Jacobian errors, and converts
the nonlocal bond budget into a uniform \(BV\) bound.  It identifies the
anisotropy
\[
 h_{\rm fl}(x,\xi)=\sum_i c_i(x)|v_i(x)\cdot\xi|
\]
while preserving constant binary recovery.

Third, an area formula on
\(K\times(0,h)\) gives the exact flow-crossing multiplicity.  Rectifiable
slicing shows that multiplicity and parity agree to first order and reduces
the sequential liminf to a finite family of one-dimensional split extensions.
This is the step that removes curvature, reach, and branch regularity without
weakening the Gamma-liminf to a fixed-state calculation.  The same slicing
handles crack-boundary contact because internal and outer crossings are simply
different endpoint types on almost every flow line.

Fourth, radial truncation separates local and genuinely nonlocal scales.  A
uniform first-moment estimate removes an integrable tail.  At the critical
power, a controlled logarithmic shell restores compactness; below it, a
fractional translation modulus replaces the \(BV\) estimate.  Kac surface
tension, nonlocal boundary laws, random convolution limits, and
finite-partition relaxation provide the surrounding variational context
\citep{alberti1996surface,mellet2023boundary,braides2021random,
ambrosio1990partitions2,ambrosio2011nonlocal}.

The finite positive-reach theory remains useful because it realizes the
direct rule through ordinary collar ghosts and explicit slit or sectorial
detectors.  Theorem~\ref{thm:regular-detector-equivalence-main} proves uniform
Gamma-scale equivalence between those regular detectors and the direct
crossing rule.  The slit and fan cells therefore identify a concrete physical
implementation; they are no longer regularity assumptions in the principal
theorem.

The resolved geometry is motivated by crack and interface models, but every
result here is spatial rather than evolutionary.  Variational fracture treats
creation, growth, and irreversibility through free-discontinuity and elliptic
energies
\citep{ambrosio1990approximation,francfort1998revisiting,dalmaso2002model}.
We fix the rectifiable cut datum and its labels.  We do not claim a theorem
for purely unrectifiable fractals, arbitrary path-lifting algorithms, or
topology-changing evolutions.  Rectifiability is structural: it supplies the
approximate normal, the two \(BV\) traces, and the discrete slices needed by a
Griffith surface law.  Appendix~\ref{app:moving-sectorial-junction} proves
uniform Gamma-convergence under label-preserving bi-\(C^{1,1}\) transport with
fixed incidence and explains why topology change leaves that class.

\paragraph{Core notation.}
The calligraphic letter identifies the interaction mechanism:
\begin{description}
 \item[\(\mathcal K_\eps\)] heterogeneous convolution-kernel energies and
 their localized or frozen forms;
 \item[\(\mathcal F_\eps\)] straight-ray energies with fixed directions;
 \item[\(\mathcal G_\eps\)] flow-generated energies and their lifted versions
 on resolved cut spaces;
 \item[\(D_{\mathfrak K}u\)] the resolved derivative, obtained by deleting
 the jump measure already carried by the labelled cut;
 \item[\(X_{\mathfrak C}\)] the intrinsic completion of a finite cut complex,
 carrying one copy of physical volume and separate open-bank traces;
 \item[\(\Psi_\eps\)] the direct crossing rule, or its uniformly equivalent
 regular slit and junction realization.
\end{description}
Subscripts specify the body, slit, or finite complex; superscripts specify
coefficients and boundary data.  A tilde denotes pullback to a fixed reference
completion in Appendix~\ref{app:moving-sectorial-junction}.

\section{Energy and Cell Density}

Let \(d\ge2\), let \(\Omega\subset\R^d\) be open, and let
\[
 J:\R^d\times\R^d\to[0,\infty)
\]
be measurable.  We impose the following assumptions:
\begin{align}
 &J(x,Z)=J(x,-Z)\quad\text{for a.e. }Z,\text{ for every }x\in\R^d,
 \label{eq:even}\tag{J1}\\
 &\supp J(x,\cdot)\subset B_{R_J}\quad\text{for every }x\in\R^d,
 \label{eq:range}\tag{J2}\\
 &x\longmapsto J(x,\cdot)\quad\text{belongs to }
 C_{\rm loc}(\R^d;L^1(\R^d)).
 \label{eq:l1-cont}\tag{J3}
\end{align}
Only the values of \(J\) near \(\supp\zeta\) enter the theorem.  In
particular, for every compact \(K\subset\R^d\),
\[
 M_K:=\sup_{x\in K}\|J(x,\cdot)\|_{L^1(\R^d)}<\infty.
\]
Let \(W\in C([0,1];[0,\infty))\) satisfy \(W(0)=W(1)=0\).
Fix \(\theta\in[0,1]\).  For \(u:\R^d\to[0,1]\) and
\(\zeta\in C_c(\Omega)\), \(\zeta\ge0\), set
\begin{align}
 \mathcal K_\eps^{J,\theta}(u;\zeta)
 &:=\frac1\eps\int_\Omega\zeta(x)W(u(x))\dd x \notag\\
 &\quad+\frac1{4\eps}\int_\Omega\zeta(x)
 \int_{\R^d}J(x-\theta\eps Z,Z)
 |u(x)-u(x-\eps Z)|^2\dd Z\dd x.
 \label{eq:energy}
\end{align}
The coefficient and the values at the second endpoint are sampled only inside
an \(O(\eps)\)-neighbourhood of \(\supp\zeta\).  Since the test support is
interior and the range is finite, no boundary cell is sampled for all
sufficiently small \(\eps\).  Midpoint sampling \(\theta=1/2\), together with
\eqref{eq:even}, gives a symmetric pair coefficient before multiplication by
the localizing weight.

Let \(\mathcal A\) be the set of measurable \(q:\R\to[0,1]\) for which some
\(L_q<\infty\) satisfies
\[
 q(s)=0\quad(s\le-L_q),\qquad q(s)=1\quad(s\ge L_q).
\]
For \(x\in\Omega\), \(n\in\Sph^{d-1}\), and \(q\in\mathcal A\), define
\begin{align}
 e(x,n;q)
 &:={\int_\R W(q(s))\dd s}
 +\frac14\int_\R\int_{\R^d}J(x,Z)
 |q(s)-q(s-Z\cdot n)|^2\dd Z\dd s,
 \label{eq:profile-cost}\\
 \gamma(x,n)&:=\inf_{q\in\mathcal A}e(x,n;q).
 \label{eq:cell}
\end{align}

\begin{lemma}[Cell regularity]\label{lem:cell-regularity}
For every \(q\in\mathcal A\), the map
\((x,n)\mapsto e(x,n;q)\) is continuous on
\(\Omega\times\Sph^{d-1}\).  The cell density \(\gamma\) is jointly upper
semicontinuous, even in \(n\), and locally uniformly bounded.  More precisely,
for every compact \(K\Subset\Omega\),
\[
 0\le\gamma(x,n)\le \frac14\int_{\R^d}J(x,Z)|Z\cdot n|\dd Z
 \le \frac{R_JM_K}{4}
\]
for \((x,n)\in K\times\Sph^{d-1}\).
\end{lemma}

\begin{proof}
Fix \(q\in\mathcal A\) and write
\[
 I_q(a):=\int_\R|q(s)-q(s-a)|^2\dd s.
\]
Write \(q=q_0+w\), where
\(q_0=\one_{(0,\infty)}\) and \(w\in L^2(\R)\) has compact support.  The map
\(a\mapsto q_0-q_0(\cdot-a)\) is continuous into \(L^2(\R)\), because the
squared distance between the increments at \(a\) and \(b\) is \(|a-b|\).
Translation continuity in \(L^2(\R)\) gives the same conclusion for
\(w-w(\cdot-a)\).  Hence
\(a\mapsto q-q(\cdot-a)\) is \(L^2\)-continuous and \(I_q\) is continuous.
On \([-R_J,R_J]\) it is bounded by a constant depending only on \(q\) and
\(R_J\).  Thus
\[
 e(x,n;q)=\int_\R W(q(s))\dd s
 +\frac14\int_{B_{R_J}}J(x,Z)I_q(Z\cdot n)\dd Z.
\]
If \((x_k,n_k)\to(x,n)\), the part caused by \(x_k\to x\) tends to zero by
the \(L^1\)-continuity in \eqref{eq:l1-cont}; the part caused by
\(n_k\to n\) tends to zero by dominated convergence.  Hence the fixed-profile
cost is continuous.  Its infimum over \(q\) is upper semicontinuous.

Evenness follows by \(Z\mapsto-Z\) and \eqref{eq:even}.  Finally use the step
profile \(q_0=\one_{(0,\infty)}\).  Its potential cost vanishes, while the set
of \(s\) for which \(q_0(s)\ne q_0(s-Z\cdot n)\) has length
\(|Z\cdot n|\).  This proves the displayed bound.
\end{proof}

Whenever a set is called polyhedral near a compact set \(K\), its boundary is
understood to agree in an open neighbourhood of \(K\) with a finite
polyhedral complex.  No regularity away from that neighbourhood is intended.

\begin{theorem}[Heterogeneous interior recovery]\label{thm:main}
Let \(\zeta\in C_c(\Omega)\), \(\zeta\ge0\), and let \(E\subset\Omega\) have
finite perimeter in a neighbourhood of \(\supp\zeta\).  Then there exist
\(u_\eps:\R^d\to[0,1]\) such that
\(u_\eps\to\one_E\) in \(L^1_{\rm loc}(\Omega)\) and
\[
 \limsup_{\eps\downarrow0}\mathcal K_\eps^{J,\theta}(u_\eps;\zeta)
 \le
 \int_{\partial^*E}\zeta(x)\gamma(x,\nu_E(x))
 \dd\Hh^{d-1}(x).
\]
Here \(\nu_E=D\one_E/|D\one_E|\) is the normal from the zero phase to the
one phase.
\end{theorem}

\section{Profile Allocation on a Face}

The next lemma is the point at which spatial heterogeneity is handled.  It
uses no regularity of a minimizing profile as a function of \(x\).

\begin{lemma}[Finite profile allocation]\label{lem:allocation}
Let \(F\) be a bounded relatively open polyhedral subset of an affine
hyperplane with unit normal \(n\), compactly contained in \(\Omega\).  For
every \(\eta>0\), there are finitely many pairwise disjoint relatively open
polyhedral sets \(B_0,B_1,\ldots,B_N\subset F\), covering \(F\) up to an
\(\Hh^{d-1}\)-null set, and profiles \(q_0,q_1,\ldots,q_N\in\mathcal A\)
such that \(q_0=\one_{(0,\infty)}\) and
\[
 \sum_{a=0}^N\int_{B_a}\zeta(x)e(x,n;q_a)\dd\Hh^{d-1}(x)
 \le
 \int_F\zeta(x)\gamma(x,n)\dd\Hh^{d-1}(x)+\eta.
 \label{eq:allocation}
\]
The union of the relative boundaries of the \(B_a\)'s is a finite
codimension-two polyhedral complex.
\end{lemma}

\begin{proof}
Set \(f(x)=\zeta(x)\gamma(x,n)\) on \(F\).  Lemma
\ref{lem:cell-regularity} makes \(f\) bounded and measurable.  Let \(G\subset
F\) be the full-measure set of its Lebesgue points.  Fix a small
\(\delta>0\).  For each \(x\in G\), choose \(q_x\in\mathcal A\) with
\[
 e(x,n;q_x)\le\gamma(x,n)+\delta.
\]
The function \(g_x(y):=\zeta(y)e(y,n;q_x)\) is continuous at \(x\) by
Lemma~\ref{lem:cell-regularity}.  Hence, for all sufficiently small cubes
\(Q\Subset F\), centred at \(x\) and with axes fixed in the hyperplane,
\[
 \int_Q g_x\dd\Hh^{d-1}
 \le \int_Q f\dd\Hh^{d-1}+C\delta\Hh^{d-1}(Q),
 \label{eq:local-cube}
\]
where \(C\) depends only on \(\|\zeta\|_\infty\).  Indeed, after division by
\(\Hh^{d-1}(Q)\), the two sides converge respectively to \(g_x(x)\) and
\(f(x)\).

The cubes satisfying \eqref{eq:local-cube}, with arbitrarily small side
length, form a Vitali cover of \(G\).  Select a pairwise disjoint countable
subfamily \((Q_a)_{a\ge1}\) covering \(F\) up to a null set.  Let
\[
 C_0:=\sup_{y\in F}e(y,n;q_0)<\infty,
\]
which follows from the step-profile estimate in
Lemma~\ref{lem:cell-regularity}.  Choose \(N\) so large that
\[
 \Hh^{d-1}\left(F\setminus\bigcup_{a=1}^NQ_a\right)
 \le \frac{\delta}{1+C_0\|\zeta\|_\infty}.
\]
Set \(B_a=Q_a\) and \(q_a=q_{x_a}\) for \(1\le a\le N\), where \(x_a\) is
the centre selected with \(Q_a\), and decompose the polyhedral remainder
\(F\setminus\bigcup_{a=1}^N\overline{Q_a}\) into its finitely many relatively
  open polyhedral components, all carrying the same profile \(q_0\); their union
  is denoted by \(B_0\).  Summing \eqref{eq:local-cube} and using the uniform
  bound on the remainder gives, with
  \(R_N:=F\setminus\bigcup_{a=1}^N\overline{Q_a}\),
  \[
  \begin{aligned}
   \sum_{a=0}^N\int_{B_a}\zeta e(\cdot,n;q_a)\dd\Hh^{d-1}
   &\le \int_{\cup_{a=1}^NQ_a}f\dd\Hh^{d-1}
      +C\delta\Hh^{d-1}(F)\\
   &\quad+C_0\|\zeta\|_\infty\Hh^{d-1}(R_N)\\
   &\le \int_F f\dd\Hh^{d-1}+C_F\delta.
  \end{aligned}
  \]
Choose \(\delta\) so that \(C_F\delta\le\eta\).  A finite family of
polyhedral boundaries has a finite codimension-two union, proving the last
claim.
\end{proof}

\section{Polyhedral Recovery}

\begin{lemma}[Codimension-two replacement]\label{lem:collar}
Let \(S\) be a finite codimension-two polyhedral complex in a neighbourhood
of \(\supp\zeta\), and let
\[
 \eps\ll\ell_\eps\downarrow0,
 \qquad \frac{\ell_\eps^2}{\eps}\to0.
\]
If two functions \(u_\eps,v_\eps:\R^d\to[0,1]\) agree outside
\(N_{\ell_\eps}(S)\), then
\[
 |\mathcal K_\eps^{J,\theta}(u_\eps;\zeta)
  -\mathcal K_\eps^{J,\theta}(v_\eps;\zeta)|=o(1).
\]
\end{lemma}

\begin{proof}
Choose a compact \(K\Subset\Omega\) containing a fixed neighbourhood of
\(\supp\zeta\).  In the interaction difference, make for each \(Z\) the
change of variables \(m=x-\theta\eps Z\).  If either endpoint
\(x=m+\theta\eps Z\) or
\(x-\eps Z=m-(1-\theta)\eps Z\) lies in
\(N_{\ell_\eps}(S)\), then
\[
 m\in N_{\ell_\eps+R_J\eps}(S).
\]
The additional condition \(x\in\supp\zeta\) keeps \(m\) in the fixed compact
set \(K\).  The volume of the displayed tubular set inside \(K\) is at most
\(C_S(\ell_\eps+R_J\eps)^2\).  Since \(W\) is bounded on \([0,1]\), the
potential difference is bounded by \(C\ell_\eps^2/\eps\).  For the interaction
part use \(|(a-b)^2-(c-d)^2|\le2\) for \(a,b,c,d\in[0,1]\) and
\(\sup_K\|J(x,\cdot)\|_1=M_K\).  Its absolute difference is at most
\[
 \frac C\eps(\ell_\eps+R_J\eps)^2=o(1).
\]
\end{proof}

\begin{proposition}[Heterogeneous polyhedral recovery]\label{prop:poly}
If \(P\) is polyhedral in a neighbourhood of \(\supp\zeta\), then there are
\(u_\eps:\R^d\to[0,1]\) with
\(u_\eps\to\one_P\) in \(L^1_{\rm loc}(\Omega)\) and
\[
 \limsup_{\eps\downarrow0}\mathcal K_\eps^{J,\theta}(u_\eps;\zeta)
 \le
 \int_{\partial P}\zeta(x)\gamma(x,\nu_P(x))
 \dd\Hh^{d-1}(x).
\]
\end{proposition}

\begin{proof}
Fix \(\eta>0\).  Only finitely many faces of \(P\) meet \(\supp\zeta\).
If there are none, finite range gives
\(\mathcal K_\eps^{J,\theta}(\one_P;\zeta)=0\) for all sufficiently small \(\eps\),
and there is nothing to prove.
Truncate them inside a compact polyhedral neighbourhood on whose outer
boundary \(\zeta=0\).  If \(N_F\) faces remain, apply
Lemma~\ref{lem:allocation} to each face \(F_i\), with phase normal \(n_i\),
using error \(\eta/N_F\).  This produces finitely many face
patches \(B_a\), profiles \(q_a\), and total allocated cost at most
\[
 \int_{\partial P}\zeta(x)\gamma(x,\nu_P(x))\dd\Hh^{d-1}(x)+\eta.
 \label{eq:allocated-total}
\]
Let \(S_\eta\) be the union of the original face junctions and all relative
patch boundaries.  It is a finite codimension-two polyhedral complex.

Let \(L\) exceed all transition radii of the finitely many \(q_a\)'s.  Finite
polyhedral geometry gives \(\kappa_\eta>0\) such that, away from
\(N_r(S_\eta)\), every point within distance \(\kappa_\eta r\) of
\(\partial P\) lies in a unique patch tube.  Set
\(\ell_\eps=\eps^{2/3}\).  For all small \(\eps\),
\[
 \eps(2L+R_J)\le \frac{\kappa_\eta}{4}\ell_\eps.
 \label{eq:separation}
\]
Outside \(N_{\ell_\eps}(S_\eta)\), define in the tube of patch \(B_a\)
\[
 v_\eps(x)=q_a\left(\frac{d_a(x)}\eps\right),
\]
where \(d_a\) is signed distance to the supporting face, positive in \(P\).
Set \(v_\eps=\one_P\) away from all patch tubes.  Inside the collar choose any
measurable interpolation in \([0,1]\), for instance \(\one_P\), and call the
result \(u_{\eps,\eta}\).  The transition tubes have thickness \(O(\eps)\),
while the collar has volume \(O_\eta(\ell_\eps^2)\); hence
\(u_{\eps,\eta}\to\one_P\) locally in \(L^1\).

  Put \(C_\eps=N_{\ell_\eps}(S_\eta)\).  The proof of
  Lemma~\ref{lem:collar} gives, uniformly over all \([0,1]\)-valued choices in
  \(C_\eps\),
  \[
   \frac{C_\eta}{\eps}|N_{\ell_\eps+R_J\eps}(S_\eta)\cap K|
   \le \frac{C_\eta}{\eps}(\ell_\eps+R_J\eps)^2=o(1).
  \]
  Hence all potential terms in \(C_\eps\) and all interactions having at least
  one endpoint there may be removed at an \(o(1)\) cost.  For every patch set
  \[
   B_{a,\eps}:=
   \{y\in B_a:\operatorname{dist}(y,\partial_F B_a)>2\ell_\eps\}.
  \]
  Here \(\partial_F\) denotes relative boundary in the supporting face.  One
  has \(\Hh^{d-1}(B_a\setminus B_{a,\eps})=O_\eta(\ell_\eps)\).  Since the
  potential profile has compact support in \(s\), and the interaction profile
  is supported in \(|s|\le L+R_J\) with \(J\) uniformly bounded in
  \(L^1_Z\), the total cost over these deleted face strips is
  \(O_\eta(\ell_\eps)=o(1)\).
  After deleting both the collar and these strips, condition
  \eqref{eq:separation} ensures that every remaining non-zero interaction has
  its two endpoint projections in one and the same safe patch.

  On the resulting safe patch, for the potential term write
  \(x=y+\eps s n_i\); it
converges to
\[
 \int_{B_a}\zeta(y)\int_\R W(q_a(s))\dd s\dd\Hh^{d-1}(y).
\]
For the interaction term first set \(m=x-\theta\eps Z\), and then write
\(m=y+\eps s n_i\).  Up to the explicitly estimated collar and face-strip
terms, the integral is
\begin{align*}
 \frac14\int_{B_a}\int_\R\int_{\R^d}
 &\zeta(y+\eps s n_i+\theta\eps Z)J(y+\eps s n_i,Z)\\
 &\times|q_a(s+\theta Z\cdot n_i)
       -q_a(s-(1-\theta)Z\cdot n_i)|^2
 \dd Z\dd s\dd\Hh^{d-1}(y),
\end{align*}
where tangential shifts of the patch base are confined to the same removed
strip.  The difference factor vanishes unless \(|s|\le L+R_J\).  On the relevant compact set,
\(\|J(x,\cdot)\|_1\) is uniformly bounded, and
\[
 \sup_{|h|\le (L+R_J)\eps}
 \|J(y+h,\cdot)-J(y,\cdot)\|_1\to0
\]
uniformly for \(y\) in the compact patch and \(|s|\le L+R_J\), by
\eqref{eq:l1-cont}.  Dominated convergence therefore gives the limit
\[
 \frac14\int_{B_a}\zeta(y)\int_\R\int_{\R^d}J(y,Z)
 |q_a(s)-q_a(s-Z\cdot n_i)|^2
 \dd Z\dd s\dd\Hh^{d-1}(y).
\]
Here, for each fixed \(Z\), the last identification uses the translation
\(s\mapsto s+\theta Z\cdot n_i\).
Summing over the finite allocation and using \eqref{eq:allocated-total} yields
\[
 \limsup_{\eps\downarrow0}\mathcal K_\eps^{J,\theta}(u_{\eps,\eta};\zeta)
 \le \int_{\partial P}\zeta\gamma(\cdot,\nu_P)\dd\Hh^{d-1}+\eta.
\]
Write the surface integral on the right as \(G_P\), choose
\(\eta_k\downarrow0\), and let \((K_m)\) exhaust \(\Omega\) by compact sets.
For each \(k\), choose \(\eps_k\downarrow0\) so that, whenever
\(0<\eps\le\eps_k\),
\[
 \mathcal K_\eps^{J,\theta}(u_{\eps,\eta_k};\zeta)
      \le G_P+2\eta_k,
 \qquad
 \sum_{m=1}^k\|u_{\eps,\eta_k}-\one_P\|_{L^1(K_m)}\le k^{-1}.
\]
For \(\eps\in(\eps_{k+1},\eps_k]\), set
\(u_\eps=u_{\eps,\eta_k}\).  Then \(u_\eps\to\one_P\) locally in
\(L^1\) and
\(\limsup_{\eps\downarrow0}\mathcal K_\eps^{J,\theta}(u_\eps;\zeta)\le G_P\).
\end{proof}

\section{Finite-Perimeter Interfaces}

Define the non-negative upper semicontinuous, positively one-homogeneous
integrand with linear growth
\[
 \Phi(x,p)=
 \begin{cases}
  \zeta(x)|p|\gamma(x,p/|p|),&p\ne0,\\
  0,&p=0.
 \end{cases}
\]

\begin{proposition}[Localized polyhedral approximation]\label{prop:strict}
Let \(E\) have finite perimeter in a neighbourhood of \(\supp\zeta\).  There
are measurable sets \(P_j\), polyhedral in a fixed neighbourhood of
\(\supp\zeta\), such that \(\one_{P_j}\to\one_E\) in
\(L^1_{\rm loc}(\Omega)\) and
\[
 \limsup_{j\to\infty}
 \int_{\partial P_j}\zeta(x)\gamma(x,\nu_{P_j}(x))\dd\Hh^{d-1}(x)
 \le
 \int_{\partial^*E}\zeta(x)\gamma(x,\nu_E(x))\dd\Hh^{d-1}(x).
\]
\end{proposition}

\begin{proof}
Choose open sets
\(\supp\zeta\Subset U\Subset V\Subset\Omega\) such that \(E\) has finite
perimeter in \(V\) and \(|D\one_E|(\partial U)=0\).  Strict local
polyhedral approximation gives sets \(R_j\), polyhedral in \(U\), with
\[
 \one_{R_j}\to\one_E\text{ in }L^1(U),\qquad
 D\one_{R_j}\stackrel{*}{\rightharpoonup}D\one_E,\qquad
 |D\one_{R_j}|(U)\to|D\one_E|(U)
\]
\citep{ambrosio1993polyhedral,ambrosio2000functions,maggi2012sets}.
Lemma~\ref{lem:cell-regularity} shows that the restriction of \(\Phi\) to
\(\overline U\times\Sph^{d-1}\) is bounded and upper semicontinuous.  Choose
bounded continuous functions \(\phi_k\) on this compact space such that
\(\phi_k\downarrow\Phi\) there, and define
\[
 \Phi_k(x,p):=
 \begin{cases}
  |p|\phi_k(x,p/|p|),&p\ne0,\\
  0,&p=0.
 \end{cases}
\]
Each \(\Phi_k\) is continuous on \(\overline U\times\R^d\), positively
one-homogeneous, and has linear growth.  Reshetnyak continuity for strict convergence of
vector measures gives, for every fixed \(k\),
\[
 \lim_{j\to\infty}
 \int_U\Phi_k\left(x,\frac{\dd D\one_{R_j}}{\dd|D\one_{R_j}|}\right)
 \dd|D\one_{R_j}|
 =
 \int_U\Phi_k\left(x,\frac{\dd D\one_E}{\dd|D\one_E|}\right)
 \dd|D\one_E|.
\]
Since \(\Phi\le\Phi_k\), first take the limsup in \(j\), then let
\(k\to\infty\).  Dominated convergence on the finite measure
\(|D\one_E|\) yields
\[
 \limsup_{j\to\infty}
 \int_U\Phi\left(x,\frac{\dd D\one_{R_j}}{\dd|D\one_{R_j}|}\right)
 \dd|D\one_{R_j}|
 \le
 \int_U\Phi\left(x,\frac{\dd D\one_E}{\dd|D\one_E|}\right)
 \dd|D\one_E|.
\]
Set \(P_j=R_j\) in \(U\) and \(P_j=E\) outside \(U\).  Since
\(\supp\zeta\Subset U\), the splice has positive distance from
\(\supp\zeta\).  It has no weighted surface cost, and finite interaction range
ensures that the later diffuse recovery cannot sample it for all sufficiently
small \(\eps\).  The asserted local convergence and inequality follow.
\end{proof}

\begin{lemma}[Diagonal selection]\label{lem:diagonal}
Suppose \(P_j\to E\) locally in measure and
\[
 \limsup_{j\to\infty}G_j\le G,
 \qquad
 G_j:=\int_{\partial P_j}\zeta\gamma(\cdot,\nu_{P_j})\dd\Hh^{d-1}.
\]
If for every fixed \(j\) there are \(u_\eps^j\to\one_{P_j}\) locally in
\(L^1\) with
\(\limsup_{\eps\downarrow0}\mathcal K_\eps^{J,\theta}(u_\eps^j;\zeta)\le G_j\), then
there is one family \(u_\eps\to\one_E\) locally in \(L^1\) satisfying
\[
 \limsup_{\eps\downarrow0}\mathcal K_\eps^{J,\theta}(u_\eps;\zeta)\le G.
\]
\end{lemma}

\begin{proof}
If \(G=+\infty\), take \(u_\eps=\one_E\).  Assume \(G<\infty\).  Let
\((K_m)\) be a compact exhaustion of \(\Omega\).  Pass to a subsequence, not
relabelled, for which every \(G_j\) is finite, \(G_j\le G+o(1)\), and
\[
 \sum_{m=1}^j
 \|\one_{P_j}-\one_E\|_{L^1(K_m)}\le j^{-1}.
\]
The last condition follows by a further diagonal extraction from the local
convergence \(P_j\to E\).  Choose \(\eps_j\downarrow0\)
so that, whenever \(0<\eps<\eps_j\),
\[
 \sum_{m=1}^j\|u_\eps^j-\one_{P_j}\|_{L^1(K_m)}\le j^{-1},
 \qquad
 \mathcal K_\eps^{J,\theta}(u_\eps^j;\zeta)\le G_j+j^{-1}.
\]
For \(\eps\in(\eps_{j+1},\eps_j]\), set \(u_\eps=u_\eps^j\).  The two
assertions follow from the displayed approximation rate, the triangle
inequality, and the bound on \(G_j\).
\end{proof}

\begin{proof}[Proof of Theorem~\ref{thm:main}]
Choose \(P_j\) from Proposition~\ref{prop:strict}.  Proposition
\ref{prop:poly} supplies a recovery family for each fixed \(P_j\), and
Lemma~\ref{lem:diagonal} gives the required single family.
\end{proof}

\section{A Matching Lower Bound under Pointwise Relative Lower Continuity}

We impose stronger assumptions under which the preceding upper bound has
a matching interior lower bound.  First, assume local uniform coercivity: for
every compact \(K\subset\R^d\), there are \(0<r_K<R_J\) and \(j_K>0\) such
that, outside a null set \(N_K\subset B_{r_K}\) independent of \(x\),
\begin{equation}
 J(x,Z)\ge j_K
 \quad\text{for every }x\in K\text{ and every }Z\in B_{r_K}\setminus N_K.
 \label{eq:uniform-core}
\end{equation}
Second, require only pointwise, one-sided continuity in relative error.  For
every base point \(y\in\R^d\), there are \(\rho_y>0\), a non-decreasing modulus
\(\omega_y:[0,\rho_y]\to[0,1)\), with
\(\omega_y(r)\downarrow0\) as \(r\downarrow0\), and a null set
\(N_y'\subset B_{R_J}\), such that
\begin{equation}
 J(x,Z)\ge
 \bigl(1-\omega_y(|x-y|)\bigr)J(y,Z)
 \label{eq:relative-continuity}
\end{equation}
for every \(x\in B_{\rho_y}(y)\) and every
\(Z\in B_{R_J}\setminus N_y'\).  The exceptional set may depend on the frozen
base point \(y\), but not on \(x\); this permits substitution of the
\(Z\)-dependent sampling point \(x-\theta\eps Z\).  No comparison between two
arbitrary nearby base points is assumed.
Finally, in this section assume that the two wells are the only zeroes:
\begin{equation}
 W^{-1}(0)=\{0,1\}.
 \label{eq:true-wells}
\end{equation}

For an open set \(A\Subset\Omega\), write
\begin{align}
 \mathcal K_{\eps,A}^{J,\theta}(u)
 &:=\frac1\eps\int_A W(u(x))\dd x \notag\\
 &\quad+\frac1{4\eps}\int_A\int_{\R^d}
 J(x-\theta\eps Z,Z)|u(x)-u(x-\eps Z)|^2\dd Z\dd x.
 \label{eq:local-energy}
\end{align}
\begin{lemma}[Compact-tail approximation]\label{lem:compact-tail}
Let \(K\in L^1(\R^d)\) be non-negative, even, and supported in \(B_R\), and
let \(n\in\Sph^{d-1}\).  Suppose \(q:\R\to[0,1]\) has finite profile energy
\[
 E_{K,n}(q):=\int_\R W(q(s))\dd s
 +\frac14\int_\R\int_{\R^d}K(Z)
 |q(s)-q(s-Z\cdot n)|^2\dd Z\dd s
\]
and satisfies \(q(s)\to0\) as \(s\to-\infty\) and
\(q(s)\to1\) as \(s\to+\infty\).  For \(L>R\), set
\[
 q_L(s):=
 \begin{cases}
  0,&s\le-L,\\
  q(s),&|s|<L,\\
  1,&s\ge L.
 \end{cases}
\]
Then \(q_L\in\mathcal A\) and \(E_{K,n}(q_L)\to E_{K,n}(q)\).
\end{lemma}

\begin{proof}
The potential term converges because
\(\int W(q_L)=\int_{-L}^LW(q)\).  Put \(a=Z\cdot n\) and let
\(A_L(a)\) be the set of \(s\) for which at least one of \(s,s-a\) lies
outside \((-L,L)\).  On the complement of \(A_L(a)\), the increments of
\(q_L\) and \(q\) agree.  Moreover,
\[
 \int_{\R^d}K(Z)\int_{A_L(Z\cdot n)}
 |q(s)-q(s-Z\cdot n)|^2\dd s\dd Z\longrightarrow0
\]
by dominated convergence, since the full interaction density of \(q\) is
integrable.

Define
\[
 \alpha_L:=\sup_{s\ge L-R}|1-q(s)|
       +\sup_{s\le-L+R}|q(s)|.
\]
Then \(\alpha_L\to0\).  If the increment of \(q_L\) is non-zero on
\(A_L(a)\), the segment with endpoints \(s,s-a\) crosses one of
\(\{-L,L\}\).  The set of such \(s\) has measure at most \(2|a|\), and the
increment is bounded by \(\alpha_L\).  Consequently
\[
 \int_{\R^d}K(Z)\int_{A_L(Z\cdot n)}
 |q_L(s)-q_L(s-Z\cdot n)|^2\dd s\dd Z
 \le2\alpha_L^2\int_{\R^d}K(Z)|Z\cdot n|\dd Z\longrightarrow0.
\]
The interaction energies therefore converge as claimed.
\end{proof}

\begin{lemma}[Frozen surface tension]\label{lem:frozen-tension}
Under \eqref{eq:uniform-core} and \eqref{eq:true-wells}, for every fixed
\(x_0\in\Omega\), the density \(n\mapsto\gamma(x_0,n)\) is the surface tension
of the homogeneous Kac functional with kernel \(J(x_0,\cdot)\).  Its
positively one-homogeneous extension is finite and convex; in particular,
\(\gamma(x_0,\cdot)\) is continuous on \(\Sph^{d-1}\).
\end{lemma}

\begin{proof}
We first record the normalization needed to apply the homogeneous results.
Set \(v=2u-1\), \(\widetilde J=J(x_0,\cdot)/4\), and define
\[
 \widetilde W(v):=
 \begin{cases}
  W((v+1)/2),&|v|\le1,\\
  |v|-1,&|v|>1.
 \end{cases}
\]
This extension is continuous, has zeroes only at \(\{-1,1\}\), and has linear
growth.  Moreover, truncation to \([-1,1]\) does not increase either the
potential or the pair-interaction term.  Thus the unrestricted homogeneous
problem has the same infimum as its \([-1,1]\)-valued restriction.  On
functions valued in \([0,1]\),
the homogeneous energy for \((J(x_0,\cdot),W)\) is then exactly the energy for
\((\widetilde J,\widetilde W)\) after the affine change of variables, since
\[
 \frac14\widetilde J(Z)|v(x)-v(x-\eps Z)|^2
 =\frac14J(x_0,Z)|u(x)-u(x-\eps Z)|^2.
\]
Both sides use the ordered pair integral; restricting both endpoints to a set
and changing variables from the second endpoint to \(Z\) preserves the factor
\(1/4\).  The
kernel \(\widetilde J\) is non-negative, even, integrable, and has finite first
moment, so the hypotheses of the homogeneous theorem are satisfied.

The homogeneous optimal-profile theorem identifies the multidimensional
planar cell problem with its one-dimensional restriction
\citep{alberti1998optimal}.  That result is stated with profiles converging to
the wells at infinity.  Lemma~\ref{lem:compact-tail}, applied after the affine
phase change, shows that the infimum is unchanged when profiles are required
to equal the pure phases outside a compact interval.  Thus its cell value
agrees with \eqref{eq:cell}.  The homogeneous \(\Gamma\)-limit theorem identifies this
surface tension as a lower-semicontinuous anisotropic perimeter density
\citep{alberti1998nonlocal}.  Its one-homogeneous extension is consequently
convex.  Finiteness follows from the step-profile bound, and finite convex
one-homogeneous functions are continuous.
\end{proof}

\begin{theorem}[Interior compactness and liminf]\label{thm:liminf}
Assume \eqref{eq:uniform-core}, \eqref{eq:relative-continuity}, and
\eqref{eq:true-wells}.  Let
\(u_\eps:\R^d\to[0,1]\) satisfy
\begin{equation}
 \sup_{\eps>0}\mathcal K_{\eps,A}^{J,\theta}(u_\eps)<\infty
 \qquad\text{for every }A\Subset\Omega.
 \label{eq:local-bound}
\end{equation}
Then a subsequence converges in \(L^1_{\rm loc}(\Omega)\) to \(\one_E\) for
some set \(E\) of locally finite perimeter.  Moreover, whenever the whole
sequence converges to \(\one_E\), every \(\zeta\in C_c(\Omega)\),
\(\zeta\ge0\), satisfies
\begin{equation}
 \liminf_{\eps\downarrow0}
 \mathcal K_\eps^{J,\theta}(u_\eps;\zeta)
 \ge
 \int_{\partial^*E}\zeta(x)\gamma(x,\nu_E(x))\dd\Hh^{d-1}(x).
 \label{eq:liminf}
\end{equation}
Together with Theorem~\ref{thm:main}, this gives matching interior
sharp-interface bounds for this intrinsically coercive class.
\end{theorem}

\begin{proof}
\emph{Step 1: compactness.}
Fix \(A'\Subset A\Subset\Omega\), and choose a compact \(K\) containing a
neighbourhood of \(\overline A\).  For all sufficiently small \(\eps\), the
coefficient points associated with first endpoints in \(A\) belong to \(K\).
Set
\[
 J_{{\rm c},K}(Z):=j_K\one_{B_{r_K}}(Z).
\]
Condition \eqref{eq:uniform-core} gives
\(J(x-\theta\eps Z,Z)\ge J_{{\rm c},K}(Z)\) for \(x\in A\).  Hence
\(\mathcal K_{\eps,A}^{J,\theta}\) controls the internal homogeneous energy
\[
 \frac1\eps\int_A W(u_\eps)\dd x
 +\frac{j_K}{4\eps}\int_A
   \int_{\substack{Z\in B_{r_K}\\x-\eps Z\in A}}
   |u_\eps(x)-u_\eps(x-\eps Z)|^2\dd Z\dd x.
\]
Indeed, one restricts the interaction to \(Z\in B_{r_K}\) and to pairs with
both endpoints in \(A\), and then discards all remaining non-negative terms.
The kernel \(J_{{\rm c},K}\) and
the true-well condition \eqref{eq:true-wells} satisfy the hypotheses of the
homogeneous compactness theorem on \(A'\), after the affine normalization used
in Lemma~\ref{lem:frozen-tension} \citep{alberti1998nonlocal}.  Since
\(A'\Subset A\), this is an interior application and produces no artificial
boundary contribution.  A diagonal
extraction on an exhaustion of \(\Omega\) gives
\(u_\eps\to\one_E\) in \(L^1_{\rm loc}\), with \(E\) of locally finite
perimeter.

\emph{Step 2: energy measures.}
It suffices to prove \eqref{eq:liminf} along a subsequence realizing its left
side.  On a relatively compact open set containing \(\supp\zeta\), define
positive Radon measures
\begin{align*}
 \mu_\eps(B)
 &:=\frac1\eps\int_B W(u_\eps(x))\dd x\\
 &\quad+\frac1{4\eps}\int_B\int_{\R^d}
 J(x-\theta\eps Z,Z)|u_\eps(x)-u_\eps(x-\eps Z)|^2\dd Z\dd x.
\end{align*}
The local bound gives, after a further extraction,
\(\mu_\eps\stackrel{*}{\rightharpoonup}\mu\).  We prove
\begin{equation}
 \mu\ge
 \gamma(x,\nu_E(x))\Hh^{d-1}\mathbin{\vrule height 1.4ex depth -0.4ex
 width 0.07ex\vrule height 0.07ex depth -0.02ex width 0.7ex}\partial^*E.
 \label{eq:measure-domination}
\end{equation}

\emph{Step 3: multiplicative freezing.}
Fix \(x_0\in\partial^*E\), and let \(\rho_{x_0}\) and \(\omega_{x_0}\) be
supplied by \eqref{eq:relative-continuity}.  If
\(x\in B_r(x_0)\) and \(r+R_J\eps\le\rho_{x_0}\), then
\begin{align*}
 J(x-\theta\eps Z,Z)
 &\ge \bigl(1-\delta_{r,\eps}\bigr)J(x_0,Z),
 \qquad
 \delta_{r,\eps}:=\omega_{x_0}(r+R_J\eps).
\end{align*}
Here we used
\(|x-\theta\eps Z-x_0|\le r+R_J\eps\).  Define the internal frozen energy
\begin{align}
 \mathcal K_\eps^{x_0}(u;B_r(x_0))
 &:=\frac1\eps\int_{B_r(x_0)}W(u(x))\dd x \notag\\
 &\quad+\frac1{4\eps}\int_{B_r(x_0)}
 \int_{\{Z:\,x-\eps Z\in B_r(x_0)\}}J(x_0,Z)
 |u(x)-u(x-\eps Z)|^2\dd Z\dd x.
 \label{eq:frozen-internal}
\end{align}
Then
\begin{equation}
 \mu_\eps(B_r(x_0))
 \ge(1-\delta_{r,\eps})
 \mathcal K_\eps^{x_0}(u_\eps;B_r(x_0)),
 \label{eq:freeze}
\end{equation}
Indeed, the variable energy first controls the frozen energy with all second
endpoints, and dropping the cross-boundary interactions gives exactly
\eqref{eq:frozen-internal} and can only lower that energy.

\emph{Step 4: frozen liminf and differentiation.}
Choose radii \(r\downarrow0\) for which both \(\mu\) and
\(|D\one_E|\) give zero mass to \(\partial B_r(x_0)\), and restrict to
\(r<\rho_{x_0}/2\).  For each fixed such \(r\), one has
\(\delta_{r,\eps}\le\omega_{x_0}(2r)\) for all sufficiently small \(\eps\).
Passing \(\eps\downarrow0\) in \eqref{eq:freeze} and applying the homogeneous
frozen liminf theorem on the ball gives
\begin{equation}
 \mu(B_r(x_0))
 \ge
 \bigl(1-\omega_{x_0}(2r)\bigr)
 \int_{\partial^*E\cap B_r(x_0)}
 \gamma(x_0,\nu_E(y))\dd\Hh^{d-1}(y).
 \label{eq:ball-density}
\end{equation}
Indeed, by evenness of \(J(x_0,\cdot)\),
\eqref{eq:frozen-internal} is exactly the homogeneous pair-interaction
functional with both endpoints in \(B_r(x_0)\), written in the displacement
variable \(Z\).  The ball is a regular open set, and the affine normalization
in Lemma~\ref{lem:frozen-tension} puts the functional in the form of the
classical lower-bound theorem.
At \(|D\one_E|\)-almost every reduced-boundary point, \(x_0\) is a
Lebesgue point of the polar field \(\nu_E\) relative to \(|D\one_E|\).
Lemma~\ref{lem:frozen-tension} gives continuity of
\(n\mapsto\gamma(x_0,n)\).  Divide \eqref{eq:ball-density} by
\(|D\one_E|(B_r(x_0))\) and let \(r\downarrow0\).  The differentiation
theorem for Radon measures yields
\[
 \frac{\dd\mu}{\dd|D\one_E|}(x_0)
 \ge\gamma(x_0,\nu_E(x_0))
 \quad\text{for }|D\one_E|\text{-a.e. }x_0,
\]
which is \eqref{eq:measure-domination}.

Finally, weak-star convergence of \(\mu_\eps\) and continuity of \(\zeta\)
give
\[
 \lim_{\eps\downarrow0}\mathcal K_\eps^{J,\theta}(u_\eps;\zeta)
 =\int\zeta\dd\mu
 \ge\int_{\partial^*E}\zeta(x)\gamma(x,\nu_E(x))\dd\Hh^{d-1}(x),
\]
along the selected subsequence, proving the liminf inequality.
\end{proof}

\begin{corollary}[Pointwise Gamma limit at characteristic states]
\label{cor:pointwise-gamma}
Under the hypotheses of Theorem~\ref{thm:liminf}, every
\(u_\eps\to\one_E\) in \(L^1_{\rm loc}(\Omega)\) satisfies
\[
 \liminf_{\eps\downarrow0}\mathcal K_\eps^{J,\theta}(u_\eps;\zeta)
 \ge \int_{\partial^*E}\zeta(x)\gamma(x,\nu_E(x))
 \dd\Hh^{d-1}(x).
\]
Together with Theorem~\ref{thm:main}, this identifies the pointwise
\(\Gamma\)-limit at every finite-perimeter characteristic state.
\end{corollary}
\section{Heterogeneous Ray and Flow Kac Energies}

Let
\[
 d_*(r):=\min\{r,1-r\}\qquad(0\le r\le1),
\]
and let \(\chi(r)=\one_{\{r\ge1/2\}}\).  The elementary estimate
\begin{equation}
 |\chi(a)-\chi(b)|^2-|a-b|^2
 \le2\bigl(d_*(a)+d_*(b)\bigr)
 \qquad(a,b\in[0,1])
 \label{eq:diffuse-rounding-pointwise}
\end{equation}
follows directly when \(a,b\) lie on the same side of \(1/2\); on opposite
sides use \(|a-b|\ge1-d_*(a)-d_*(b)\).  This estimate makes threshold
rounding energy-decreasing for the steep potentials used below.
\subsection{Radially spread ray kernels}
\label{subsec:global-ray-realization}

A ray kernel is a positive Radon measure concentrated on finitely many lines
and spread over a non-degenerate interval of radii.  Its angular singularity
creates corners in the cell density, while radial spreading prevents the
exact-shift resonance of finitely many point masses and supplies compactness.

Let \(0<a<b<\infty\), and let
\begin{equation}
 \lambda(\dd t)=w(t)\dd t,
 \qquad 0<w_*\le w(t)\le w^*<\infty\quad\text{a.e. on }(a,b),
 \qquad \int_a^b t\,\lambda(\dd t)=1.
 \label{eq:ray-radial-law}
\end{equation}
Given vectors \(v_1,\ldots,v_N\in\R^d\setminus\{0\}\) and coefficients
\(c_i>0\), define the even, finite, first-moment interaction measure
\begin{equation}
 \nu:=2\sum_{i=1}^N c_i\int_a^b
       \bigl(\delta_{t v_i}+\delta_{-t v_i}\bigr)\lambda(\dd t),
 \qquad
 m_\nu:=\nu(\R^d)=4\lambda((a,b))\sum_{i=1}^N c_i.
 \label{eq:ray-measure}
\end{equation}
Its support is compact and avoids the origin.  The normalization by the first
radial moment, rather than by the total mass, is what fixes the surface
tension.

On the flat torus \(\mathbb T^d\), translations below are understood modulo
the period.  For \(M\ge m_\nu\), put
\begin{equation}
 \mathcal F_\eps^\nu(u)
 :=\frac{M}{\eps}\int_{\mathbb T^d}d_*(u)\dd x
 +\frac1{4\eps}\int_{\mathbb T^d}\int_{\R^d}
      |u(x)-u(x-\eps Z)|^2\nu(\dd Z)\dd x,
 \label{eq:ray-diffuse-energy}
\end{equation}
for measurable \(u:\mathbb T^d\to[0,1]\), and set it to \(+\infty\)
otherwise.  The periodic topology removes boundary cells; the corresponding
local statement on \(U\Subset\Omega\) follows by inserting a cut-off equal to
one near \(\overline U\).  No assertion about a physical boundary cell is
contained in the theorem.

\begin{lemma}[Measure-kernel cell formula and normalization]
\label{lem:ray-measure-cell-formula}
Let \(\mu\) be a finite, non-negative Borel measure on \(\R^d\) with compact
support and finite first moment, and let \(M\ge\mu(\R^d)\).  For the profile
problem obtained from \eqref{eq:profile-cost} by replacing
\(J(Z)\dd Z\) with \(\mu(\dd Z)\), the steep potential \(W_M=M d_*\) gives
the exact one-homogeneous cell density
\begin{equation}
 h_\mu(\xi)=\frac14\int_{\R^d}|Z\cdot\xi|\mu(\dd Z).
 \label{eq:measure-kernel-cell}
\end{equation}
In particular, \eqref{eq:ray-radial-law}--\eqref{eq:ray-measure} give
\begin{equation}
 h_\nu(\xi)=\sum_{i=1}^N c_i|v_i\cdot\xi|.
 \label{eq:zonotope-density}
\end{equation}
This is a norm if and only if \(\operatorname{span}\{v_1,\ldots,v_N\}=\R^d\).
\end{lemma}

\begin{proof}
The threshold map \(\chi(r)=\one_{\{r\ge1/2\}}\) satisfies
\eqref{eq:diffuse-rounding-pointwise}.  Integrating that inequality first in
the profile variable and then against \(\mu\) gives
\[
 e_\mu(n;\chi\circ q)-e_\mu(n;q)
 \le\bigl(\mu(\R^d)-M\bigr)
       \int_\R d_*(q(r))\dd r\le0.
\]
Thus it is enough to minimize over binary profiles.  For every
\(v:\R\to\{0,1\}\) with the prescribed end states and every \(s\in\R\),
truncate on \((-R,R)\), count the net transition, and let \(R\to\infty\) to
obtain
\[
 \int_\R|v(r)-v(r-s)|\dd r\ge |s|.
\]
The monotone step attains equality for every \(s\) simultaneously.  Since
binary squared differences equal absolute differences, integration against
\(\mu\) proves \eqref{eq:measure-kernel-cell}.  Substituting
\eqref{eq:ray-measure} and using \(\int t\dd\lambda=1\) yields
\[
 \frac14\,2\sum_i c_i\int_a^b
 \bigl(|t v_i\cdot\xi|+|-t v_i\cdot\xi|\bigr)\lambda(\dd t)
 =\sum_i c_i|v_i\cdot\xi|.
\]
The last expression vanishes precisely on the orthogonal complement of the
span of the \(v_i\), which proves the norm criterion.
\end{proof}

\begin{theorem}[Gamma-limit for radially spread singular ray kernels]
\label{thm:ray-measure-gamma-limit}
Assume \eqref{eq:ray-radial-law}, \(M\ge m_\nu\), and
\(\operatorname{span}\{v_1,\ldots,v_N\}=\R^d\).  In the strong
\(L^1(\mathbb T^d)\) topology, the functionals
\(\mathcal F_\eps^\nu\) Gamma-converge as \(\eps\downarrow0\) to
\begin{equation}
 \mathcal F_0^\nu(u):=
 \begin{cases}
 \displaystyle\int_{\partial^*E}h_\nu(\nu_E)\dd\Hh^{d-1},
   &u=\one_E,\quad E\text{ of finite perimeter in }\mathbb T^d,\\[4pt]
 +\infty,&\text{otherwise}.
 \end{cases}
 \label{eq:ray-gamma-limit}
\end{equation}
The family is equicoercive in \(L^1\).  More precisely:
\begin{enumerate}
 \item if \(u_\eps\to u\) in \(L^1\), then
 \(\liminf_{\eps\downarrow0}\mathcal F_\eps^\nu(u_\eps)
 \ge\mathcal F_0^\nu(u)\);
 \item for every finite-perimeter set \(E\), the constant binary sequence
 \(u_\eps=\one_E\) is a recovery sequence.
\end{enumerate}
The lower bound and recovery remain valid for any finite radial measure with
unit first moment.  The density bounds in \eqref{eq:ray-radial-law} are used
only for equicoercivity.
\end{theorem}

\begin{proof}
We separate rounding, the directional lower bound, recovery, and compactness.
For a periodic field \(u\), translation invariance and
\eqref{eq:diffuse-rounding-pointwise} give
\begin{align}
 \mathcal F_\eps^\nu(\chi\circ u)-\mathcal F_\eps^\nu(u)
 &\le \frac{m_\nu-M}{\eps}
       \int_{\mathbb T^d}d_*(u)\dd x\le0.
 \label{eq:ray-global-rounding}
\end{align}
Moreover
\begin{equation}
 \|u-\chi\circ u\|_{L^1}=\int_{\mathbb T^d}d_*(u)\dd x
 \le\frac\eps M\mathcal F_\eps^\nu(u).
 \label{eq:ray-rounding-distance}
\end{equation}
Hence any bounded-energy sequence may be replaced by a binary sequence
\(z_\eps\) without increasing its energy or changing its \(L^1\) limit.
For binary fields, \eqref{eq:ray-measure} reduces the energy exactly to
\begin{equation}
 \mathcal F_\eps^\nu(z)
 =\sum_{i=1}^N\frac{c_i}{\eps}
   \int_a^b\int_{\mathbb T^d}
   |z(x)-z(x-\eps t v_i)|\dd x\lambda(\dd t).
 \label{eq:ray-binary-directional-energy}
\end{equation}

Suppose first that \(z_\eps\to z\) in \(L^1\).  For fixed \(t>0\), the
difference quotients converge distributionally:
\[
 \frac{z_\eps(\,cdot+\eps t v_i)-z_\eps}{\eps t}
 \longrightarrow D_{v_i}z.
\]
Indeed, testing against \(\phi\in C^1(\mathbb T^d)\), changing variables,
and using uniform convergence of the corresponding difference quotient of
\(\phi\) proves the assertion.  Lower semicontinuity of total variation then
gives
\begin{equation}
 |D_{v_i}z|(\mathbb T^d)
 \le\liminf_{\eps\downarrow0}\frac1{\eps t}
 \int_{\mathbb T^d}|z_\eps(x)-z_\eps(x-\eps t v_i)|\dd x.
 \label{eq:ray-directional-liminf}
\end{equation}
Fatou's lemma, followed by \(\int t\dd\lambda=1\), yields from
\eqref{eq:ray-binary-directional-energy}
\[
 \liminf_{\eps\downarrow0}\mathcal F_\eps^\nu(z_\eps)
 \ge\sum_i c_i|D_{v_i}z|(\mathbb T^d).
\]
If the left side is finite, \eqref{eq:ray-rounding-distance} makes \(z\)
binary.  Since the \(v_i\) span \(\R^d\), finite variation in these
directions implies \(z\in BV\), and the polar decomposition of \(Dz\) gives
\[
 \sum_i c_i|D_{v_i}z|
 =\int_{\mathbb T^d}\sum_i c_i|v_i\cdot\sigma_z|\dd|Dz|
 =\int_{\partial^*E}h_\nu(\nu_E)\dd\Hh^{d-1}
\]
when \(z=\one_E\).  This proves the Gamma-liminf.

For recovery, let \(z=\one_E\in BV(\mathbb T^d;\{0,1\})\).  The standard
translation characterization of directional variation gives, for every
fixed \(t\in(a,b)\),
\begin{equation}
 \frac1\eps\int_{\mathbb T^d}
 |z(x)-z(x-\eps t v_i)|\dd x
 \longrightarrow t|D_{v_i}z|(\mathbb T^d).
 \label{eq:ray-translation-limit}
\end{equation}
The BV translation estimate bounds the left side by
\(t|D_{v_i}z|(\mathbb T^d)\).  Dominated convergence in \(t\), the unit
first moment of \(\lambda\), and \eqref{eq:ray-binary-directional-energy}
therefore prove
\(\mathcal F_\eps^\nu(\one_E)\to\mathcal F_0^\nu(\one_E)\).

It remains to justify equicoercivity, because averaging only at finitely many
single shifts would not suffice.  Choose a basis
\(v_{i_1},\ldots,v_{i_d}\) from the given directions.  For a binary bounded
energy sequence \(z_\eps\), \eqref{eq:ray-radial-law} and
\eqref{eq:ray-binary-directional-energy} imply
\begin{equation}
 \int_a^b\frac1\eps
 \|z_\eps(\cdot-\eps t v_{i_j})-z_\eps\|_{L^1}\dd t\le C
 \quad(j=1,\ldots,d).
 \label{eq:ray-averaged-translation-bound}
\end{equation}
Divide \((a,b)\) into three equal subintervals.  In the first and third one,
select \(a_{j,\eps}<b_{j,\eps}\), respectively, so that each selected shift
has \(L^1\) difference at most \(C\eps\); their separation is bounded below
by \((b-a)/3\).  Let \(S_{j,\eps}\) average translations along \(v_{i_j}\)
over \([\eps a_{j,\eps},\eps b_{j,\eps}]\).  Then
\begin{align}
 \|S_{j,\eps}z_\eps-z_\eps\|_{L^1}&\le C\eps,
 \label{eq:ray-smoothing-distance}\\
 |D_{v_{i_j}}S_{j,\eps}z_\eps|(\mathbb T^d)
 &\le\frac{
 \|z_\eps(\cdot-\eps b_{j,\eps}v_{i_j})
      -z_\eps(\cdot-\eps a_{j,\eps}v_{i_j})\|_{L^1}}
 {\eps(b_{j,\eps}-a_{j,\eps})}\le C.
 \label{eq:ray-smoothing-variation}
\end{align}
The numerator in the last line is bounded by the two selected translation
differences.  The averaging operators commute and are \(L^1\) contractions.
Thus \(Z_\eps:=S_{1,\eps}\cdots S_{d,\eps}z_\eps\) is uniformly bounded in
\(BV\), while \(\|Z_\eps-z_\eps\|_1\to0\).  BV compactness and
\eqref{eq:ray-rounding-distance} prove equicoercivity for the original
sequence.
\end{proof}

\subsection{Spatially heterogeneous ray fields}
\label{subsec:heterogeneous-ray-fields}

We next allow the generating segments of the zonotope to vary in space.  The
definition of the bond measure matters.  Evaluating a coefficient at only one
endpoint produces a directed measure, although the quadratic integrand itself
is symmetric.  We therefore record the associated symmetric pair measure
explicitly.  This also fixes all factors in the steep-potential estimate.

\paragraph{Route through this subsection.}
We first keep the ray directions fixed, symmetrize the endpoint weights, and
prove the weighted translation limit.  We then replace spatially varying
directions by their \(C^{1,1}\) flows, define the symmetric flow-pair measure,
and derive local coordinates from ordered flow compositions.  The ordered-flow
smoothing lemma supplies equicoercivity; the final theorem combines that
compactness with flow-translation lower bounds and constant binary recovery.
Thus the progression is from weighted translations to flow pairs, then to
ordered smoothing and the flow-ray Gamma-limit.

Put \(\Lambda:=\lambda((a,b))\).  First keep
\(v_1,\ldots,v_N\) fixed and let
\begin{equation}
 c_i\in W^{1,\infty}(\mathbb T^d),\qquad
 0<c_*\le c_i(x)\le c^*<\infty.
 \label{eq:heterogeneous-fixed-ray-assumptions}
\end{equation}
For \(h=\eps t v_i\), define the endpoint-symmetric coefficient
\begin{equation}
 a_{i,\eps,t}(x):=\frac12\{c_i(x)+c_i(x-h)\}
 \label{eq:heterogeneous-endpoint-symmetrization}
\end{equation}
and the energy
\begin{equation}
 \mathcal F_\eps^{c,v}(u):=\frac M\eps\int_{\mathbb T^d}d_*(u)\dd x
 +\sum_{i=1}^N\frac1\eps\int_a^b\int_{\mathbb T^d}
 a_{i,\eps,t}(x)|u(x)-u(x-\eps t v_i)|^2\dd x\lambda(\dd t),
 \label{eq:heterogeneous-fixed-ray-energy}
\end{equation}
with value \(+\infty\) outside \([0,1]\), where
\begin{equation}
 M\ge4\Lambda Nc^*.
 \label{eq:heterogeneous-fixed-steepness}
\end{equation}
More explicitly, the pair measure used in the second term is
\begin{equation}
 P_{i,\eps,t}:=\frac12a_{i,\eps,t}(x)\dd x\,
 \delta_{x-\eps t v_i}(\dd y)
 +\frac12\left(a_{i,\eps,t}(x)\dd x\,
 \delta_{x-\eps t v_i}(\dd y)\right)^{\mathsf T}.
 \label{eq:heterogeneous-fixed-pair-measure}
\end{equation}
It is symmetric and its integral against a symmetric bond integrand is exactly
the term in \eqref{eq:heterogeneous-fixed-ray-energy}.  Thus neither endpoint
is privileged.

The candidate density is
\begin{equation}
 h_{c,v}(x,\xi):=\sum_{i=1}^Nc_i(x)|v_i\cdot\xi|.
 \label{eq:heterogeneous-zonotopal-density}
\end{equation}
If the fixed directions span \(\R^d\), compactness of the sphere gives
\begin{equation}
 \kappa|\xi|\le h_{c,v}(x,\xi)\le K|\xi|
 \quad(x\in\mathbb T^d,\ \xi\in\R^d)
 \label{eq:heterogeneous-uniform-coercivity}
\end{equation}
for constants \(0<\kappa\le K\).

\begin{theorem}[Heterogeneous weights on fixed ray directions]
\label{thm:heterogeneous-fixed-ray-gamma}
Assume \eqref{eq:ray-radial-law},
\eqref{eq:heterogeneous-fixed-ray-assumptions}, and
\(\operatorname{span}\{v_1,\ldots,v_N\}=\R^d\).  Then
\(\mathcal F_\eps^{c,v}\) is equicoercive and Gamma-converges in strong
\(L^1(\mathbb T^d)\) to
\begin{equation}
 \mathcal F_0^{c,v}(u):=
 \begin{cases}
 \displaystyle\int_{\partial^*E}h_{c,v}(x,\nu_E(x))\dd\Hh^{d-1}(x),
 &u=\one_E,\ E\text{ of finite perimeter},\\[4pt]
 +\infty,&\text{otherwise}.
 \end{cases}
 \label{eq:heterogeneous-fixed-ray-limit}
\end{equation}
Every finite-perimeter set has the constant binary recovery sequence
\(u_\eps=\one_E\).  In particular, no polyhedral approximation is needed.
\end{theorem}

\begin{proof}
The pointwise rounding estimate \eqref{eq:diffuse-rounding-pointwise} and the
symmetry of every bond give
\begin{equation}
 \mathcal F_\eps^{c,v}(\chi\circ u)-\mathcal F_\eps^{c,v}(u)
 \le\frac{4\Lambda Nc^*-M}{\eps}\int d_*(u)\dd x\le0.
 \label{eq:heterogeneous-fixed-rounding}
\end{equation}
Also \(\|u-\chi\circ u\|_1\le(\eps/M)\mathcal F_\eps^{c,v}(u)\).
It is therefore enough to consider binary fields \(z_\eps\).

Fix \(i\) and \(t\).  The measures
\begin{equation}
 \mu_{i,\eps,t}:=
 \frac{z_\eps-z_\eps(\,\cdot-\eps t v_i)}{\eps t}\,\mathcal L^d
 \stackrel{*}{\rightharpoonup}D_{v_i}z
 \label{eq:heterogeneous-directional-measure-convergence}
\end{equation}
distributionally whenever \(z_\eps\to z\) in \(L^1\).  Since
\(a_{i,\eps,t}\to c_i\) uniformly, weighted lower semicontinuity of total
variation yields
\begin{equation}
 \int c_i\dd|D_{v_i}z|
 \le\liminf_{\eps\downarrow0}\int a_{i,\eps,t}\dd|\mu_{i,\eps,t}|.
 \label{eq:heterogeneous-weighted-directional-liminf}
\end{equation}
Fatou in \(t\), followed by \(\int t\dd\lambda=1\), proves the liminf.
The lower bound \(a_{i,\eps,t}\ge c_*\), applied to a fixed basis among the
\(v_i\), reduces equicoercivity to the smoothing argument in
\eqref{eq:ray-averaged-translation-bound}--\eqref{eq:ray-smoothing-variation}.

For \(z=\one_E\), the weighted translation formula for BV functions gives
\begin{equation}
 \frac1\eps\int c_i(x)|z(x)-z(x-\eps t v_i)|\dd x
 \longrightarrow t\int c_i\dd|D_{v_i}z|.
 \label{eq:heterogeneous-weighted-translation-limit}
\end{equation}
The BV translation estimate and the bounds on \(c_i\) provide an integrable
majorant in \(t\).  Moreover endpoint freezing is quantitative:
\begin{align}
 &\left|\frac1\eps\int
 [a_{i,\eps,t}(x)-c_i(x)]|z(x)-z(x-\eps t v_i)|\dd x\right|
 \notag\\
 &\hspace{25mm}\le
 \frac\eps2\|\nabla c_i\|_\infty t^2|v_i|\,|D_{v_i}z|(\mathbb T^d).
 \label{eq:heterogeneous-endpoint-freezing-error}
\end{align}
Dominated convergence proves recovery and completes the proof.
\end{proof}

We next allow the ray directions to rotate without imposing an integrable global
frame.
Directly writing
\(x-\eps t v_i(x)\) obscures reversibility and may even destroy injectivity.
The intrinsic replacement is the flow.  Let
\begin{equation}
 c_i\in W^{1,\infty}(\mathbb T^d),\quad 0<c_*\le c_i\le c^*,
 \qquad v_i\in W^{2,\infty}(\mathbb T^d;\R^d),
 \label{eq:flow-ray-regularity}
\end{equation}
 and denote by \(\Phi_i^s\) the bi-Lipschitz flow of \(v_i\).  No commutation
or common coordinate representation of these flows is assumed.  Assume the
uniform first-order spanning condition
\begin{equation}
 h_{\rm fl}(x,\xi):=\sum_{i=1}^Nc_i(x)|v_i(x)\cdot\xi|
 \ge\kappa|\xi|.
 \label{eq:flow-ray-spanning}
\end{equation}
For each \(\eps,t,i\), set
\begin{align}
 Q_{i,\eps,t}(\dd x,\dd y)
 &:=c_i(x)\dd x\,\delta_{\Phi_i^{-\eps t}(x)}(\dd y),
 &\widehat Q_{i,\eps,t}&:=\tfrac12(Q_{i,\eps,t}+Q_{i,\eps,t}^{\mathsf T}).
 \label{eq:flow-ray-pair-measure}
\end{align}
The flow-ray energy is
\begin{align}
 \mathcal G_\eps^{c,v}(u)
 &:=\frac M\eps\int d_*(u)\dd x
 +\frac1\eps\sum_i\int_a^b\iint|u(x)-u(y)|^2
 \widehat Q_{i,\eps,t}(\dd x,\dd y)\lambda(\dd t)
 \notag\\
 &=\frac M\eps\int d_*(u)\dd x
 +\frac1\eps\sum_i\int_a^b\int c_i(x)
 |u(x)-u(\Phi_i^{-\eps t}(x))|^2\dd x\lambda(\dd t).
 \label{eq:flow-ray-energy}
\end{align}
Fix \(\eps_0>0\), write \(D_i=\|\operatorname{div}v_i\|_\infty\), and take
\begin{equation}
 M\ge2\Lambda\sum_i c^*(1+e^{\eps_0bD_i}).
 \label{eq:flow-ray-steepness}
\end{equation}

 \begin{lemma}[Quantitative local coordinates from ordered flows]
\label{lem:ordered-flow-local-coordinates}
There are \(r_0,\delta_0,C_0>0\), depending only on \(d,N,\kappa,c^*\)
and the norms in \eqref{eq:flow-ray-regularity}, with the following property.
For every \(x_0\) one can select an ordered \(d\)-tuple
\(I=(i_1,\ldots,i_d)\) and define
\begin{equation}
 \Theta_{x_0,I}(s)
 :=\Phi_{i_d}^{s_d}\circ\cdots\circ\Phi_{i_1}^{s_1}(x_0),
 \qquad s\in Q_{r_0}:=(-r_0,r_0)^d,
 \label{eq:ordered-flow-chart}
\end{equation}
so that \(\Theta_{x_0,I}\) is bi-Lipschitz onto its image.  Writing
\(J_{x_0,I}:=|\det D_s\Theta_{x_0,I}|\), one has
\begin{align}
 \delta_0|s-s'|&\le
 |\Theta_{x_0,I}(s)-\Theta_{x_0,I}(s')|
 \le C_0|s-s'|,                                      \label{eq:flow-chart-bilip}\\
 \Theta_{x_0,I}(s)&=x_0+\sum_{k=1}^ds_kv_{i_k}(x_0)+R_{x_0,I}(s),
 \notag\\[-2pt]
 |R(s)-R(s')|&\le C_0r_0|s-s'|,                    \label{eq:flow-chart-linearization}\\
 \bigl|J_{x_0,I}(s)-|\det V_I(x_0)|\bigr|&\le C_0|s|,
 \qquad 0<\delta_0\le J_{x_0,I}\le C_0 .          \label{eq:flow-chart-jacobian}
\end{align}
Moreover, interchanging two adjacent flows changes the chart by at most
\begin{equation}
 \left|\Phi_j^t\Phi_i^s(x)-\Phi_i^s\Phi_j^t(x)\right|
 \le C_0|st|,
 \qquad |s|+|t|\le r_0.
 \label{eq:flow-chart-commutator-error}
\end{equation}
The torus admits a finite cover by smaller images of such charts, with a
Lipschitz partition of unity subordinate to the cover.
\end{lemma}

\begin{proof}
Put \(A(x)\xi=(\sqrt{c_i(x)}v_i(x)\cdot\xi)_{i=1}^N\).
Condition \eqref{eq:flow-ray-spanning} and \(c_i\le c^*\) imply
\begin{equation}
 \|A(x)\xi\|_{\ell^2}
 \ge \kappa_0|\xi|,
 \qquad \kappa_0:=\frac{\kappa}{\sqrt{Nc^*}},
 \label{eq:ordered-frame-singular-value}
\end{equation}
because \(\sum_i c_i|v_i\cdot\xi|
\le\sqrt{Nc^*}\,\|A(x)\xi\|_{\ell^2}\).
Consequently \(\det(A^{\mathsf T}A)\ge\kappa_0^{2d}\).  Cauchy--Binet and
\(c_i\le c^*\) then give, at every \(x\), a minor
\(V_I(x)=(v_{i_1}(x),\ldots,v_{i_d}(x))\) satisfying
\begin{equation}
 |\det V_I(x)|\ge
 \delta_1:=\frac{\kappa_0^d}
 {\sqrt{\binom Nd}\,(c^*)^{d/2}}>0.
 \label{eq:ordered-frame-minor-bound}
\end{equation}
Since the fields are
Lipschitz, Gronwall's inequality gives
\begin{equation}
 |\Phi_i^s(x)-x-sv_i(x)|\le C s^2,
 \qquad \operatorname{Lip}(\Phi_i^s-\mathrm{Id})\le C|s|.
 \label{eq:single-flow-linearization}
\end{equation}
Iterating \eqref{eq:single-flow-linearization} yields
\eqref{eq:flow-chart-linearization}.  Choose \(r_0\) so that its Lipschitz
remainder is at most half the least singular value of every selected
\(V_I(x_0)\).  The quantitative Lipschitz inverse theorem gives
\eqref{eq:flow-chart-bilip}.  More explicitly, differentiating the ordered
composition in \(s_k\) gives the \(k\)-th column
\begin{equation}
 D_{s_k}\Theta_{x_0,I}(s)
 =D\Phi_{i_d}^{s_d}\cdots D\Phi_{i_{k+1}}^{s_{k+1}}
   v_{i_k}(x_k),
 \label{eq:ordered-chart-column-formula}
\end{equation}
where \(x_k=\Phi_{i_k}^{s_k}\circ\cdots\circ
\Phi_{i_1}^{s_1}(x_0)\).  The flow variational equation and the Lipschitz
bounds give
\(\|D_s\Theta_{x_0,I}(s)-V_I(x_0)\|\le C|s|\).
The determinant Lipschitz estimate on bounded matrices, together with
\eqref{eq:ordered-frame-minor-bound}, yields the frozen-Jacobian estimate and,
after reducing \(r_0\) once more, the two positive Jacobian bounds in
\eqref{eq:flow-chart-jacobian}.  These identities hold for the
\(C^{1,1}\) representatives furnished by \(W^{2,\infty}\); no smooth
approximation of the vector fields is needed.  Finally, two applications of
\eqref{eq:single-flow-linearization} give
\eqref{eq:flow-chart-commutator-error}.  Compactness of \(\mathbb T^d\)
produces the finite cover.  Notice that no Frobenius condition was used.
\end{proof}

Fix once and for all protected chart pairs
\(U_\alpha^-\Subset U_\alpha^+\), selected frames \(I_\alpha\), and a
subordinate partition of unity \(\{\chi_\alpha\}\) from
Lemma~\ref{lem:ordered-flow-local-coordinates}.  Denote the overlap
multiplicity by \(m_{\rm atl}\), and set
\begin{align}
 V_0&:=\max_i\|v_i\|_\infty,&
 V_1&:=\max_i\|Dv_i\|_\infty,\notag\\
 V_2&:=\max_i\|D^2v_i\|_\infty,&
 D_0&:=\max_i\|\operatorname{div}v_i\|_\infty,\notag\\
 \delta_{\rm atl}&:=\min_\alpha
 \operatorname{dist}(U_\alpha^-,\mathbb T^d\setminus U_\alpha^+),&
 G_{\rm atl}&:=\max_\alpha\|D\chi_\alpha\|_\infty,&
 K_{\rm fr}&:=\max_{\alpha,x\in U_\alpha^+}
 \|V_{I_\alpha}(x)^{-1}\| .
 \label{eq:ordered-flow-constant-data}
\end{align}
Choose \(\rho\in C_c^1((a,b))\), \(\rho\ge0\), with
\(\int\rho=1\), and let
\begin{equation}
 k_\rho(s):=\int_\R\rho(t)\rho(t-s)\dd t .
 \label{eq:ordered-flow-autocorrelation}
\end{equation}
The autocorrelation is continuous and \(k_\rho(0)>0\).  Fix
\(0<c_\rho<(b-a)/4\), \(k_\rho\ge k_*>0\) on
\([-2c_\rho,2c_\rho]\), and choose
\(\eta\in C_c^1((-c_\rho,c_\rho))\), \(\eta\ge0\),
\(\int\eta=1\).  Set
\begin{equation}
 R_\rho:=k_*^{-1}\bigl(1+\|\eta\|_\infty+
 \|\eta'\|_\infty\bigr)
 \left\|\frac{\rho}{w}\right\|_\infty .
 \label{eq:ordered-flow-radial-constant}
\end{equation}
Fix \(\eps_{\rm sm}>0\) and \(C_{\rm sm}<\infty\) depending only on the
dimension, the radial law, the spanning constant, the coefficient and vector
field bounds, and the finite atlas data in
\eqref{eq:ordered-flow-constant-data}.  Appendix
\ref{app:ordered-flow-constants} gives an explicit admissible choice in
\eqref{eq:ordered-flow-explicit-constant}; its numerical value is not used.

\begin{lemma}[Ordered-flow smoothing]
\label{lem:ordered-flow-smoothing}
Let \(z_\eps:\mathbb T^d\to\{0,1\}\) satisfy
\begin{equation}
 B_\eps:=\sum_i\int_a^b\frac1\eps\int
 |z_\eps-z_\eps\circ\Phi_i^{-\eps t}|\dd x\lambda(\dd t)\le C.
 \label{eq:flow-ray-averaged-bound}
\end{equation}
For \(0<\eps\le\eps_{\rm sm}\), there are
\(Z_\eps\in BV(\mathbb T^d)\) such that
\begin{equation}
 \|Z_\eps-z_\eps\|_{L^1}\le C_{\rm sm}\eps B_\eps,
 \qquad |DZ_\eps|(\mathbb T^d)\le C_{\rm sm} B_\eps.
 \label{eq:ordered-flow-smoothing-bound}
\end{equation}
The constant is independent of \(\eps\) and \(z_\eps\), and its dependence
on all analytic and atlas data is displayed in
\eqref{eq:ordered-flow-constant-data}--
\eqref{eq:ordered-flow-explicit-constant}.
\end{lemma}

\begin{proof}
We use parameter-space integration by parts, which avoids differentiating an
unsmoothed field through a chain of noncommuting averages.

\emph{Step 1: signed increments supplied by the radial interval.}
For one field put
\(d_{i,\eps}(s):=\|z_\eps-z_\eps\circ\Phi_i^{-\eps s}\|_1\).
The group property, the triangle inequality, and the Liouville bound give
\begin{align}
 \int_\R k_\rho(s)d_{i,\eps}(s)\dd s
 &=\iint\rho(t)\rho(r)d_{i,\eps}(t-r)\dd t\dd r\notag\\
 &\le C e^{2\eps_0bD_0}
   \int_a^b\rho(t)d_{i,\eps}(t)\dd t .
 \label{eq:ordered-signed-increment}
\end{align}
Indeed, insert the intermediate state
\(z_\eps\circ\Phi_i^{\eps r}\), then change variables once in each term.
Since \(k_\rho\ge k_*\) on the support of \(\eta\) and \(\eta'\),
\begin{equation}
 \int_{-c_\rho}^{c_\rho}(\eta+|\eta'|)(s)
 d_{i,\eps}(s)\dd s
 \le C\eps R_\rho B_\eps,
 \label{eq:ordered-signed-source}
\end{equation}
where the lower bound for \(w\) converts \(\rho\dd t\) to
\(\lambda(\dd t)\).

\emph{Step 2: a full-dimensional ordered parameter kernel.}
For the selected frame \(I_\alpha=(i_{\alpha,1},\ldots,i_{\alpha,d})\), set
\begin{align}
 T_{\alpha,\eps}(x,s)
 &:={\Phi}_{i_{\alpha,d}}^{-\eps s_d}\circ\cdots\circ
 {\Phi}_{i_{\alpha,1}}^{-\eps s_1}(x),\notag\\
 \eta_d(s)&:=\prod_{k=1}^d\eta(s_k),\qquad
 A_{\alpha,\eps}z(x):=\int_{\R^d}\eta_d(s)
 z(T_{\alpha,\eps}(x,s))\dd s .
 \label{eq:ordered-flow-average}
\end{align}
For \(x\in U_\alpha^-\), every point in this formula lies in
\(U_\alpha^+\).  Write
\begin{equation}
 \mathsf A_{\alpha,\eps}(x,s):=D_xT_{\alpha,\eps}(x,s),\qquad
 \mathsf B_{\alpha,\eps}(x,s):=-\eps^{-1}D_sT_{\alpha,\eps}(x,s).
 \label{eq:ordered-parameter-matrices}
\end{equation}
Put \(x_0=x\) and
\(x_k=\Phi_{i_{\alpha,k}}^{-\eps s_k}(x_{k-1})\).  If
\begin{equation}
 R_k:=D\Phi_{i_{\alpha,d}}^{-\eps s_d}(x_{d-1})\cdots
 D\Phi_{i_{\alpha,k+1}}^{-\eps s_{k+1}}(x_k),
 \qquad R_d=I,
 \label{eq:ordered-later-flow-product}
\end{equation}
then the chain rule gives the exact column formula
\begin{equation}
 \mathsf B_{\alpha,\eps}e_k
 =R_kv_{i_{\alpha,k}}(x_k),
 \qquad
 \mathsf A_{\alpha,\eps}
 =D\Phi_{i_{\alpha,d}}^{-\eps s_d}(x_{d-1})\cdots
  D\Phi_{i_{\alpha,1}}^{-\eps s_1}(x_0).
 \label{eq:ordered-parameter-columns}
\end{equation}
In particular \(\mathsf B_{\alpha,\eps}(x,0)=V_{I_\alpha}(x)\) and
\(\mathsf A_{\alpha,\eps}(x,0)=I\).

The first and second variational equations for a \(C^{1,1}\) flow give,
for \(|r|\le\eps c_\rho\),
\begin{align}
 \|D\Phi_i^r\|&\le e^{|r|V_1},\notag\\
 \|D\Phi_i^r-I\|&\le e^{|r|V_1}-1,\notag\\
 \|D^2\Phi_i^r\|&\le |r|V_2e^{3|r|V_1}.
 \label{eq:ordered-flow-first-second-variation}
\end{align}
Indeed, differentiate \(\dot\Phi_i^r=v_i(\Phi_i^r)\) once and twice in the
initial point and apply Gronwall.  Since
\(|x_k-x|\le d\eps c_\rho V_0e^{d\eps c_\rho V_1}\),
\eqref{eq:ordered-parameter-columns} yields
\begin{align}
 \|\mathsf A_{\alpha,\eps}-I\|
 &\le e^{d\eps c_\rho V_1}-1,\notag\\
 \|\mathsf B_{\alpha,\eps}-V_{I_\alpha}(x)\|
 &\le C_d\eps c_\rho V_0V_1e^{2d\eps c_\rho V_1}.
 \label{eq:ordered-parameter-freezing}
\end{align}
The definition of \(\eps_{\rm sm}\) and the Neumann-series estimate therefore
give
\begin{equation}
 \|\mathsf A_{\alpha,\eps}\|+\|\mathsf B_{\alpha,\eps}\|
 +\|\mathsf B_{\alpha,\eps}^{-1}\|
 \le C_d(1+K_{\rm fr})e^{2d\eps_0c_\rho V_1}.
 \label{eq:ordered-parameter-inverse-bound}
\end{equation}
Differentiating \eqref{eq:ordered-parameter-columns} in one parameter and
using \eqref{eq:ordered-flow-first-second-variation} gives, term by term,
\begin{equation}
 \|\nabla_s\mathsf A_{\alpha,\eps}\|+
 \|\nabla_s\mathsf B_{\alpha,\eps}\|
 \le C_d\eps(1+V_0)(V_1+V_2)
 e^{C_d d\eps_0c_\rho V_1}.
 \label{eq:ordered-parameter-matrix-bounds}
\end{equation}
Thus, for each Euclidean basis vector \(e_\ell\),
\begin{equation}
 a_{\alpha,\ell}:=\mathsf B_{\alpha,\eps}^{-1}
 \mathsf A_{\alpha,\eps}e_\ell .
 \end{equation}
 It satisfies
 \begin{equation}
 \|a_{\alpha,\ell}\|_{L^\infty_s}
 +\|\nabla_sa_{\alpha,\ell}\|_{L^\infty_s}
 \le C_d(1+K_{\rm fr})^2(1+V_0)(1+V_1+V_2)
 e^{C_d d\eps_0c_\rho V_1}.
 \label{eq:ordered-parameter-coefficients}
\end{equation}

For smooth \(z\), the identity
\[
 \nabla z(T_{\alpha,\eps})\mathsf A_{\alpha,\eps}e_\ell
 =-\eps^{-1}\sum_{k=1}^d
 a_{\alpha,\ell k}\,\partial_{s_k}
 z(T_{\alpha,\eps})
\]
and integration by parts in the compact parameter cube give
\begin{equation}
 \partial_\ell A_{\alpha,\eps}z(x)
 =\frac1\eps\sum_{k=1}^d\int
 \partial_{s_k}(\eta_d a_{\alpha,\ell k})(x,s)
 [z(T_{\alpha,\eps}(x,s))-z(x)]\dd s .
 \label{eq:ordered-parameter-integration-by-parts}
\end{equation}
The subtraction is exact because the integral of every parameter derivative
vanishes.  It is the step that removes all derivatives of \(z\).

Let \(T_{\alpha,\eps}^{(0)}(x,s)=x\) and let
\(T_{\alpha,\eps}^{(j)}\) be the composition of the first \(j\) factors in
\eqref{eq:ordered-flow-average}.  The exact telescoping identity is
\begin{equation}
 z(T_{\alpha,\eps}(x,s))-z(x)
 =\sum_{j=1}^d
 \left[z(T_{\alpha,\eps}^{(j)}(x,s))
       -z(T_{\alpha,\eps}^{(j-1)}(x,s))\right].
 \label{eq:ordered-path-telescoping}
\end{equation}
Changing variables by \(T_{\alpha,\eps}^{(j-1)}\), whose inverse Jacobian is
at most \(e^{d\eps c_\rho D_0}\), gives
\begin{equation}
 \int_{U_\alpha^-}|z(T_{\alpha,\eps}(x,s))-z(x)|\dd x
 \le e^{d\eps c_\rho D_0}
 \sum_{j=1}^dd_{i_{\alpha,j},\eps}(s_j).
 \label{eq:ordered-telescoping-jacobian}
\end{equation}
Moreover \eqref{eq:ordered-parameter-coefficients} implies
\begin{equation}
 |\partial_{s_k}(\eta_da_{\alpha,\ell k})|
 \le C_{\rm par}(|\eta'(s_k)|+\eta(s_k))
       \prod_{j\ne k}\eta(s_j),
 \label{eq:ordered-parameter-weight-bound}
\end{equation}
where
\[
 C_{\rm par}:=C_d(1+K_{\rm fr})^2(1+V_0)(1+V_1+V_2)
 e^{C_d d\eps_0b(V_1+D_0)}.
\]
Insert \eqref{eq:ordered-telescoping-jacobian} first with weight \(\eta_d\)
and then with each weight in
\eqref{eq:ordered-parameter-weight-bound}.  Fubini and
\eqref{eq:ordered-signed-source} give, without an unlisted remainder,
\begin{align}
 \|A_{\alpha,\eps}z_\eps-z_\eps\|_{L^1(U_\alpha^-)}
 &\le d e^{d\eps c_\rho D_0} C\eps R_\rho B_\eps,
 \label{eq:localized-ordered-flow-distance}\\
 |D(A_{\alpha,\eps}z_\eps)|(U_\alpha^-)
 &\le d^2C_{\rm par}e^{d\eps c_\rho D_0}CR_\rho B_\eps.
 \label{eq:localized-ordered-flow-variation}
\end{align}
Equivalently, the whole parameter system closes in the single estimate
\begin{equation}
 |D(A_{\alpha,\eps}z_\eps)|(U_\alpha^-)
 +\eps^{-1}\|A_{\alpha,\eps}z_\eps-z_\eps\|_{L^1(U_\alpha^-)}
 \le C B_\eps.
 \label{eq:ordered-smoothing-coupled-absorption}
\end{equation}
There is no first-cycle derivative and no commutator remainder to absorb.
For general \(z_\eps\in L^1\), convolve it temporarily in the physical
variable.  The right-hand sides converge by continuity of composition with a
bi-Lipschitz flow, while the averaged fields converge in \(L^1\).  Lower
semicontinuity of variation passes the two estimates to \(z_\eps\).

The \(W^{2,\infty}\) regularity in
\eqref{eq:flow-ray-regularity} is used precisely in
\eqref{eq:ordered-flow-first-second-variation} and nowhere else in the
compactness argument.  Approximation is needed only for the scalar field
\(z\); all constants remain those displayed above.

\emph{Step 3: gluing.}
Choose \(\chi_\alpha\in C_c^1(U_\alpha^-)\),
\(\sum_\alpha\chi_\alpha=1\), and put
\begin{equation}
 Z_\eps:=\sum_\alpha\chi_\alpha
 A_{\alpha,\eps}z_\eps .
 \label{eq:localized-flow-gluing}
\end{equation}
Since
\begin{align}
 DZ_\eps={}&\sum_\alpha\chi_\alpha
 D(A_{\alpha,\eps}z_\eps)
 +\sum_\alpha(A_{\alpha,\eps}z_\eps-z_\eps)D\chi_\alpha,
 \label{eq:partition-cancellation-identity}
\end{align}
the finite overlap, \eqref{eq:localized-ordered-flow-distance}, and
\eqref{eq:localized-ordered-flow-variation} give
\eqref{eq:ordered-flow-smoothing-bound} with \(C_{\rm sm}\).
\end{proof}

 \begin{lemma}[Compactness for a uniformly spanning \(C^{1,1}\) flow family]
\label{lem:flow-ray-compactness}
Under \eqref{eq:flow-ray-regularity} and \eqref{eq:flow-ray-spanning}, binary
fields \(z_\eps\) satisfying \eqref{eq:flow-ray-averaged-bound} are relatively
compact in \(L^1(\mathbb T^d)\).  Every limit belongs to \(BV(\mathbb T^d)\).
\end{lemma}

\begin{proof}
Lemma~\ref{lem:ordered-flow-smoothing} gives a uniformly bounded sequence in
\(BV\) at vanishing \(L^1\) distance from \(z_\eps\).  BV compactness proves
the assertion.
\end{proof}

\begin{theorem}[Gamma-limit for a noncommuting rotating \(C^{1,1}\) ray family]
\label{thm:flow-ray-gamma-limit}
Assume \eqref{eq:ray-radial-law},
\eqref{eq:flow-ray-regularity},
\eqref{eq:flow-ray-spanning}, and
\eqref{eq:flow-ray-steepness}.  Then \(\mathcal G_\eps^{c,v}\) is
equicoercive and Gamma-converges in strong \(L^1(\mathbb T^d)\) to
\begin{equation}
 \mathcal G_0^{c,v}(u):=
 \begin{cases}
 \displaystyle\int_{\partial^*E}h_{\rm fl}(x,\nu_E(x))\dd\Hh^{d-1}(x),
 &u=\one_E,\ E\text{ of finite perimeter},\\[4pt]
 +\infty,&\text{otherwise}.
 \end{cases}
 \label{eq:flow-ray-limit}
\end{equation}
Again \(u_\eps=\one_E\) is a recovery sequence for every finite-perimeter
set.  The conclusion concerns the flow-generated energy
\eqref{eq:flow-ray-energy}; it is not a theorem for an unsymmetrized Euler
 sampling rule.  The fields may be noncommuting; no global chart or Frobenius
integrability is required.
\end{theorem}

\begin{proof}
The first marginal of \(Q_{i,\eps,t}\) has density \(c_i\).  After the change
of variables \(x=\Phi_i^{\eps t}(y)\), its second marginal has density
\(c_i(\Phi_i^{\eps t}(y))|\det D\Phi_i^{\eps t}(y)|\), bounded by
\(c^*e^{\eps_0bD_i}\).  Equations
\eqref{eq:diffuse-rounding-pointwise} and \eqref{eq:flow-ray-steepness}
therefore show that threshold rounding does not increase the energy, and the
potential again makes the rounding error vanish in \(L^1\).  The lower bound
\(c_i\ge c_*\) and Lemma~\ref{lem:flow-ray-compactness} give equicoercivity.

Let \(z_\eps\to z\) in \(L^1\).  For fixed \(i,t\), the signed measures
\begin{equation}
 \frac{z_\eps-z_\eps\circ\Phi_i^{-\eps t}}{\eps t}\,\mathcal L^d
 \stackrel{*}{\rightharpoonup}v_i\cdot Dz.
 \label{eq:flow-ray-measure-convergence}
\end{equation}
Indeed, testing against \(\phi\in C^1\), changing variables along the flow,
and differentiating at zero gives
\(-\int z\,\operatorname{div}(\phi v_i)\), the distributional pairing with
\(v_i\cdot Dz\).  Weighted lower semicontinuity and Fatou now yield
\begin{equation}
 \liminf_{\eps\downarrow0}\mathcal G_\eps^{c,v}(z_\eps)
 \ge\sum_i\int c_i\dd|v_i\cdot Dz|
 =\int h_{\rm fl}(x,\sigma_z)\dd|Dz|.
 \label{eq:flow-ray-liminf}
\end{equation}
Uniform first-order spanning makes the last measure coercive and identifies binary limits
with finite-perimeter sets.

For \(z=\one_E\), the BV transport formula for a Lipschitz flow gives
\begin{equation}
 \frac1s\int c_i(x)|z(x)-z(\Phi_i^{-s}(x))|\dd x
 \longrightarrow\int c_i\dd|v_i\cdot Dz|.
 \label{eq:flow-ray-translation-limit}
\end{equation}
Its left side is bounded, for \(0<s\le\eps_0b\), by a constant depending only
on the displayed Lipschitz data times \(|Dz|(\mathbb T^d)\).  Dominated
convergence in \(t\) proves the recovery statement.
\end{proof}

The freezing used above has an explicit local modulus.  Set
\begin{equation}
 L_h:=\sum_i\bigl(\|\nabla c_i\|_\infty\|v_i\|_\infty
 +c^*\|Dv_i\|_\infty\bigr).
 \label{eq:flow-ray-freezing-modulus}
\end{equation}
Then for all \(x,y,\xi\),
\begin{equation}
 |h_{\rm fl}(x,\xi)-h_{\rm fl}(y,\xi)|
 \le L_h|x-y||\xi|.
 \label{eq:flow-ray-density-freezing-error}
\end{equation}
Consequently, freezing at \(x_0\) on a ball \(B_r(x_0)\) costs at most
\begin{equation}
 L_hr\,|D\one_E|(B_r(x_0))
 \label{eq:flow-ray-perimeter-freezing-error}
\end{equation}
in the sharp perimeter.  If \(\partial E\) is globally \(C^2\), the
flow-box expansion behind \eqref{eq:flow-ray-translation-limit} gives
\begin{equation}
 \left|\mathcal G_\eps^{c,v}(\one_E)
 -\int_{\partial E}h_{\rm fl}(x,\nu_E)\dd\Hh^{d-1}\right|
 \le C_E\eps\int_a^b t^2\lambda(\dd t),
 \label{eq:flow-ray-smooth-interface-error}
\end{equation}
where \(C_E\) depends only on the norms in
\eqref{eq:flow-ray-freezing-modulus}, the uniform flow Jacobian bounds, and
the curvature and area of \(\partial E\).  For an arbitrary
finite-perimeter set the convergence remains valid but no uniform rate is
asserted.  At every \(x\), both \eqref{eq:heterogeneous-zonotopal-density} and
\(h_{\rm fl}(x,\cdot)\) are cosine transforms of positive finite atomic
measures, hence fields of zonoid norms.  The rotating theorem requires only
uniform first-order spanning.
The parameter-space smoothing does not commute directional derivatives through
successive averages.  Instead, \eqref{eq:ordered-smoothing-coupled-absorption}
uses the invertible ordered-flow parameter Jacobian and integration by parts
in all parameters simultaneously.  Bracket generation without pointwise
spanning is not a substitute
for the stated Euclidean-BV conclusion: the cell density \(h_{\rm fl}\) then
vanishes on a non-zero covector and is not coercive.  More sharply, if on an
open cube all fields lie in a fixed proper subspace, binary stripes in an
annihilating direction have bounded energy and no strongly convergent
subsequence.  Thus \(h_{\rm fl}\ge\kappa|\cdot|\), equivalently uniform
first-order spanning for this finite positive-weight family, is the minimal
frame condition for equicoercivity in Euclidean \(BV\).  Merely
bracket-generating families lead instead to a horizontal-perimeter problem.
The generally non-injective Euler sampling rule
 \(x\mapsto x-\eps t v_i(x)\) is excluded for the same reason.

\begin{proposition}[Sharp loss of Euclidean compactness without a first-order frame]
\label{prop:flow-ray-spanning-obstruction}
Suppose an open cube \(Q\Subset\mathbb T^d\), a proper linear subspace
\(H\subsetneq\R^d\), and \(C<\infty\) satisfy
\(v_i(x)\in H\) and \(|v_i(x)|\le C\) for every \(x\in Q\) and every \(i\).
Then the flow-bond energies are not equicoercive in strong \(L^1(Q)\): there
are binary \(z_n\) and \(\eps_n\downarrow0\) with
\(\sup_n\mathcal G_{\eps_n}^{c,v}(z_n)<\infty\) but no strongly convergent
subsequence on a smaller cube \(Q'\Subset Q\).
\end{proposition}

\begin{proof}
Choose a unit \(\eta\in H^\perp\), a cube \(Q'\Subset Q\), and let
\(z_n=\one_{Q'}\one_{\{\sin(2\pi n\eta\cdot x)>0\}}\).  Along every flow segment
remaining in \(Q'\),
\(\frac{d}{ds}(\eta\cdot\Phi_i^s(x))=\eta\cdot v_i(\Phi_i^s(x))=0\), so all
interior flow-bond differences vanish.  Bonds meeting \(\partial Q'\) occupy
a set of measure at most \(C\eps_n t\Hh^{d-1}(\partial Q')\); after division
by \(\eps_n\), their total cost is bounded independently of \(n\).  Take
\(\eps_n=o(n^{-1})\).  The potential is zero because \(z_n\) is binary, while
the half-volume stripe sequence has no strongly convergent subsequence in
\(L^1(Q')\).  This proves the claim.
\end{proof}

\section{Bounded Domains and the Full Crossing-Bond Contact Law}
\label{sec:flow-ray-boundary}

The preceding flow theorem was stated on the torus in order to isolate the
noncommuting-frame issue.  We next treat the physical boundary problem for the
same energy.  The convention is important: an imposed exterior state is
implemented by retaining every bond having at least one endpoint in the
body.  This is different from censoring the interaction at the boundary.

Let \(\Omega\Subset U\subset\R^d\) be a bounded Lipschitz domain.  The
coefficients and fields satisfy
\begin{equation}
 c_i\in W^{1,\infty}(U),\quad 0<c_*\le c_i\le c^*,\qquad
 v_i\in W^{2,\infty}(U;\R^d),\qquad
 \sum_i c_i|v_i\cdot\xi|\ge\kappa|\xi|
 \label{eq:boundary-ambient-frame}
\end{equation}
on a fixed neighbourhood of \(\overline\Omega\).  Their ambient
flows are denoted by \(\Phi_i^s\).  We decrease \(\eps_0\) so that every
flow segment with \(|s|\le b\eps_0\), starting in that neighbourhood
of \(\overline\Omega\), remains in \(U\).  This ambient extension is part of
the datum; the energy does not stop or reflect a flow at \(\partial\Omega\).

Choose a bounded Lipschitz set \(V\) with
\(\overline\Omega\Subset V\Subset U\), and put
\(D_{\rm e}=V\setminus\overline\Omega\).  Thus \(D_{\rm e}\) is a bounded
Lipschitz domain with two boundary components.  Let
\(g\in BV(\partial\Omega;\{0,1\})\), where boundary \(BV\) is defined
through any finite Lipschitz atlas.  Fix a binary function
\(g^{\rm e}\in BV(D_{\rm e};\{0,1\})\) whose trace is \(g\) on
\(\partial\Omega\) and zero on \(\partial V\), and extend it by zero to
\(U\setminus V\).
Lemma~\ref{lem:lipschitz-boundary-package} below proves that such an extension
exists; unlike the normal extension used for a \(C^2\) boundary, it does not
require a single-valued nearest-point projection.  For
\(u:\Omega\to[0,1]\), write
\begin{equation}
 \bar u(x)=u(x)\quad\hbox{in }\Omega,\qquad
 \bar u(x)=g^{\rm e}(x)\quad\hbox{in }U\setminus\overline\Omega.
 \label{eq:boundary-glued-state}
\end{equation}
Put
\begin{equation}
 B_\Omega:=\{(x,y)\in U\times U:x\in\Omega\ \hbox{or}\ y\in\Omega\}.
 \label{eq:crossing-bond-set}
\end{equation}
With \(\widehat Q_{i,\eps,t}\) defined by
\eqref{eq:flow-ray-pair-measure} using the ambient flow, define
\begin{align}
 \mathcal G_{\eps,\Omega}^{g}(u)
 :=&\frac M\eps\int_\Omega d_*(u)\dd x
 +\frac1\eps\sum_i\int_a^b
 \iint_{B_\Omega}|\bar u(x)-\bar u(y)|^2
 \widehat Q_{i,\eps,t}(\dd x,\dd y)\lambda(\dd t).
 \label{eq:boundary-crossing-flow-energy}
\end{align}
The restriction is imposed on the symmetric pair measure, not on one chosen
endpoint.  Thus an interior bond is counted with exactly the normalization of
\eqref{eq:flow-ray-energy}, and a crossing bond retains both half-orientations.
The same steepness condition \eqref{eq:flow-ray-steepness} suffices, since
restriction can only decrease both marginals.

\begin{lemma}[Lipschitz boundary package]
\label{lem:lipschitz-boundary-package}
Let \(\Omega\) be bounded and Lipschitz.  It has a finite cover by two-sided
cylinders which, after a rigid motion, are images under \(F\) of
\(Q'\times(-h,h)\), with \(Q'\times(0,h)\) on the interior side and
\(Q'\times(-h,0)\) on the exterior side.  In each such cylinder
\begin{equation}
 \Omega=\{(x',x_d):x_d>\varphi(x')\},\qquad
 \|\nabla\varphi\|_{L^\infty}\le L,
 \label{eq:lipschitz-graph-domain}
\end{equation}
and the flattening map
\(F(y',s)=(y',\varphi(y')+s)\) satisfies
\begin{equation}
 |F(y)-F(z)|\le(2+L)|y-z|,
 \quad |F^{-1}(x)-F^{-1}(w)|\le(2+L)|x-w|,
 \quad |\det DF|=1\quad\text{a.e.}
 \label{eq:lipschitz-flattening-estimate}
\end{equation}
At \(\Hh^{d-1}\)-almost every boundary point
\(x_0=(x'_0,\varphi(x'_0))\), set
\begin{align}
 \omega_{x_0}(r)&:=\frac1r
 \sup_{0<|z'|\le r}
 |\varphi(x'_0+z')-\varphi(x'_0)-\nabla\varphi(x'_0)\cdot z'|,
 \notag\\
 \eta_{x_0}(r)&:=\frac1{|B'_r|}\int_{B'_r(x'_0)}
 |\nabla\varphi(z')-\nabla\varphi(x'_0)|\dd z'.
 \label{eq:lipschitz-tangent-moduli}
\end{align}
Then \(\omega_{x_0}(r)+\eta_{x_0}(r)\to0\), and, writing \(H_{x_0}\)
for the tangent half-space and \(C_r(x_0)\) for the corresponding cylinder,
\begin{align}
 |(\Omega\mathbin\triangle H_{x_0})\cap C_r(x_0)|
 &\le C r^d\omega_{x_0}(r),
 \label{eq:lipschitz-halfspace-volume-error}\\
 \frac1{\Hh^{d-1}(\partial\Omega\cap C_r(x_0))}
 \int_{\partial\Omega\cap C_r(x_0)}
 |\nu_\Omega-\nu_\Omega(x_0)|\dd\Hh^{d-1}
 &\le C_L\eta_{x_0}(Cr).
 \label{eq:lipschitz-normal-mean-error}
\end{align}
Finally, if \(V\) and \(D_{\rm e}\) are as above, every
\(g\in BV(\partial\Omega;\{0,1\})\) is the inner trace of a binary
\(g^{\rm e}\in BV(D_{\rm e};\{0,1\})\) having zero outer trace.
\end{lemma}

\begin{proof}
The first two bounds in \eqref{eq:lipschitz-flattening-estimate} follow from
the Lipschitz estimate for \(\varphi\); the determinant identity follows from
the triangular a.e. derivative of \(F\).  Rademacher differentiability gives
\(\omega_{x_0}(r)\to0\), while the Lebesgue differentiation theorem applied
to \(\nabla\varphi\) gives \(\eta_{x_0}(r)\to0\).  Integrating the vertical
gap between the graph and its tangent plane proves
\eqref{eq:lipschitz-halfspace-volume-error}.  The map
\(p\mapsto(-p,1)/(1+|p|^2)^{1/2}\) is Lipschitz on \(\{|p|\le L\}\), and
the area element is bounded above and below in terms of \(L\); this proves
\eqref{eq:lipschitz-normal-mean-error}.

The same two-sided cylinders cover the exterior side
\(Q'\times(-h,0)\); after shrinking them, their union is an exterior collar
of \(\partial\Omega\).  More invariantly, use the standard theorem that for
every bounded Lipschitz domain \(D\), the trace map
\(\operatorname{Tr}:BV(D)\to L^1(\partial D)\) is onto (the right inverse
need not be linear or bounded in the \(L^1\) norm alone).  Apply it to
\(D=D_{\rm e}\) with boundary datum \(g\) on \(\partial\Omega\) and zero on
\(\partial V\), and truncate the resulting \(w\) to \([0,1]\).
For almost every \(s\in(0,1)\), both the coarea conclusion
\(\one_{\{w>s\}}\in BV(D_{\rm e})\) and the level-set trace identity
\[
 \operatorname{Tr}\one_{\{w>s\}}
 =\one_{\{\operatorname{Tr}w>s\}}
 \quad\Hh^{d-1}\text{-a.e. on }\partial D_{\rm e}
\]
hold.  Choose one \(s\) in this full-measure set.  Since the prescribed trace
is binary, the right-hand side is \(g\) on the inner component and zero on
the outer component.  This gives the asserted binary extension.  These trace,
extension, coarea, level-set trace, and gluing facts are standard for
Lipschitz domains; see \citep{ambrosio2000functions}.
\end{proof}

\begin{lemma}[Exact half-space crossing-bond cell]
\label{lem:boundary-halfspace-cell}
Fix constants \(c>0\), \(v\in\R^d\), \(t>0\), a unit vector \(n\), and
\(H=\{x\cdot n>0\}\).  Let \(z=A\) in \(H\) and \(z=B\) in
\(\R^d\setminus H\), where \(A,B\in\{0,1\}\).  Per unit
\((d-1)\)-area of the plane,
\begin{equation}
 \frac1\eps\iint_{B_H}|z(x)-z(y)|^2
 \frac c2\bigl[\dd x\,\delta_{x-\eps tv}(\dd y)
 +(\dd x\,\delta_{x-\eps tv}(\dd y))^{\mathsf T}\bigr]
 =ct|v\cdot n|\,|A-B|.
 \label{eq:boundary-halfspace-cell}
\end{equation}
Consequently the complete contact density is
\begin{equation}
 \int_a^b t\lambda(\dd t)\sum_i c_i(x)|v_i(x)\cdot n|
 |A-B|=h_{\rm fl}(x,n)|A-B|.
 \label{eq:boundary-contact-factor-one}
\end{equation}
In particular, its coefficient is one rather than one half.
\end{lemma}

\begin{proof}
For the unsymmetrized directed measure
\(c\,\dd x\,\delta_{x-\eps tv}(\dd y)\), the integrand is non-zero precisely
when the normal coordinate of \(x\) lies between \(0\) and
\(\eps t v\cdot n\).  That set is a slab of thickness
\(\eps t|v\cdot n|\).  Every such bond meets \(H\), and binary values give
\(|A-B|^2=|A-B|\), so its normalized integral is
\(ct|v\cdot n||A-B|\) per unit area.  In the symmetrized measure the directed
part and its transpose each carry weight \(1/2\) and have that same integral.
Their sum is therefore the displayed value, with neither a missing nor an
extra factor of two.  Summation and
\(\int t\lambda(\dd t)=1\) prove \eqref{eq:boundary-contact-factor-one}.
\end{proof}

For \(u=\one_E\in BV(\Omega;\{0,1\})\), set
\begin{align}
 \mathcal G_{0,\Omega}^{g}(u):={}&
 \int_{\partial^*E\cap\Omega}h_{\rm fl}(x,\nu_E)\dd\Hh^{d-1}
 \notag\\
 &+\int_{\partial\Omega}h_{\rm fl}(x,\nu_\Omega)
 |\operatorname{Tr}_\Omega u-g|\dd\Hh^{d-1},
 \label{eq:boundary-flow-limit}
\end{align}
and give it value \(+\infty\) otherwise.

\begin{lemma}[Ambient completion and the trace measure]
\label{lem:boundary-ambient-completion}
Let \(\overline V\Subset W\Subset U\), and define \(B_W\) as in
\eqref{eq:crossing-bond-set}, with \(W\) in place of \(\Omega\).  Extend
\(\bar u\) by zero on \(U\setminus V\).  Adding to
\(\mathcal G_{\eps,\Omega}^{g}\) the bonds in
\(B_W\setminus B_\Omega\) gives exactly the full crossing-bond energy of
\(\bar u\) on \(W\).  The added term depends only on \(g^{\rm e}\), is
uniformly bounded, and converges to
\begin{equation}
 \int_{V\setminus\overline\Omega}
 h_{\rm fl}(x,\sigma_{Dg^{\rm e}})\dd|Dg^{\rm e}|.
 \label{eq:boundary-fixed-exterior-limit}
\end{equation}
Moreover the gluing formula is
\begin{align}
 D\bar u={}&Du\mathbin{\llcorner}\Omega
 +Dg^{\rm e}\mathbin{\llcorner}(V\setminus\overline\Omega)
 +(g-\operatorname{Tr}_\Omega u)\nu_\Omega
 \Hh^{d-1}\mathbin{\llcorner}\partial\Omega .
 \label{eq:boundary-bv-gluing}
\end{align}
\end{lemma}

\begin{proof}
The pair set \(B_W\) is the disjoint union, up to a pair-measure null set, of
\(B_\Omega\) and \(B_W\setminus B_\Omega\): pairs in the second set have both
endpoints exterior to \(\Omega\), so their integrand is the fixed field
\(g^{\rm e}\).  This proves the exact energy decomposition without an
orientation convention.  Because \(g^{\rm e}=0\) on \(U\setminus V\), there
is no jump or contact contribution on \(\partial W\).  The fixed exterior
term converges by the local flow translation formula
\eqref{eq:flow-ray-translation-limit}; every interior jump of the arbitrary
fixed exterior extension, including chart-patching jumps if present, is
included in \(Dg^{\rm e}\).  Formula
\eqref{eq:boundary-bv-gluing} is the standard Gauss--Green product rule,
obtained first for smooth traces and then by strict \(BV\) approximation.
Its boundary polar is parallel to \(\nu_\Omega\), so the one-homogeneity and
evenness of \(h_{\rm fl}\) give exactly the second line of
\eqref{eq:boundary-flow-limit}.
\end{proof}

\begin{theorem}[Lipschitz-domain crossing-bond Gamma-limit]
\label{thm:boundary-flow-gamma-limit}
Assume \eqref{eq:ray-radial-law}, \eqref{eq:boundary-ambient-frame}, and
\eqref{eq:flow-ray-steepness}.  Then
\(\mathcal G_{\eps,\Omega}^{g}\) is equicoercive in strong
\(L^1(\Omega)\) and Gamma-converges to
\(\mathcal G_{0,\Omega}^{g}\) in \eqref{eq:boundary-flow-limit}.  More
precisely, bounded energy gives, after a subsequence,
\begin{equation}
 u_\eps\to u=\one_E\quad\hbox{in }L^1(\Omega),\qquad
 \bar u_\eps\to\bar u\quad\hbox{in }L^1(V),\qquad \bar u\in BV(V),
 \label{eq:boundary-trace-equicoercivity}
\end{equation}
so the limiting boundary value is the interior \(BV\) trace in
\eqref{eq:boundary-bv-gluing}.  For every finite-perimeter \(E\subset\Omega\),
the constant binary sequence \(u_\eps=\one_E\) is a recovery sequence.
\end{theorem}

\begin{proof}
Threshold rounding is unchanged because the crossing pair measure is
symmetric and has smaller marginals than the ambient measure.  We may
therefore replace \(u_\eps\) by binary \(z_\eps\) at vanishing \(L^1\)
cost.  Add the fixed exterior--exterior term from
Lemma~\ref{lem:boundary-ambient-completion}.  The resulting ambient energies
are bounded.  The ordered-flow smoothing proof of
Lemma~\ref{lem:ordered-flow-smoothing} is local: choose its protected flow
atlas in \(V\), including charts crossing \(\partial\Omega\).  It gives
compactness of the glued fields and hence
\eqref{eq:boundary-trace-equicoercivity}.  In particular, boundary
oscillations cannot disappear into a strip of thickness \(O(\eps)\).

Apply the ambient directional liminf
\eqref{eq:flow-ray-liminf} to \(\bar z_\eps\).  Since the exterior term is a
fixed sequence and has the exact limit
\eqref{eq:boundary-fixed-exterior-limit}, subtracting it gives
\begin{align*}
 \liminf_{\eps\downarrow0}\mathcal G_{\eps,\Omega}^{g}(z_\eps)
 \ge{}&\int_\Omega h_{\rm fl}(x,\sigma_{Du})\dd|Du|\\
 &+\int_{\partial\Omega}h_{\rm fl}(x,\nu_\Omega)
 |\operatorname{Tr}_\Omega u-g|\dd\Hh^{d-1}.
\end{align*}
Here subtraction is legitimate because the exterior summand is independent
of \(z_\eps\) and converges, rather than merely satisfying a liminf.
Formula \eqref{eq:boundary-bv-gluing} identifies the displayed measures.

Conversely take \(z=\one_E\) in \(\Omega\) and glue it to \(g^{\rm e}\).
The ambient \(BV\) transport formula gives convergence of the full ambient
bond energy of \(\bar z\).  Subtract the convergent fixed exterior term.
The remaining measure is precisely the interior part and the jump part in
\eqref{eq:boundary-bv-gluing}; Lemma~\ref{lem:boundary-halfspace-cell}
fixes its normalization.  The potential vanishes identically.  This proves
the upper bound and the theorem.
\end{proof}

We next record exactly what replaces the \(C^2\) tubular estimates.  Let
\(m_2=\int_a^b t^2\lambda(\dd t)\),
\(L_v=\max_i\|Dv_i\|_\infty\),
and \(L_c=\max_i\|\nabla c_i\|_\infty\).  At every boundary point \(x_0\)
for which Lemma~\ref{lem:lipschitz-boundary-package} applies, a chart of
diameter \(r\), coefficient freezing, and flow linearization produce
\begin{align}
 |(\Omega\mathbin\triangle H_{x_0})\cap C_r(x_0)|
 &\le Cr^d\omega_{x_0}(r),\notag\\
 \frac1{\Hh^{d-1}(\partial\Omega\cap C_r(x_0))}
 \int_{\partial\Omega\cap C_r(x_0)}
 |\nu_\Omega-\nu_\Omega(x_0)|\dd\Hh^{d-1}
 &\le C_L\eta_{x_0}(Cr),\notag\\
 |c_i(x)-c_i(x_0)|+c^*|v_i(x)-v_i(x_0)|
 &\le (L_c+c^*L_v)r,\notag\\
 |\Phi_i^{-\eps t}(x)-x+\eps t v_i(x)|
 &\le C L_v\eps^2t^2.
 \label{eq:boundary-chart-freezing-errors}
\end{align}
The graph-flattening Jacobian is exactly one a.e.; its surface Jacobian lies
between \(1\) and \((1+L^2)^{1/2}\).  The ordered-flow commutator remains
\(O(\eps^2t^2)\), as in \eqref{eq:flow-chart-commutator-error}.  Thus the
direct blow-up error, after division by the area scale \(r^{d-1}\), is bounded
by
\begin{equation}
 C_L\{\omega_{x_0}(r)+\eta_{x_0}(Cr)
 +(L_c+c^*L_v)r+L_v\eps m_2+\eps/r\}.
 \label{eq:lipschitz-boundary-blowup-estimate}
\end{equation}
Choosing first a differentiability and Lebesgue point \(x_0\), then
\(r\downarrow0\), and finally \(\eps/r\downarrow0\), makes this error bound
vanish.  No global rate is asserted: a Lipschitz graph has no uniform modulus
for \(\omega_{x_0}\) or \(\eta_{x_0}\).  The actual liminf and recovery proof
above is stronger and shorter because ambient completion bypasses boundary
flattening altogether.  For every finite-perimeter \(E\), the ambient
translation formula applies directly to the glued \(BV\) state, so the
constant recovery sequence needs neither a smooth approximation of \(E\) nor
a curvature bound.

\begin{proposition}[Null cuts are invisible to Euclidean pair measures]
\label{prop:null-cut-invisible}
Let \(\Omega\) be a finite-perimeter body and let
\(K\subset\Omega^1\) satisfy \(\mathcal L^d(K)=0\), for example a compact
countably \((d-1)\)-rectifiable internal cut.  Put
\(\widetilde\Omega=\Omega\setminus K\).  Then
\begin{equation}
 \one_{\widetilde\Omega}=\one_\Omega\quad\mathcal L^d\text{-a.e.},
 \qquad
 \partial^*\widetilde\Omega=\partial^*\Omega
 \quad\Hh^{d-1}\text{-a.e.},
 \label{eq:null-cut-same-body}
\end{equation}
and, after the canonical identification of \(L^1\) classes,
\begin{equation}
 \mathcal G_{\eps,\widetilde\Omega}^{g}
 =\mathcal G_{\eps,\Omega}^{g}.
 \label{eq:null-cut-energy-invariance}
\end{equation}
Consequently a Euclidean finite-perimeter model cannot assign two exterior
traces or two contact costs to the two sides of \(K\).  Theorem
\ref{thm:lifted-cut-gamma-limit} below constructs a finite-scale pairing rule
on a genuine single-volume cut completion and derives both costs.
\end{proposition}

\begin{proof}
Removing \(K\) does not change the characteristic function as an
\(L^1\) function, hence does not change its distributional derivative or
reduced boundary.  Every finite-\(\eps\) pair integral is absolutely
continuous in its starting endpoint and is also unchanged; the flow graph of
a null set is null because the flow is bi-Lipschitz.  This proves
\eqref{eq:null-cut-energy-invariance}.  In particular, the constant state
\(u\equiv1\) has no interaction across \(K\), whereas a doubled crack with
two imposed zero exterior states would have the positive cost
\(2\int_Kh_{\rm fl}(x,\nu_K)\dd\Hh^{d-1}\).  The contradiction proves that
the two-face law cannot be encoded by changing an open representative on a
 null cut.
\end{proof}

\section{Direct Limits on Rectifiable Resolved Cuts}
\label{sec:rectifiable-direct-cuts}

The preceding boundary theorem has no internal cut.  We now add an arbitrary
fixed rectifiable cut without smoothing its geometry.  Let
\begin{equation}
 \mathfrak K=(K_j,\nu_j,g_j^+,g_j^-)_{j\in\mathcal J}
 \label{eq:rectifiable-cut-datum}
\end{equation}
be a finite or countable family such that each
\(K_j\subset\Omega\) is Borel and countably \((d-1)\)-rectifiable,
\begin{equation}
 \sum_{j\in\mathcal J}\Hh^{d-1}(K_j)<\infty,
 \qquad
 \Hh^{d-1}(K_j\cap K_k)=0\quad(j\ne k),
 \label{eq:rectifiable-cut-mass}
\end{equation}
and \(\nu_j\) is an \(\Hh^{d-1}\)-measurable choice of approximate unit
normal on \(K_j\).  The bank data
\(g_j^\pm:K_j\to\{0,1\}\) are only assumed measurable.  Set
\(K=\bigcup_jK_j\).  No condition is imposed on
\(\overline K\cap\partial\Omega\).

For an oriented countably rectifiable set \(A\Subset U\), a field \(v_i\),
and \(h>0\), define the backward crossing number
\begin{equation}
 N_{i,h}^A(x):=\#\{s\in(0,h):\Phi_i^{-s}(x)\in A\},
 \qquad P_{i,h}^A:=N_{i,h}^A\pmod 2.
 \label{eq:rectifiable-crossing-number}
\end{equation}
The value on the null set where the number is infinite is fixed to be zero.
The next lemma replaces every tubular-neighbourhood estimate used for a
regular cut.

\begin{lemma}[Exact crossing formula and first-order parity]
\label{lem:rectifiable-crossing-parity}
Let \(A\Subset U\) be Borel, countably \((d-1)\)-rectifiable, and of finite
surface measure, with an \(\Hh^{d-1}\)-measurable choice of approximate unit
normal \(\nu_A\).  If \(a\) is bounded,
Borel, and compactly supported in a common flow region, then
\begin{align}
 \int_Ua(x)N_{i,h}^A(x)\dd x
  =\int_A\int_0^h a(\Phi_i^sz)J\Phi_i^s(z)
       |v_i(z)\cdot\nu_A(z)|\dd s\dd\Hh^{d-1}(z).
 \label{eq:rectifiable-exact-crossing}
\end{align}
Moreover,
\begin{equation}
 \lim_{h\downarrow0}\frac1h\int_U|a(x)|
       \bigl(N_{i,h}^A(x)-P_{i,h}^A(x)\bigr)\dd x=0.
 \label{eq:rectifiable-parity-error}
\end{equation}
In particular, every uniformly bounded modification of the bond rule on arcs
with at least two crossings changes its surface-scaled integral by \(o(1)\).
\end{lemma}

\begin{proof}
The product \(A\times(0,h)\) is countably \(d\)-rectifiable.  Apply the area
formula to
\(F(z,s)=\Phi_i^s(z)\).  If
\(\tau_1,\ldots,\tau_{d-1}\) spans the approximate tangent plane of \(A\),
then \(D\Phi_i^s(z)v_i(z)=v_i(\Phi_i^sz)\), and hence
\[
 J_dF(z,s)
 =\left|\det(D\Phi_i^s\tau_1,\ldots,D\Phi_i^s\tau_{d-1},
              v_i(\Phi_i^sz))\right|
 =J\Phi_i^s(z)|v_i(z)\cdot\nu_A(z)|.
\]
The multiplicity of \(F\) at \(x\) is exactly \(N_{i,h}^A(x)\), which proves
\eqref{eq:rectifiable-exact-crossing} and its integrability.

For parity, work first in a flow box
\(\Theta(y,r)\) with \(\partial_r\Theta=v_i\circ\Theta\).  Rectifiable
coarea shows that, for almost every \(y\), the set
\(\{r:\Theta(y,r)\in A\}\) is finite on each relevant compact interval.
It therefore has positive minimum separation, and every sufficiently short
interval on that slice contains at most one point.  Thus
\(N_{i,h}^A=P_{i,h}^A\) eventually on almost every slice.  The normalized
multiplicity is dominated by the integrable slice count furnished by
\eqref{eq:rectifiable-exact-crossing}; dominated convergence proves the claim
in each box.  A finite localization treats \(\{|v_i|>\delta\}\).  The
remaining flux is bounded by
\(C\delta\Hh^{d-1}(A)\); send \(h\downarrow0\) and then
\(\delta\downarrow0\).  Finally,
\(\one_{\{N\ge2\}}\le (N-P)/2\), proving the last assertion.
\end{proof}

The resolved derivative is defined directly on the physical body.  For
\(u\in BV(\Omega)\cap L^\infty(\Omega)\), put
\begin{equation}
 D_{\mathfrak K}u
 :=Du-\sum_j(T_j^+u-T_j^-u)\nu_j
       \Hh^{d-1}\mathbin{\llcorner}K_j.
 \label{eq:rectifiable-resolved-derivative}
\end{equation}
The series converges in total variation by
\eqref{eq:rectifiable-cut-mass}.  The standard decomposition of \(Du\) gives
\begin{equation}
 D_{\mathfrak K}u=Du\mathbin{\llcorner}(\Omega\setminus K),
 \qquad
 |Du|(\Omega)=|D_{\mathfrak K}u|(\Omega)
 +\sum_j\int_{K_j}|T_j^+u-T_j^-u|\dd\Hh^{d-1}.
 \label{eq:rectifiable-resolved-splitting}
\end{equation}
Indeed, the absolutely continuous part of \(Du\) vanishes on \(K\), the
Cantor part does not charge a sigma-finite \(\Hh^{d-1}\) set, and the jump
part has the displayed two traces at almost every point of \(K\).

We next isolate the one-dimensional lower bound.  Let
\(I=(q_0,q_{m+1})\), let
\(A=\{q_1<\cdots<q_m\}\), and put \(I_j=(q_j,q_{j+1})\).  Give every
endpoint of every \(I_j\) an independent binary ghost
\(\gamma_{j,L},\gamma_{j,R}\), and extend a binary state \(z\) by
\begin{equation}
 \widetilde z_j(s)=
 \begin{cases}
  \gamma_{j,L},&s\le q_j,\\
  z(s),&q_j<s<q_{j+1},\\
  \gamma_{j,R},&s\ge q_{j+1}.
 \end{cases}
 \label{eq:rectifiable-split-extension}
\end{equation}

\begin{lemma}[Weighted split-line lower bound]
\label{lem:rectifiable-split-line}
Let \(w_\eps\ge0\) converge locally uniformly to a continuous weight \(w\),
uniformly for \(t\in(a,b)\), and let binary \(z_\eps\to z\) in \(L^1(I)\).
Then
\begin{align}
 &\liminf_{\eps\downarrow0}\frac1\eps
  \int_a^b\sum_{j=0}^m\int_\R w_\eps(s,t)
  |\widetilde z_{j,\eps}(s)-\widetilde z_{j,\eps}(s-\eps t)|^2
  \dd s\lambda(\dd t)\notag\\
 &\quad\ge
 \sum_{j=0}^m\int_\R w(s)\dd|D\widetilde z_j|(s).
 \label{eq:rectifiable-split-line-liminf}
\end{align}
For every binary state with finite right-hand side, equality is attained by
the constant sequence.  No estimate depends on the minimum separation of
the points of \(A\).
\end{lemma}

\begin{proof}
Fix \(t\in(a,b)\).  The signed difference measures
\[
 \frac{\widetilde z_{j,\eps}
       -\widetilde z_{j,\eps}(\,\cdot-\eps t)}{\eps t}\,\mathcal L^1
\]
converge distributionally to \(D\widetilde z_j\).  Since the states are
binary, the square equals the absolute value.  Weighted lower semicontinuity
of total variation gives \(t\int w\dd|D\widetilde z_j|\).  Sum over \(j\),
apply Fatou in \(t\), and use \(\int_a^bt\lambda(\dd t)=1\).  For a fixed
target, the one-dimensional \(BV\) translation formula and its variation
majorant give equality by dominated convergence.  Separation is used only to
identify the direct rule with the split extensions on a fixed slice; every
finite slice is eventually separated, while the estimate itself contains no
separation constant.
\end{proof}

We now define the finite-scale rule.  Put
\(\Sigma=K\cup\partial\Omega\).  For the directed representative
\(y=\Phi_i^{-\eps t}(x)\) of a pair in \(B_\Omega\), count intersections of
the flow arc with \(\Sigma\).  Up to tangential and endpoint null sets, define
the symmetric typed bond \(\Psi_{i,\eps,t}^{\Omega,\mathfrak K}[u](x,y)\) as
follows:
\begin{enumerate}
 \item with no crossing, use \(|u(x)-u(y)|^2\);
 \item with one crossing \(z\in K_j\), use
 \[
  |u(x)-g_j^\sigma(z)|^2+|u(y)-g_j^{-\sigma}(z)|^2,
  \qquad \sigma=\operatorname{sign}(v_i(z)\cdot\nu_j(z));
 \]
 \item with one crossing \(z\in\partial\Omega\), use
 \(|u(x_{\rm in})-g(z)|^2\), where \(x_{\rm in}\) is the interior endpoint;
 \item with at least two crossings, use zero.
\end{enumerate}
In the first line both endpoints lie in \(\Omega\).  The sum in the second
line makes the rule independent of the orientation chosen for the pair.  By
Lemma~\ref{lem:rectifiable-crossing-parity}, any uniformly bounded symmetric
choice in the fourth line has the same Gamma-limit.
We record the measurability of this rule, including the countable label case.
Enumerate \(\mathcal J\), if necessary, and replace \(K_j\) by the disjoint
Borel representative
\begin{equation}
 K_j^\circ:=K_j\setminus\bigcup_{k<j}K_k.
 \label{eq:rectifiable-disjoint-label-representatives}
\end{equation}
This changes neither a surface integral nor a flow-pair integral: the removed
sets are \(\Hh^{d-1}\)-null, and
\eqref{eq:rectifiable-exact-crossing} makes the corresponding family of arcs
Lebesgue-null.  Choose Borel representatives of the approximate normals, bank
data, and outer trace datum after changing them on surface-null sets.  Then
\(\Sigma=\partial\Omega\cup\bigcup_jK_j^\circ\) is Borel, and so is the
incidence relation
\begin{equation}
 \mathcal R_{i,h}:=\{(x,s)\in U\times(0,h):
                 \Phi_i^{-s}x\in\Sigma\}.
 \label{eq:rectifiable-crossing-relation}
\end{equation}
By \eqref{eq:rectifiable-exact-crossing} and
\eqref{eq:rectifiable-cut-mass}, its fibre is finite for Lebesgue-almost every
starting point in every compact flow region.  The area-formula multiplicity is
Lebesgue measurable, so we may choose a Borel full-measure set
\(X_{i,h}\subset U\) on which every fibre is finite.  Apply the
Lusin--Novikov theorem to
\(\mathcal R_{i,h}\cap(X_{i,h}\times(0,h))\); it enumerates the crossing times
by Borel partial functions.  Composition with the continuous flow gives Borel
crossing points, and membership in the disjoint sets \(K_j^\circ\) gives the
unique label measurably.  Hence the crossing number, the unique-crossing point
and label, the sign of \(v_i\cdot\nu_j\), and every case of \(\Psi\) are
measurable for the completed directed pair measure.  Arcs tangent at a
crossing, arcs with an endpoint on \(\Sigma\), fibres that are infinite, and
crossings in a removed label overlap form a null family by the same area
formula; define the bond arbitrarily there.  Symmetrization then gives a
measurable function for \(\widehat Q_{i,\eps,t}\).

We now prove the global translation formula used in every subsequent lower
bound and recovery argument.  For a physical flow time \(h>0\), set
\begin{align}
 Q_{i,h}(\dd x,\dd y)&:=c_i(x)\dd x\,\delta_{\Phi_i^{-h}(x)}(\dd y),
 &\widehat Q_{i,h}&:=\tfrac12(Q_{i,h}+Q_{i,h}^{\mathsf T}),\notag\\
 \mathcal B_h^\zeta(u)&:=\sum_i\iint_{B_\Omega}
 \frac{\zeta(x)+\zeta(y)}2\Psi_{i,h}^{\Omega,\mathfrak K}[u](x,y)
 \widehat Q_{i,h}(\dd x,\dd y),
 \label{eq:global-resolved-weighted-bond}
\end{align}
where \(\zeta\in C_c(U)\), and the typed bond is the rule above with
\(\eps t\) replaced by \(h\).  Lower-bound tests will always satisfy
\(\zeta\ge0\).  We write
\(\mathcal B_h=\mathcal B_h^1\), taking a continuous cutoff equal to one on
the common flow neighbourhood of \(\overline\Omega\).

For \(u\in BV(\Omega;\{0,1\})\), define the positive Radon measure on
\(\overline\Omega\)
\begin{align}
 \nu_{\mathfrak K,u}^{g}:={}&\sum_i c_i\,|v_i\cdot D_{\mathfrak K}u|\notag\\
 &+\sum_i c_i|v_i\cdot\nu_\Omega|\,|T_\partial u-g|
       \Hh^{d-1}\mathbin{\llcorner}\partial\Omega\notag\\
 &+\sum_{i,j}c_i|v_i\cdot\nu_j|
 \bigl(|T_j^+u-g_j^+|+|T_j^-u-g_j^-|\bigr)
       \Hh^{d-1}\mathbin{\llcorner}K_j.
\label{eq:global-resolved-limit-measure}
\end{align}
The total mass of this measure is the resolved variation and contact
functional appearing below in \eqref{eq:rectifiable-direct-limit}.

\begin{theorem}[Global resolved translation and contact formula]
\label{thm:global-resolved-translation-contact}
Assume \eqref{eq:boundary-ambient-frame} and
\eqref{eq:rectifiable-cut-mass}.  Let \(h_n\downarrow0\) and let
\(u_n\in L^1(\Omega;\{0,1\})\) converge to \(u\) in \(L^1(\Omega)\).
If
\begin{equation}
 \liminf_{n\to\infty}h_n^{-1}\mathcal B_{h_n}(u_n)<\infty,
 \label{eq:global-resolved-finite-budget}
\end{equation}
then \(u\in BV(\Omega;\{0,1\})\), and for every nonnegative
\(\zeta\in C_c(U)\),
\begin{equation}
 \liminf_{n\to\infty}\frac1{h_n}\mathcal B_{h_n}^{\zeta}(u_n)
 \ge\int_{\overline\Omega}\zeta\,\dd\nu_{\mathfrak K,u}^{g}.
 \label{eq:global-resolved-sequential-liminf}
\end{equation}
Conversely, for every \(u\in BV(\Omega;\{0,1\})\),
\begin{equation}
 \lim_{h\downarrow0}\frac1h\mathcal B_h^\zeta(u)
 =\int_{\overline\Omega}\zeta\,\dd\nu_{\mathfrak K,u}^{g}.
 \label{eq:global-resolved-translation-limit}
\end{equation}
There are \(h_0>0\) and a constant \(C\), depending only on the ambient flow
atlas and coefficient bounds, such that
\begin{equation}
 \sup_{0<h<h_0}\frac1h\mathcal B_h(u)
 \le C\left(|Du|(\Omega)+\Hh^{d-1}(\partial\Omega)
       +\sum_j\Hh^{d-1}(K_j)\right).
 \label{eq:global-resolved-translation-majorant}
\end{equation}
The same statements hold on \(\mathbb T^d\) after deleting the outer-boundary
term.
\end{theorem}

\begin{proof}
\emph{Step 1: directed representation and compactness of the limit.}
The pair set \(B_\Omega\), the weight
\((\zeta(x)+\zeta(y))/2\), and the typed bond are symmetric.  Consequently
\eqref{eq:global-resolved-weighted-bond} has the exact directed form
\begin{equation}
 \mathcal B_h^\zeta(u)=\sum_i\int_{D_{i,h}}c_i(x)
 \frac{\zeta(x)+\zeta(\Phi_i^{-h}x)}2
 \Psi_{i,h}^{\Omega,\mathfrak K}[u]
       (x,\Phi_i^{-h}x)\,\dd x,
 \label{eq:global-resolved-directed-form}
\end{equation}
where
\(D_{i,h}:=\{x\in U:x\in\Omega\ \text{or}\ \Phi_i^{-h}x\in\Omega\}\).

Suppose \eqref{eq:global-resolved-finite-budget} holds and first pass to a
subsequence, not relabelled, on which the unweighted quotient is bounded.  On every
\(\Omega'\Subset\Omega\), the ordinary binary increment is bounded by the
typed increment plus \(C\one_{\{N_{i,h_n}^K\ge1\}}\).  This follows from the
triangle inequality on a singly crossed arc; on the multiple class both sides
are bounded and the same indicator is used.  The exact crossing formula gives
\begin{equation}
 \sup_n\frac1{h_n}\sum_i\int_{\Omega'}c_i
 |u_n-u_n\circ\Phi_i^{-h_n}|\,\dd x<\infty
 \label{eq:global-resolved-interior-increment-bound}
\end{equation}
Testing the signed difference quotients against
\(C_c^1(\Omega')\) and passing to the limit shows that every
\(v_i\cdot Du\) is a finite measure there.  On each member of the finite
spanning atlas, choose \(d\) fields forming an invertible frame and reconstruct
the coordinate distributions \(D_k u\) as bounded continuous linear
combinations of these directional measures.  Hence \(Du\) is a vector measure.
The spanning inequality \eqref{eq:boundary-ambient-frame}, now applied to its
polar, gives
\(|Du|(\Omega')\le\kappa^{-1}\sum_i c_i|v_i\cdot Du|(\Omega')\).
Exhausting \(\Omega\) proves \(u\in BV(\Omega;\{0,1\})\).

\emph{Step 2: one flow box.}
Fix \(\zeta\ge0\).  If the left-hand side of
\eqref{eq:global-resolved-sequential-liminf} is infinite there is nothing to
prove.  Otherwise pass to a new subsequence, again not relabelled, which
attains that liminf.
Fix \(i\), a relatively compact flow box
\(\Theta:Y\times I\to U\) satisfying
\(\partial_r\Theta(y,r)=v_i(\Theta(y,r))\), and a nonnegative cutoff
\(\eta\in C_c(\Theta(Y\times I))\).  Put
\(J_\Theta=|\det D\Theta|\).  After extracting the subsequence realizing the
liminf, Fubini gives
\begin{equation}
 u_n\circ\Theta(y,\cdot)\longrightarrow u\circ\Theta(y,\cdot)
 \quad\text{in }L^1(I)\quad\text{for a.e. }y\in Y.
 \label{eq:global-resolved-slice-convergence}
\end{equation}
Rectifiable coarea applied to
\(\Sigma=K\cup\partial\Omega\) gives, for almost every \(y\), a finite set
\begin{equation}
 A_y:=\{r\in I:\Theta(y,r)\in\Sigma\}\cap
       \operatorname{spt}(\eta\circ\Theta).
 \label{eq:global-resolved-finite-slice}
\end{equation}
Tangential intersections form a set of zero transverse measure.  Surface-null
overlaps of the \(K_j\) also have zero transverse measure, so every internal
point of \(A_y\) has one label \(j\) and two bank values.  The components of
\(\{r:\Theta(y,r)\in\Omega\}\) cut by \(A_y\) are finitely many intervals on
the support of the cutoff.  Extend the state on each such interval by its
prescribed bank value at an internal endpoint and by \(g\) at an outer
endpoint, exactly as in \eqref{eq:rectifiable-split-extension}.

For this fixed slice, the points in \(A_y\) have positive minimum separation.
Hence, for all sufficiently large \(n\), the direct rule at distance \(h_n\)
equals the sum of the split-extension increments wherever \(\eta\ne0\).
Artificial endpoints of the flow box contribute nothing because \(\eta\)
vanishes near them.  The corresponding one-dimensional weight is
\begin{equation}
 w_{n,y}(r):=c_i(\Theta(y,r))J_\Theta(y,r)
 \frac{\eta(\Theta(y,r))+\eta(\Theta(y,r-h_n))}{2},
 \label{eq:global-resolved-slice-weight}
\end{equation}
and it converges locally uniformly to
\(w_y(r)=c_i(\Theta(y,r))J_\Theta(y,r)\eta(\Theta(y,r))\).
The fixed-shift proof of Lemma~\ref{lem:rectifiable-split-line}, with
\(\eps t=h_n\), therefore yields the complete sliced lower bound
\begin{equation}
 \liminf_{n\to\infty}\frac1{h_n}\mathcal B_{i,h_n}^{\eta,y}(u_n)
 \ge\sum_{J\in\mathcal I_y}\int_{\mathbb R}w_y\,\dd|D\widetilde u_{y,J}|.
\label{eq:global-resolved-sliced-liminf}
\end{equation}
Here \(\mathcal B_{i,h_n}^{\eta,y}\) denotes the \(y\)-slice of the
\(i\)-th localized integral in
\eqref{eq:global-resolved-directed-form}, \(\mathcal I_y\) is the finite
family of physical slice intervals, and
\(\widetilde u_{y,J}\) is its independently ghosted extension.  This formula
contains, on the same footing, the variation inside the intervals, two
endpoint mismatches at an internal cut, and one mismatch at an outer
boundary.

\emph{Step 3: transverse integration and geometric identification.}
Fatou's lemma in \(y\) applies to the nonnegative localized interaction.
The slicing theorem for \(BV\) functions gives
\begin{equation}
 \int_Y\sum_{J\in\mathcal I_y}\int_Jw_y\,\dd|D(u\circ\Theta)_y|\dd y
 =\int_{\Theta(Y\times I)\cap\Omega}
 \eta c_i\,\dd|v_i\cdot D_{\mathfrak K}u|.
 \label{eq:global-resolved-bulk-slicing}
\end{equation}
At a rectifiable endpoint, the coarea factor is identified by
\begin{equation}
 J_\Theta(y,r)\,\dd y
 =|v_i(z)\cdot\nu_A(z)|\,\dd\Hh^{d-1}(z),\qquad
 z=\Theta(y,r),\quad z\in A,
 \label{eq:global-resolved-coarea-factor}
\end{equation}
which is the cofactor identity underlying
\eqref{eq:rectifiable-exact-crossing}.  Therefore the internal endpoints in
\eqref{eq:global-resolved-sliced-liminf} integrate to
\begin{equation}
 \sum_j\int_{K_j}\eta c_i|v_i\cdot\nu_j|
 \bigl(|T_j^+u-g_j^+|+|T_j^-u-g_j^-|\bigr)\dd\Hh^{d-1},
 \label{eq:global-resolved-bank-identification}
\end{equation}
and the outer endpoints integrate to
\begin{equation}
 \int_{\partial\Omega}\eta c_i|v_i\cdot\nu_\Omega|
 |T_\partial u-g|\dd\Hh^{d-1}.
 \label{eq:global-resolved-boundary-identification}
\end{equation}

The countable cut causes no hidden limiting operation.  Indeed,
\eqref{eq:global-resolved-finite-slice} contains only finitely many active
labels on almost every localized slice.  After transverse integration the
positive endpoint terms are the sum over all \(j\) in
\eqref{eq:global-resolved-bank-identification}; equivalently, first retain
\(j\le m\) and then let \(m\uparrow\infty\).  Tonelli's theorem and
\eqref{eq:rectifiable-cut-mass} justify this monotone exhaustion.

\emph{Step 4: exhaustion of the flow atlas.}
Fix a nondecreasing \(\vartheta\in C([0,\infty);[0,1])\) that vanishes on
\([0,1/2]\) and equals one on \([1,\infty)\), and put
\(\zeta_\delta=\zeta\vartheta(|v_i|/\delta)\).  The compact support of
\(\zeta_\delta\) is contained in \(\{|v_i|>\delta/2\}\) and is covered by
finitely many protected flow boxes.  Choose a partition of unity
\(\{\theta_k\}\) equal to one in sum on a neighbourhood of that support, and
use the local weights \(\eta_k=\zeta_\delta\theta_k\) in Steps 2--3.  Write
\(\mathcal B_{i,h}^\eta\) for the \(i\)-th summand of
\eqref{eq:global-resolved-weighted-bond}.  Linearity of the endpoint average
gives the exact identity
\begin{equation}
 \sum_k\mathcal B_{i,h}^{\eta_k}(u)
 =\mathcal B_{i,h}^{\zeta_\delta}(u)
 \quad\text{for every }h>0.
 \label{eq:global-resolved-atlas-partition}
\end{equation}
Thus summing the local inequalities gives the directional limit measure tested
against \(\zeta_\delta\).  As \(\delta\downarrow0\), these tests increase to
\(\zeta\) on \(\{|v_i|>0\}\).  Monotone convergence recovers the whole
directional measure because every one of its bulk, bank, and boundary
densities contains the factor \(|v_i\cdot\nu|\) and hence vanishes on
\(\{v_i=0\}\).  Summing over the finite direction family proves
\eqref{eq:global-resolved-sequential-liminf} and identifies the right-hand
side with \eqref{eq:global-resolved-limit-measure}.

\emph{Step 5: fixed-state convergence and domination.}
Let \(u\in BV(\Omega;\{0,1\})\) be fixed.  On almost every slice, the
one-dimensional \(BV\) translation formula applied to every split extension
turns \eqref{eq:global-resolved-sliced-liminf} into an equality in the limit.
For \(0<h<h_0\), the translation inequality gives the slice majorant
\begin{equation}
 \frac1h\sum_{J\in\mathcal I_y}\int
 |\widetilde u_{y,J}(r)-\widetilde u_{y,J}(r-h)|\dd r
 \le\sum_{J\in\mathcal I_y}|D\widetilde u_{y,J}|(\mathbb R).
 \label{eq:global-resolved-slice-majorant}
\end{equation}
The right-hand side is transversely integrable: its physical part is bounded
by \(C|Du|(\Omega)\), each internal endpoint by two, and each outer endpoint
by one.  Coarea and \eqref{eq:rectifiable-cut-mass} therefore bound its
integral by the right-hand side of
\eqref{eq:global-resolved-translation-majorant}.  Multiply crossed arcs are
dominated by the same slice count and converge to zero by
\eqref{eq:rectifiable-parity-error}.  Dominated convergence first in \(y\),
then over the finite atlas and the finite direction family, proves
\eqref{eq:global-resolved-translation-limit} and
\eqref{eq:global-resolved-translation-majorant}.  The periodic proof is the
same, with no outer endpoints.
\end{proof}

Define
\begin{align}
 \mathcal H_{\eps,\Omega,\mathfrak K}^{g}(u)
 :={}&\frac M\eps\int_\Omega d_*(u)\dd x
 +\frac1\eps\sum_i\int_a^b\iint_{B_\Omega}
 \Psi_{i,\eps,t}^{\Omega,\mathfrak K}[u](x,y)
 \widehat Q_{i,\eps,t}(\dd x,\dd y)\lambda(\dd t),
 \label{eq:rectifiable-direct-energy}
\end{align}
with value \(+\infty\) outside \(0\le u\le1\).  The candidate limit is
\begin{align}
 \mathcal H_{0,\Omega,\mathfrak K}^{g}(u)
 :={}&\int_\Omega h_{\rm fl}(x,\sigma_{D_{\mathfrak K}u})
       \dd|D_{\mathfrak K}u|\notag\\
 &+\int_{\partial\Omega}h_{\rm fl}(x,\nu_\Omega)
       |T_\partial u-g|\dd\Hh^{d-1}\notag\\
 &+\sum_j\int_{K_j}h_{\rm fl}(x,\nu_j)
 \bigl(|T_j^+u-g_j^+|+|T_j^-u-g_j^-|\bigr)\dd\Hh^{d-1},
 \label{eq:rectifiable-direct-limit}
\end{align}
with value \(+\infty\) outside \(BV(\Omega;\{0,1\})\).

\begin{theorem}[Master Gamma-limit on a rectifiable resolved cut]
\label{thm:rectifiable-cut-master-gamma}
Assume \eqref{eq:ray-radial-law}, \eqref{eq:boundary-ambient-frame},
\eqref{eq:flow-ray-steepness}, and \eqref{eq:rectifiable-cut-mass}.  Then
\(\mathcal H_{\eps,\Omega,\mathfrak K}^{g}\) is equicoercive in strong
\(L^1(\Omega)\) and
\begin{equation}
 \mathcal H_{\eps,\Omega,\mathfrak K}^{g}
 \mathop{\longrightarrow}^{\Gamma}
 \mathcal H_{0,\Omega,\mathfrak K}^{g}.
 \label{eq:rectifiable-master-gamma}
\end{equation}
Every binary state of finite limit energy is recovered by the constant
sequence.  Consequently the Gamma-limit contains no compulsory term on a
surface-null branch or tip set, or on
\(\overline K\cap\partial\Omega\); this does not exclude excess concentration
for non-recovery sequences.
\end{theorem}

\begin{proof}
\emph{Step 1: rounding and equicoercivity.}
The threshold map \(\chi(r)=\one_{\{r\ge1/2\}}\) satisfies
\begin{align}
 |\chi(r)-\chi(s)|^2-|r-s|^2
 &\le2\{d_*(r)+d_*(s)\},\notag\\
 |\chi(r)-\gamma|^2-|r-\gamma|^2
 &\le2d_*(r),\qquad \gamma\in\{0,1\}.
 \label{eq:rectifiable-rounding-inequalities}
\end{align}
The two marginals of \(\widehat Q_{i,\eps,t}\) obey the Liouville bounds used
in \eqref{eq:flow-ray-steepness}.  Integrating
\eqref{eq:rectifiable-rounding-inequalities} over ordinary and typed bonds,
and then using \eqref{eq:flow-ray-steepness}, gives
\begin{equation}
 \mathcal H_{\eps,\Omega,\mathfrak K}^{g}(\chi\circ u)
 \le \mathcal H_{\eps,\Omega,\mathfrak K}^{g}(u),
 \qquad
 \|u-\chi\circ u\|_{L^1(\Omega)}
 \le\frac\eps M\mathcal H_{\eps,\Omega,\mathfrak K}^{g}(u).
 \label{eq:rectifiable-rounding-global}
\end{equation}

Put \(\Sigma=K\cup\partial\Omega\).  Since
\(\one_{\{N\ge1\}}\le N\), the exact area formula
\eqref{eq:rectifiable-exact-crossing}, first for every \(K_j\) and then for
\(\partial\Omega\), yields
\begin{align}
 &\sup_{0<\eps<\eps_0}\frac1\eps
 \sum_i\int_a^b\int_U c_i(x)
 \one_{\{N_{i,\eps t}^{\Sigma}(x)\ge1\}}\dd x\lambda(\dd t)\notag\\
 &\qquad\le C\left(\Hh^{d-1}(\partial\Omega)
             +\sum_j\Hh^{d-1}(K_j)\right)<\infty.
 \label{eq:rectifiable-crossing-budget}
\end{align}
Here \(N^{\Sigma}\le N^{\partial\Omega}+\sum_jN^{K_j}\) outside the null
set where a crossing has more than one label, and Tonelli justifies the
countable sum.  For binary states the glued exterior bond in
\eqref{eq:boundary-crossing-flow-energy} is bounded by the typed bond plus a
fixed multiple of \(\one_{\{N_{i,\eps t}^{\Sigma}\ge1\}}\).  Therefore
\begin{equation}
 \mathcal G_{\eps,\Omega}^{g}(u)
 \le \mathcal H_{\eps,\Omega,\mathfrak K}^{g}(u)+C_{\Sigma}
 \quad\text{for every binary }u,
 \label{eq:rectifiable-glued-compactness-comparison}
\end{equation}
with \(C_\Sigma\) independent of \(u\) and \(\eps\).  Notice that this
comparison uses only boundedness of the binary mismatch; it does not assert a
uniform \(o(1)\) comparison between the two boundary rules.  Equations
\eqref{eq:rectifiable-rounding-global}--
\eqref{eq:rectifiable-glued-compactness-comparison} and the compactness part
of Theorem~\ref{thm:boundary-flow-gamma-limit} prove equicoercivity.

\emph{Step 2: localized interaction measures.}
Let \(\eps_n\downarrow0\), let \(u_n\to u\) in \(L^1(\Omega)\), and suppose
that the liminf of the energies is finite.  Pass to a subsequence attaining
that liminf and put \(z_n=\chi\circ u_n\).  By
\eqref{eq:rectifiable-rounding-global}, \(z_n\to u\), the sequence remains
energy bounded, and \(u\) is binary.  For \(\zeta\in C_c(U)\), set
\begin{equation}
 L_n(\zeta):=\frac1{\eps_n}\int_a^b
       \mathcal B_{\eps_nt}^{\zeta}(z_n)\lambda(\dd t).
 \label{eq:rectifiable-interaction-measure}
\end{equation}
This is a positive linear functional on \(C_c(U)\).  By the Riesz
representation theorem there is a unique positive Radon measure \(\mu_n\)
such that \(L_n(\zeta)=\int_U\zeta\dd\mu_n\).  Equivalently, \(\mu_n\) is the
radial average of the two endpoint pushforwards of the positive typed pair
measure.  Its support lies in one fixed compact flow neighbourhood of
\(\overline\Omega\), and its mass is bounded by the interaction part of the
energy.  After a further subsequence,
\begin{equation}
 \mu_n\stackrel{*}{\rightharpoonup}\mu
 \quad\text{in }\mathcal M^+(U).
 \label{eq:rectifiable-interaction-measure-limit}
\end{equation}

\emph{Step 3: fixed-radius localization and slicing.}
For fixed \(t\in(a,b)\), set \(h_n=\eps_nt\) and
\begin{equation}
 f_n(t):=\frac1{h_n}\mathcal B_{h_n}(z_n),
 \qquad
 f_n^\zeta(t):=\frac1{h_n}\mathcal B_{h_n}^\zeta(z_n).
 \label{eq:rectifiable-radial-quotients}
\end{equation}
The interaction bound says
\(\sup_n\int_a^b t f_n(t)\lambda(\dd t)<\infty\).  Fatou's lemma therefore
implies \(\liminf_nf_n(t)<\infty\) for \(\lambda\)-almost every \(t\).
Selecting one such radius and applying
Theorem~\ref{thm:global-resolved-translation-contact} first shows that
\(u\in BV(\Omega;\{0,1\})\), so
\(\nu_{\mathfrak K,u}^{g}\) is well defined.
For each such \(t\), Theorem~\ref{thm:global-resolved-translation-contact}
applies to the single sequence \(h_n\downarrow0\) and gives
\begin{equation}
 \liminf_{n\to\infty}f_n^\zeta(t)
 \ge\int_{\overline\Omega}\zeta\,
       \dd\nu_{\mathfrak K,u}^{g}.
 \label{eq:rectifiable-fixed-radius-local-liminf}
\end{equation}
For completeness, the geometric content of this invocation is the following.
In every relatively compact flow box, Fubini gives strong \(L^1\) convergence
on almost every one-dimensional slice.  Rectifiable coarea makes
\(\Sigma\) a finite labelled set on that localized slice.  The components of
the physical slice are independently extended by their two bank ghosts or by
the single outer ghost, and the weighted split-line inequality
\eqref{eq:global-resolved-sliced-liminf} gives the lower bound before
transverse integration.  The cofactor identity
\eqref{eq:global-resolved-coarea-factor} converts the endpoint counts into the
bank and outer-boundary fluxes, while
\eqref{eq:global-resolved-bulk-slicing} gives the resolved bulk variation.
A finite partition of unity on \(\{|v_i|>\delta\}\), followed by
\(\delta\downarrow0\), exhausts the flow boxes.  Thus
\eqref{eq:rectifiable-fixed-radius-local-liminf} is a global localized
inequality, not a fixed-chart assertion.

\emph{Step 4: radial Fatou and domination of Radon measures.}
Using \eqref{eq:rectifiable-interaction-measure},
\eqref{eq:rectifiable-interaction-measure-limit}, and Fatou once more gives
for every \(0\le\zeta\in C_c(U)\)
\begin{align}
 \int_U\zeta\,\dd\mu
 &=\lim_{n\to\infty}\int_a^b t f_n^\zeta(t)\lambda(\dd t)\notag\\
 &\ge\int_a^b t\liminf_{n\to\infty}f_n^\zeta(t)\lambda(\dd t)\notag\\
 &\ge\left(\int_a^bt\lambda(\dd t)\right)
       \int_{\overline\Omega}\zeta\,
       \dd\nu_{\mathfrak K,u}^{g}
 =\int_{\overline\Omega}\zeta\,\dd\nu_{\mathfrak K,u}^{g}.
 \label{eq:rectifiable-measure-domination-tests}
\end{align}
The positive-test characterization of Radon measures now yields
\begin{equation}
 \mu\ge\nu_{\mathfrak K,u}^{g}.
 \label{eq:rectifiable-measure-domination}
\end{equation}
Choosing a compactly supported cutoff equal to one on the common flow
neighbourhood and using the polar identity
\(\sum_i c_i|v_i\cdot D_{\mathfrak K}u|
=h_{\rm fl}(x,\sigma_{D_{\mathfrak K}u})|D_{\mathfrak K}u|\), we obtain
\begin{equation}
 \liminf_{n\to\infty}
 \mathcal H_{\eps_n,\Omega,\mathfrak K}^{g}(u_n)
 \ge\nu_{\mathfrak K,u}^{g}(\overline\Omega)
 =\mathcal H_{0,\Omega,\mathfrak K}^{g}(u).
 \label{eq:rectifiable-master-liminf-expanded}
\end{equation}

\emph{Step 5: countable exhaustion.}
If \(\mathcal J\) is countable, fix an enumeration and let \(\beta_m\) be
the sum of the two-bank measures in
\eqref{eq:global-resolved-limit-measure} over the first \(m\) labels.  On
almost every localized slice only finitely many points of \(\Sigma\) occur,
so its endpoint lower bound contains a finite sum even when the global label
set is countable.  After transverse integration, Tonelli's theorem gives
\begin{equation}
 \beta_m\uparrow\beta:=\sum_{j\in\mathcal J}\sum_i
 c_i|v_i\cdot\nu_j|
 \bigl(|T_j^+u-g_j^+|+|T_j^-u-g_j^-|\bigr)
 \Hh^{d-1}\mathbin{\llcorner}K_j.
 \label{eq:rectifiable-countable-bank-exhaustion}
\end{equation}
The total masses are bounded by a constant times
\(\sum_j\Hh^{d-1}(K_j)\).  Monotone convergence therefore passes from every
finite label family to the full bank measure in
\eqref{eq:rectifiable-measure-domination-tests}.  No diagonal depending on
the label number, the transverse slice, or the radial parameter is taken.

\emph{Step 6: constant recovery and absence of compulsory lower-dimensional
cost.}
Let \(u\in BV(\Omega;\{0,1\})\) have finite limit energy.  Its potential
vanishes.  For every fixed \(t\in(a,b)\), the fixed-state part of
Theorem~\ref{thm:global-resolved-translation-contact} gives
\begin{equation}
 \frac1{\eps}\mathcal B_{\eps t}(u)
 =t\frac1{\eps t}\mathcal B_{\eps t}(u)
 \longrightarrow t\nu_{\mathfrak K,u}^{g}(\overline\Omega).
 \label{eq:rectifiable-fixed-radius-recovery}
\end{equation}
For all sufficiently small \(\eps\),
\eqref{eq:global-resolved-translation-majorant} bounds the left-hand side by
\begin{equation}
 Ct\left(|Du|(\Omega)+\Hh^{d-1}(\partial\Omega)
               +\sum_j\Hh^{d-1}(K_j)\right),
 \label{eq:rectifiable-radial-recovery-majorant}
\end{equation}
which is integrable against \(\lambda\).  Dominated convergence and
\(\int_a^bt\lambda(\dd t)=1\) prove
\begin{equation}
 \mathcal H_{\eps,\Omega,\mathfrak K}^{g}(u)
 \longrightarrow\mathcal H_{0,\Omega,\mathfrak K}^{g}(u).
 \label{eq:rectifiable-constant-recovery-expanded}
\end{equation}
The same argument with an arbitrary test \(\zeta\) proves weak-star
convergence of the recovery interaction measures to
\(\nu_{\mathfrak K,u}^{g}\).  Surface-null branch, tip, and
cut--boundary incidence strata therefore carry no compulsory term in the
Gamma-functional and no mass for this recovery sequence.  This conclusion
does not preclude a nonnegative excess measure along other bounded-energy
sequences.
\end{proof}

\begin{corollary}[Periodic rectifiable resolved cut]
\label{cor:periodic-rectifiable-cut}
Let \(\Omega=\mathbb T^d\), assume
\eqref{eq:ray-radial-law}, \eqref{eq:flow-ray-regularity},
\eqref{eq:flow-ray-spanning}, and \eqref{eq:flow-ray-steepness}, and let
\(\mathfrak K\) satisfy \eqref{eq:rectifiable-cut-mass}.  Define the direct
energy by the first two cases of the typed rule, with no outer boundary.  Then
it is equicoercive and Gamma-converges to the boundary-free restriction of
\eqref{eq:rectifiable-direct-limit}, with constant recovery for every binary
finite-energy state.
\end{corollary}

\begin{proof}
Repeat the proof of Theorem~\ref{thm:rectifiable-cut-master-gamma} on the
periodic flow boxes.  There are no outer endpoints in the split-line theorem
and no boundary ghost comparison.  All remaining compactness, liminf, and
recovery estimates are unchanged.
\end{proof}

\begin{remark}[Why rectifiability remains]
\label{rem:rectifiability-structural}
The theorem removes smoothness, not geometry altogether.  Rectifiability
provides the approximate normal in \(h_{\rm fl}(x,\nu_j)\), the two one-sided
\(BV\) traces, and almost-everywhere discrete flow slices.  On a purely
\((d-1)\)-unrectifiable set these objects may all fail, so the candidate limit
would be a capacity or another nonlocal content rather than the Griffith
surface integral in \eqref{eq:rectifiable-direct-limit}.
\end{remark}

\subsection{BV on a finite resolved cut complex}
\label{subsec:resolved-cut-bv}

The slit and fan limits below live on completions that have one copy of
physical volume and several traces over a cut.  We now define their \(BV\)
space without referring to a preferred atlas.  This also records the precise
compactness and gluing facts used later.  Abstract metric definitions of
variation are available under broad hypotheses
\citep{miranda2003functions,ambrosio2014equivalent}; the finite resolved
geometry considered here permits the more concrete construction below and,
in particular, keeps track of which codimension-one faces are glued.

\begin{definition}[Finite resolved presentation]
\label{def:finite-resolved-presentation}
A finite resolved presentation of a measured completion \((X,m,p)\) consists
of bounded Lipschitz chambers \(C_1,\ldots,C_A\subset\Omega\) and canonical
injective interior embeddings \(\iota_a:C_a\to X\), with
\(p\circ\iota_a(x)=x\), such that:
\begin{enumerate}
 \item the sets \(p(\iota_a(C_a))\) are pairwise disjoint up to Lebesgue-null
 sets and cover the physical body up to a null set;
 \item every relatively open codimension-one chamber face is, up to an
 \(\Hh^{d-1}\)-null set, either an outer face, a resolved face, or one member
 of a paired family of transparent seams;
 \item the two members of each transparent pair are identified by the
 physical projection with opposite measure-theoretic normals, whereas no
 identification is imposed between the banks of a resolved face.
\end{enumerate}
Artificial subdivisions of a chamber create transparent seams.  A refinement
is obtained by finitely subdividing chambers and declaring every new internal
face transparent.
 We call the presentation \emph{regular} if its faces belong to a finite
piecewise-\(C^{1,1}\) stratification compatible with the prescribed physical
cut stratification.  The physical faces retain their designated slit or fan
incidence, including the meeting of arbitrarily many labelled fan faces along
one codimension-two front.  Auxiliary faces must meet one another and every
physical stratum transversely, with the expected codimension, and all nonzero
auxiliary incidence angles and all sector angles are bounded below.  Thus its
chambers are Lipschitz domains with corners, with constants uniform over the
finite family.  In particular, regularity does not place the physical sheets
of a fan in general position.
\end{definition}

\begin{lemma}[Regular presentations and transverse refinement]
\label{lem:resolved-common-refinement}
The curved-slit and finite sectorial completions introduced below admit
regular finite resolved presentations.  Two such presentations possess a
common regular finite Lipschitz refinement whenever their distinct auxiliary
strata meet transversely.  Every fixed regular presentation admits arbitrary
finite transverse refinements, which introduce only transparent faces.
\end{lemma}

\begin{proof}
In an ordinary chart there is no physical cut.  In a slit chart the prescribed
physical strata are the two banks and their common front, and one chooses an
auxiliary radial face; in a fan chart they are the finitely many labelled rays
and their common axis, and one chooses an auxiliary face in the interior of
each sector.  Subdivide the
tangential chart variables into finitely many boxes compactly contained in
the larger chart and use the radial or sectorial variables for the remaining
faces.  Positive reach in the slit case and the angle gap
\eqref{eq:fan-angle-gap} in the fan case give uniform Lipschitz constants.
This constructs a regular presentation.

For two presentations satisfying the stated transversality condition, overlay
their finite auxiliary families while keeping the prescribed physical
stratification and every labelled physical face fixed.  A finite transverse piecewise-
\(C^{1,1}\) arrangement has finitely many strata; its connected top-dimensional
components are Lipschitz domains with corners.  Compactness of the finitely
many stratum closures gives one positive lower bound for all nonzero incidence
angles.  Label every overlaid physical face by its original outer or resolved
label and every other face transparent.  The connected components, with
these labels, form a common regular finite Lipschitz refinement.  Auxiliary
faces introduced in a further subdivision are transparent by definition.  If
a proposed finite auxiliary family already satisfies the transversality
requirement in the statement, overlay its members one at a time.  At every
step the same finite-arrangement argument applies, and a finite induction
produces the asserted transverse refinement.  No generic perturbation or
unproved transversality selection is used.
\end{proof}

Let \(D\subset\Omega\) be the physical open set before metric completion and
let \(\jmath:D\to X^\circ\) be its canonical interior embedding.  Thus
\(D=\Omega\setminus\mathfrak C\) for a resolved cut complex.  We make the
atlas-free definitions
\begin{equation}
 BV_{\rm res}(X):=\{u\in L^1(X,m):u\circ\jmath\in BV(D)\},
 \qquad
 D_Xu:=\jmath_\#D(u\circ\jmath).
 \label{eq:resolved-bv-definition}
\end{equation}
For a finite presentation, write \(u_a=u\circ\iota_a\).  Let
\(\mathcal T\) contain each transparent pair once, written
\((a,b,S)\), and orient its physical face by the outer normal \(\nu_a\)
of \(C_a\).  Repeated \(BV\) gluing gives the equivalent chamber formula
\begin{align}
 D_Xu={}&\sum_{a=1}^A(\iota_a)_\#
       (Du_a\mathbin{\llcorner}C_a)\notag\\
 &+\sum_{(a,b,S)\in\mathcal T}
   (T_au_a-T_bu_b)\nu_a\,
   \Hh^{d-1}\mathbin{\llcorner}S .
 \label{eq:resolved-derivative-definition}
\end{align}
In the second line the two transparent copies of \(S\) are identified in
\(X\); equivalently, the measure is pushed to their common copy.  Thus an
actual jump across an artificial subdivision is retained.  No analogous term
is inserted between the two banks of a resolved face.
The one-sided trace on an outer or resolved face is the ordinary Lipschitz
trace of the corresponding chamber representative.  In particular, opposite
banks of one resolved sheet have independent traces.  We use
\(BV_{\rm res}(X;\{0,1\})\) for the binary subspace.

\begin{proposition}[Presentation independence, traces, compactness, and gluing]
\label{prop:resolved-bv-package}
Let \((X,m,p)\) arise by resolving a finite cut geometry and admit a finite
resolved presentation.  Then the intrinsic definition
\eqref{eq:resolved-bv-definition} is equivalent to the chamber condition
\(u_a\in BV(C_a)\) for every \(a\), and every finite presentation computes
the same vector measure through
\eqref{eq:resolved-derivative-definition}.  No common-refinement hypothesis
is required.  More precisely:
\begin{enumerate}
 \item every outer or resolved face has a bounded trace map
 \[
  T_F:BV_{\rm res}(X)\longrightarrow L^1(F),\qquad
  \|T_Fu\|_{L^1(F)}\le C_F(\|u\|_{L^1(X)}+|D_Xu|(X));
 \]
  \item if \(u_n\) is bounded in \(L^1(X)\) and in resolved variation, then a
 subsequence converges in \(L^1(X)\) to some \(u\in BV_{\rm res}(X)\), and
 \[
  |D_Xu|(X)\le\liminf_n|D_Xu_n|(X);
 \]
 \item if \(h:\overline\Omega\times\R^d\to[0,\infty)\) is continuous,
 convex and positively one-homogeneous in its second variable, then
 \begin{equation}
 \int_X h(p,\sigma_{D_Xu})\dd|D_Xu|
  \le\liminf_n\int_Xh(p,\sigma_{D_Xu_n})\dd|D_Xu_n|;
  \label{eq:resolved-anisotropic-lsc}
 \end{equation}
 \item if, in addition, \(h(x,\xi)=h(x,-\xi)\), \(F\) is a fixed outer or
 resolved face, and \(g_F\) is the trace of a fixed \(BV\) ghost on the
 opposite collar, then the bulk-plus-contact
 functional
 \begin{equation}
  u\longmapsto \int_Xh(p,\sigma_{D_Xu})\dd|D_Xu|
   +\int_Fh(x,\nu_F)|T_Fu-g_F|\dd\Hh^{d-1}
  \label{eq:resolved-contact-lsc}
 \end{equation}
 is lower semicontinuous under \(L^1\) convergence with bounded resolved
 variation;
 \item for two adjacent chambers \(C_a,C_b\) with common Lipschitz face
 \(S\), the physical gluing \(w=u_a\one_{C_a}+u_b\one_{C_b}\) satisfies
 \begin{equation}
  Dw=Du_a\mathbin{\llcorner}C_a+Du_b\mathbin{\llcorner}C_b
   +(T_au_a-T_bu_b)\nu_a\,
     \Hh^{d-1}\mathbin{\llcorner}S.
  \label{eq:resolved-bv-gluing}
 \end{equation}
  Hence a transparent seam contributes no derivative exactly when its two
  traces agree; otherwise its trace jump is part of \(D_Xu\).  A resolved face
  is not glued and contributes no intrinsic jump even when its two bank traces
  differ.
\end{enumerate}
The same conclusions hold locally and for finite unions of face patches.
Under chamberwise strict \(BV\) convergence, all fixed-face traces converge in
\(L^1\).
\end{proposition}

\begin{proof}
On one presentation, the trace estimate is the standard trace theorem on
each Lipschitz chamber.  Formula \eqref{eq:resolved-bv-gluing} is the ordinary
\(BV\) product and gluing formula, first for smooth functions and then by
strict approximation \citep{ambrosio2000functions}.  Iterating it over the
finite transparent adjacency graph shows both implications
\[
 u\circ\jmath\in BV(D)
 \quad\Longleftrightarrow\quad
 u_a\in BV(C_a)\ \text{for every }a,
\]
and identifies \(\jmath_\#D(u\circ\jmath)\) with the right-hand side of
\eqref{eq:resolved-derivative-definition}.  Since the distributional
derivative on the fixed open set \(D\) is unique, this proves presentation
independence even when two auxiliary stratifications have no common
Lipschitz refinement.
For completeness, consider one refinement step which divides \(C\) into
\(C^+\) and \(C^-\) along \(S\), and orient \(S\) by the outer normal
\(\nu_+\) of \(C^+\).  The exact identity
\begin{equation}
 Du\mathbin{\llcorner}C
 =Du^+\mathbin{\llcorner}C^+
  +Du^-\mathbin{\llcorner}C^-
  +(T_+u^+-T_-u^-)\nu_+\,
    \Hh^{d-1}\mathbin{\llcorner}S
 \label{eq:one-step-transparent-refinement}
\end{equation}
holds after all three measures are pushed to the old chamber.  The last term
is exactly the transparent-seam term inserted by
\eqref{eq:resolved-derivative-definition}.  Hence the refined and unrefined
vector measures are equal, not merely equal in total mass.  In particular,
for every continuous convex positively one-homogeneous \(h\),
\begin{equation}
 \int h(p,\sigma_{D_Xu})\dd|D_Xu|
 =\int h(p,\sigma_{D_X^{\rm ref}u})\dd|D_X^{\rm ref}u|.
 \label{eq:anisotropic-refinement-invariance}
\end{equation}
Thus \eqref{eq:anisotropic-refinement-invariance} also gives the explicit
one-step verification of the intrinsic argument.

For compactness, the first line of
\eqref{eq:resolved-derivative-definition} bounds the variation in every
chamber.  Apply Euclidean \(BV\) compactness chamber by chamber and use one
finite diagonal.  The trace-jump measures on transparent seams are exactly
the boundary parts in the glued representatives and hence part of the unique
distributional derivative on \(D\).  Euclidean lower semicontinuity therefore
retains, rather than deletes, their limiting contribution and proves the
scalar lower bound.  Applying Reshetnyak lower semicontinuity on \(D\) gives
\eqref{eq:resolved-anisotropic-lsc}, including every transparent-seam polar.
To prove \eqref{eq:resolved-contact-lsc}, glue \(u_n\) to the fixed ghost across
\(F\).  Formula \eqref{eq:resolved-bv-gluing} identifies the jump part of the
glued anisotropic variation with the displayed contact term; ordinary lower
semicontinuity and subtraction of the fixed ghost contribution give the
claim.  Finally, chamberwise strict convergence and the Lipschitz trace
theorem give the stated trace convergence.
\end{proof}

\begin{remark}[The slit and fan presentations]
\label{rem:slit-fan-resolved-presentations}
For a curved slit, cut each front chart along one auxiliary radial ray; the
two resulting angular chambers are Lipschitz, the auxiliary ray is
transparent, and the two banks are resolved faces.  For a sectorial junction,
 subdivide each sector wedge along one interior bisector.  The bisectors and
 ordinary chart boundaries are transparent, while the sheet banks are
 resolved.  Positive reach, the finite atlas, and the angle gap verify the
 hypotheses of Lemma~\ref{lem:resolved-common-refinement}.  Transverse choices
have a common regular refinement, while arbitrary choices still compute the
same distributional derivative on the physical open set by
Proposition~\ref{prop:resolved-bv-package}.  Thus both completions below are
independent of every auxiliary radial or chart seam.
\end{remark}

\begin{lemma}[Finite disjoint resolved-cell localization]
\label{lem:resolved-cell-localization}
Let \(X\) be either regular resolved completion in
Remark~\ref{rem:slit-fan-resolved-presentations}.  Denote by \(\mathcal F\)
its finite labelled family of outer and resolved faces, by \(\Sigma\) the
codimension-two physical strata, and by \(\mathcal S\) the auxiliary
transparent seams.  For \(u\in BV_{\rm res}(X;\{0,1\})\), set
\begin{align}
 \mathfrak m_u:={}&
 h_{\rm fl}(p,\sigma_{D_Xu})|D_Xu|\notag\\
 &+\sum_{F\in\mathcal F}
 h_{\rm fl}(\cdot,\nu_F)|T_Fu-g_F|\,
 \Hh^{d-1}\mathbin{\llcorner}F ,
 \label{eq:resolved-sharp-measure}
\end{align}
where the face measures and the tubular neighbourhoods are understood on
their labelled copies.  For almost every \(r,\eta>0\) and every
\(\delta>0\), there is a finite family of
pairwise disjoint resolved cells, each compactly contained in a slightly
larger flow box, with the following properties.  Here
\(\mathcal E_\eps\) may be any nonnegative flow-bond energy whose detector is
the ordinary endpoint rule when no resolved crossing occurs and the
endpoint-to-own-face rule on a stable resolved crossing, and whose ordinary and
ambient-completed cell energies are those treated above.
\begin{enumerate}
 \item Every cell is of exactly one type: an interior cell, an outer-boundary
 cell, or a one-sided cell meeting one labelled resolved face and its fixed
 opposite-collar ghost.  The larger boxes avoid
  \(T_r(\Sigma)\cup T_\eta(\mathcal S)\).
 \item The cells exhaust the retained sharp measure:
 \begin{equation}
  \mathfrak m_u\!\left(
   (X\sqcup\mathcal F)\setminus
   \bigl(T_r(\Sigma)\cup T_\eta(\mathcal S)
         \cup\textstyle\bigcup_j C_j\bigr)\right)<\delta .
  \label{eq:resolved-cell-exhaustion}
 \end{equation}
 \item For all sufficiently small \(\eps\), the interaction pairs whose full
 flow arcs remain in the corresponding larger box can be assigned to
 pairwise disjoint measurable pair sets \(\mathcal P_{j,\eps}\).  On an
 interior cell the restricted interaction is the ordinary ambient
 interaction.  On a face cell, after adding the fixed ghost--ghost
 interaction, it is the ambient crossing-bond interaction.
 \item Consequently, every bounded-energy sequence \(u_\eps\to u\) satisfies
 \begin{equation}
  \liminf_{\eps\downarrow0}\mathcal E_\eps(u_\eps)
  \ge\sum_j\mathfrak m_u(C_j),
  \label{eq:resolved-cell-finite-liminf}
 \end{equation}
 for either the slit or fan detector on the retained region.  For the constant
 binary recovery sequence, the interaction on the same region converges to
 \(\mathfrak m_u\) there; only the quantified singular tubes and auxiliary
 seam strips remain.
\end{enumerate}
\end{lemma}

\begin{proof}
Choose \(r\) and \(\eta\) so that the boundaries of the two deleted
neighbourhoods have zero \(\mathfrak m_u\)-mass; only countably many radii are
excluded.  In every regular chart, take a rectangular grid in tangential
coordinates and a one-sided rectangular grid in the normal coordinate at a
face.  Almost every translation of each grid has zero
\(\mathfrak m_u\)-mass on all grid faces, by Fubini.  Refine the finitely many
grids on chart overlaps using Lemma~\ref{lem:resolved-common-refinement}.
Give a cell to the first chart in a fixed finite ordering and remove a thin
collar of the resulting cell boundaries.  The cells are then pairwise
disjoint, their enlarged boxes remain in the retained region, and each
enlarged box meets at most one labelled face.  First choosing the grids fine
enough and then shrinking the collars, Radon inner regularity and continuity
from above give \eqref{eq:resolved-cell-exhaustion}.  At a resolved face the
two labelled banks are treated as different copies, so the construction never
identifies their traces.

The flow displacement is at most
\(\eps bV_0e^{\eps bV_1}\).  Hence, for fixed cells and sufficiently small
\(\eps\), every arc issued from the inner part of a cell either remains in
its enlarged box or belongs to the discarded boundary collar.  Regard the
symmetrized interaction as a measure on unordered arcs.  Assign an arc to
the first cell whose inner part contains one endpoint and whose enlarged box
contains the full arc.  The resulting measurable sets
\(\mathcal P_{j,\eps}\) are disjoint.  Restricting a nonnegative energy to
their union can only decrease it.

In an interior box there is no resolved crossing.  In a one-sided face box
there is exactly one stable resolved sheet and the detector is the endpoint-to-own-face
rule.  Gluing the physical field to its prescribed fixed ghost and adding the
ghost--ghost bonds therefore gives exactly the ambient energy in the larger
box.  The ordinary interior Gamma-liminf or
Theorem~\ref{thm:boundary-flow-gamma-limit}, followed by subtraction of the
convergent fixed ghost contribution, yields \(\mathfrak m_u(C_j)\).
There are finitely many cells, so the liminf of their sum is at least the sum
of their liminfs; this proves
\eqref{eq:resolved-cell-finite-liminf}.  Letting \(\delta\downarrow0\), and
then using continuity from below as \(\eta\downarrow0\) and \(r\downarrow0\),
recovers the complete sharp measure because
\(\mathfrak m_u(\Sigma)=0\).

For the upper bound, glue chamberwise to the prescribed neighbouring chamber
or ghost and apply the \(BV\) flow-translation formula to these finitely many
global glued fields.  Formula
\eqref{eq:one-step-transparent-refinement} recombines all auxiliary-seam
interior and trace-jump terms into \(D_Xu\) exactly, while subtraction of the
fixed ghost--ghost terms leaves the labelled contact measures.  Thus the retained interactions
converge to \(\mathfrak m_u\) without a patchwise overlap error.  The only
unaccounted interactions have an endpoint or arc in the deleted singular
tubes or seam strips, as asserted.
\end{proof}

\subsection{A two-face theorem on an intrinsic curved-slit completion}
\label{subsec:lifted-cut-domain}

We next remove the global graph and planar-chord assumptions.  Let
\(\Omega\subset\R^d\), \(d\ge2\), be a bounded connected \(C^{1,1}\) domain.
Let \(K\Subset\Omega\) be a compact, embedded, orientable \(C^{1,1}\)
hypersurface-with-boundary, with nonempty \(C^{1,1}\) front
\(\Gamma:=\partial K\).  Fix the continuous orientation \(\nu_K\).
We assume that there is a geometric scale \(\rho_*>0\) such that
\begin{equation}
 \operatorname{dist}(K,\partial\Omega)\ge4\rho_*,
 \qquad \operatorname{reach}(K),\operatorname{reach}(\Gamma)\ge4\rho_*,
 \qquad \operatorname{Lip}(\nu_K)\le\rho_*^{-1}.
 \label{eq:graph-slit-geometry}
\end{equation}
The reach inequalities are understood in the ambient Euclidean metric.
They exclude self-contact and give unique nearest points wherever they are
used below.  These hypotheses, rather than a hidden global graph, are the
proved compatibility class.

\paragraph{Intrinsic completion.}
Give \(\Omega\setminus K\) its intrinsic length distance \(d_{\Omega\setminus
K}\), and set
\begin{equation}
 X_K:=\overline{(\Omega\setminus K,d_{\Omega\setminus K})},\qquad
 p:X_K\longrightarrow\overline\Omega,\qquad
 m_K:=\jmath_\#\bigl(\mathcal L^d\mathbin{\llcorner}(\Omega\setminus K)\bigr),
 \label{eq:slit-single-volume-measure}
\end{equation}
where \(\jmath:\Omega\setminus K\to X_K\) is the canonical isometric
embedding.  Equivalently, \(p\) is one-to-one \(m_K\)-a.e. and
\(p_\#m_K=\mathcal L^d\llcorner\Omega\); all completion points have zero
\(m_K\)-mass.  The projection has
one preimage over \(\Omega\setminus K\), two preimages \(q^\pm\) over
\(q\in K^\circ:=K\setminus\Gamma\), and one preimage over \(q\in\Gamma\).
Thus volume is not doubled, the two open faces are, and the faces merge at
the front.

For completeness, this description is intrinsic.  On
\(K_{2r}:=\{q\in K:\operatorname{dist}_K(q,\Gamma)>2r\}\), the normal map
\begin{equation}
 F(q,s)=q+s\nu_K(q),\qquad |s|<2r<\rho_*,
 \label{eq:slit-tip-chart}
\end{equation}
is injective, \(s\) is the oriented signed tubular coordinate, and
\[
 |\det DF(q,s)-1|\le C|s|/\rho_*.
\]
At a front point, a \(C^{1,1}\) chart identifies \(X_K\) with
\((z,\varrho,\theta)\), \(0\le\varrho<r\), \(-\pi\le\theta\le\pi\):
\(\theta=\pm\pi\), \(\varrho>0\), are the banks and all
\((z,0,\theta)\) are one point.  Its metric and volume Jacobians differ
from the flat half-plane model by at most \(C\varrho/\rho_*\).
The positive reach and compactness supply finitely many such oriented
charts with a common radius.  On an overlap their signed coordinates have
the same sign, because both orientations induce \(\nu_K\); hence the bank
labels and the quotient are independent of the atlas.

Let \(g_\partial\in BV(\partial\Omega;\{0,1\})\) have the fixed ambient
extension \(g^{\rm e}\) used in Lemma~\ref{lem:lipschitz-boundary-package}.
Let \(g^\pm\in BV(K;\{0,1\})\).  In the two normal collars choose binary
ghost fields \(G^\pm\) with traces \(g^\pm\), where \(G^+\) is defined on
the negative collar and \(G^-\) on the positive collar.  They are prescribed
ghosts, not positive-volume copies of the unknown.
Such ghosts always exist in the present class: positive reach makes each
truncated one-sided tubular collar a bounded Lipschitz domain.  Apply the
surjectivity and binary level-set construction of
Lemma~\ref{lem:lipschitz-boundary-package} with datum \(g^\pm\) on its base
and zero on the remaining collar boundary.  No compatibility is required at
\(\Gamma\), which is \(\Hh^{d-1}\)-null and is removed by the shrinking gate.

\paragraph{Intrinsic finite-\(\eps\) crossing detector.}
For a flow bond \(y=\Phi_i^{-\eps t}(x)\), write
\[
 \gamma_{i,x,t}^\eps(s):=\Phi_i^{-s}(x),\qquad0\le s\le\eps t,
 \qquad r_\eps:=\sqrt\eps .
 \label{eq:slit-tip-gate-width}
\]
For \(\eps b<\rho_*/8\), any such arc meeting \(K\) and staying at distance
\(r_\eps\) from \(\Gamma\) lies in one oriented tubular chart.  Define
\begin{equation}
 \mathfrak c_{\eps,i,t}(x):=
 \begin{cases}
  1,&\operatorname{dist}(\gamma_{i,x,t}^\eps,\Gamma)>r_\eps
  \ \text{and}\ I_2(\gamma_{i,x,t}^\eps,K)=1,\\
  0,&\text{otherwise}.
 \end{cases}
 \label{eq:slit-crossing-detector}
\end{equation}
Here \(I_2\) is the mod-two intersection number relative to the endpoints.
In any oriented tubular chart it is exactly
\[
 I_2(\gamma,K)=
 \frac{1-\operatorname{sgn}s(\gamma(0))
             \operatorname{sgn}s(\gamma(\eps t))}{2}.
\]
This formula also defines the value at a tangential contact by stable
one-sided approximation.  It is unchanged on chart overlaps and under
perturbations of the flow arc that avoid \(K\cup T_{r_\eps}(\Gamma)\).
In particular it detects an odd crossing, ignores a tangential touch or an
even recrossing, and never uses a global chord intersection.  If
\(\mathfrak c_{\eps,i,t}(x)=1\), let
\(\sigma_x,\sigma_y\in\{+,-\}\) be the oriented banks containing \(x,y\);
then \(\sigma_y=-\sigma_x\).

For the single physical field \(u\in L^1(X_K,m_K)=L^1(\Omega)\), glued to
\(g^{\rm e}\) outside \(\Omega\) and denoted by \(\bar u^{\rm p}\), define
on the support of the \(i,t\) flow pair
\begin{equation}
 \Psi_{\eps,K}^{i,t}[u](x,y):=
 \begin{cases}
 |u(x)-G^{\sigma_x}(y)|^2+
 |u(y)-G^{\sigma_y}(x)|^2,
 &x,y\in\Omega,\ \mathfrak c_{\eps,i,t}(x)=1,\\
 |\bar u^{\rm p}(x)-\bar u^{\rm p}(y)|^2,
 &\text{otherwise}.
 \end{cases}
 \label{eq:slit-finite-epsilon-pairing}
\end{equation}
The transposed half of \(\widehat Q_{i,\eps,t}\) uses the transposed
expression, so \(\Psi\) is pair-symmetric.  The first line replaces one
physical odd crossing by the two endpoint-to-own-face contacts.  The second
line retains ordinary physical pairing around the unresolved front and
along every noncrossing or even-crossing arc.

With \(B_\Omega\) from \eqref{eq:crossing-bond-set}, set
\begin{align}
 \mathcal G_{\eps,K}^{g_\partial,g^+,g^-}(u)
 :={}&\frac M\eps\int_\Omega d_*(u)\dd x\notag\\
 &+\frac1\eps\sum_i\int_a^b
 \iint_{B_\Omega}\Psi_{\eps,K}^{i,t}[u](x,y)
 \widehat Q_{i,\eps,t}(\dd x,\dd y)\lambda(\dd t).
 \label{eq:slit-prelimit-energy}
\end{align}

\begin{lemma}[Detector consistency, two face factors, and front estimate]
\label{lem:slit-face-tip-estimate}
Let \(K_r=\{q\in K:\operatorname{dist}_K(q,\Gamma)>r\}\).
There are \(C,r_0>0\), depending only on \(d,\rho_*,K,\Gamma\) and the
fixed flow bounds, such that, for \(0<r<r_0\),
\begin{align}
 |T_r(\Gamma)|&\le Cr^2\Hh^{d-2}(\Gamma),
 \label{eq:slit-full-tip-tube-volume}\\
 \mathcal L^d(A_{i,\eps,t}(r))
 &\le C\eps t\bigl(r+\eps t\bigr)\Hh^{d-2}(\Gamma),
 \label{eq:slit-active-tip-volume}\\
 \frac1\eps\widehat Q_{i,\eps,t}
 \bigl(A_{i,\eps,t}(r)\times U\cup U\times A_{i,\eps,t}(r)\bigr)
 &\le Ct(r+\eps t)\Hh^{d-2}(\Gamma).
 \label{eq:slit-normalized-tip-bound}
\end{align}
Here \(A_{i,\eps,t}(r)\) consists of endpoints whose flow arc can change
classification when crossings within intrinsic distance \(r\) of \(\Gamma\)
are left unresolved.  After integration in \(t\), the normalized bound is
\(C(r+\eps)\Hh^{d-2}(\Gamma)\).  On \(K_r\), each ghost summand has contact
coefficient one.

Moreover, in a tubular box of diameter \(\ell\le\rho_*/8\),
\begin{align}
 \left|s(\Phi_i^{-\tau}(x))
 -s(x)+\tau v_i(\pi_Kx)\cdot\nu_K(\pi_Kx)\right|
 &\le C\tau^2\bigl(\rho_*^{-1}+\operatorname{Lip}v_i\bigr),\notag\\
 |J\Phi_i^{-\tau}(x)-1|
 &\le C\tau\operatorname{Lip}v_i ,
 \label{eq:slit-curvature-flow-estimate}
\end{align}
and the normalized freezing error on \(K_r\) is at most
\[
 C\eps\int_a^b(t^2+t^3)\lambda(\dd t)
 \bigl(\rho_*^{-1}+\max_i\operatorname{Lip}v_i\bigr)
 \Hh^{d-1}(K_r).
\]
If a proof uses a finite hard partition of the oriented atlas and
\(\Sigma\) is the union of its \(C^{1,1}\) seams, deleting an
\(\eta\)-strip around \(\Sigma\) costs at most
\begin{equation}
 C(\eta+\eps)\Hh^{d-2}(\Sigma).
 \label{eq:slit-chart-seam-estimate}
\end{equation}
The detector itself has zero seam discrepancy; \eqref{eq:slit-chart-seam-estimate}
is only a localization estimate.  These quantitative bounds concern the
geometry of bounded binary detector integrands.  No \(O(\eps)\) rate is
asserted for the \(BV\) flow-translation limit itself.
\end{lemma}

\begin{proof}
The reach bound gives front normal coordinates with two transverse
variables and Jacobian \(1+O(r/\rho_*)\), proving
\eqref{eq:slit-full-tip-tube-volume}.  In a flat front box, prescribe the
crossing point on \(K\).  The admissible starting interval along the flow
has length at most \(C\eps t\), while the strip of crossing points within
distance \(r\) of \(\Gamma\) has area at most
\(Cr\Hh^{d-2}(\Gamma)\).  The signed-coordinate expansion in
\eqref{eq:slit-curvature-flow-estimate} moves the crossing point by
\(C\eps t\) and thickens the normal slab by the same amount.  The tubular
Jacobian is \(1+O(\eps t/\rho_*)\).  Fubini therefore gives
\eqref{eq:slit-active-tip-volume}; bounded \(c_i\), the two symmetric
marginals, and the finite \(t\)-moments give
\eqref{eq:slit-normalized-tip-bound}.

The first line of \eqref{eq:slit-curvature-flow-estimate} follows by applying
the \(C^{1,1}\) chain rule to \(s\circ\Phi_i^{-\tau}\); its derivative is
\(-\nabla s\cdot v_i\), and both \(\nabla s\) and \(v_i\) are Lipschitz.
The second is the standard flow-Jacobian estimate.  Division by \(\eps\)
leaves the displayed \(O(\eps/\rho_*)\) energy error.

On oriented overlaps, two signed defining functions equal a positive
Lipschitz factor times one another.  Their endpoint signs and hence \(I_2\)
agree exactly.  For a hard partition, the same two-normal-variable Fubini
argument applied to the codimension-two seam gives
\eqref{eq:slit-chart-seam-estimate}; first \(\eps\downarrow0\), then
\(\eta\downarrow0\).  Finally, away from the front, flattening leaves one
endpoint and its own-face ghost on opposite half-spaces.  Lemma
\ref{lem:boundary-halfspace-cell} gives \(1/2+1/2=1\) for that summand.
The other endpoint gives a second, independent coefficient one; the two
face contacts are not averaged.
\end{proof}

Use the finite resolved presentation in
Remark~\ref{rem:slit-fan-resolved-presentations} and write
\(BV(X_K^\circ):=BV_{\rm res}(X_K^\circ)\).  By
Proposition~\ref{prop:resolved-bv-package}, this space, its derivative
\(D_{X_K}u\), and its one-sided traces are independent of the front atlas and
of every auxiliary radial seam.  The derivative includes every physical jump
in \(\Omega\setminus K\) and the continuation around \(\Gamma\), but it does
not count a difference between \(T_K^+u\) and \(T_K^-u\).  Define
\begin{align}
 \mathcal G_{0,K}^{g_\partial,g^+,g^-}(u)
 :={}&\int_{X_K^\circ}h_{\rm fl}(p(x),\sigma_{D_{X_K}u})
 \dd|D_{X_K}u|(x)\notag\\
 &+\int_{\partial\Omega}h_{\rm fl}(x,\nu_\Omega)
 |T_\partial u-g_\partial|\dd\Hh^{d-1}\notag\\
 &+\int_Kh_{\rm fl}(x,\nu_K)
 \bigl(|T_K^+u-g^+|+|T_K^-u-g^-|\bigr)\dd\Hh^{d-1},
 \label{eq:slit-limit-energy}
\end{align}
with value \(+\infty\) outside \(BV(X_K^\circ;\{0,1\})\).  Since
\(h_{\rm fl}(x,\xi)=h_{\rm fl}(x,-\xi)\), the chosen orientation only names
the faces.  The candidate Gamma-functional contains no compulsory term on
\(\Gamma\); this statement does not exclude excess concentration along a
non-recovery sequence.

\begin{theorem}[Gamma-limit on a curved intrinsic slit]
\label{thm:lifted-cut-gamma-limit}
Assume \eqref{eq:ray-radial-law}, \eqref{eq:boundary-ambient-frame},
\eqref{eq:flow-ray-steepness}, and the geometric class
\eqref{eq:graph-slit-geometry}.  Then
\(\mathcal G_{\eps,K}^{g_\partial,g^+,g^-}\) Gamma-converges in
\begin{equation}
 L^1(X_K,m_K)=L^1(\Omega)
 \label{eq:slit-product-topology}
\end{equation}
to \(\mathcal G_{0,K}^{g_\partial,g^+,g^-}\).  It is equicoercive: every
bounded-energy sequence has, after threshold rounding and extraction,
\begin{equation}
 u_\eps\to u\quad\text{strongly in }L^1(X_K,m_K),\qquad
 u\in BV(X_K^\circ;\{0,1\}).
 \label{eq:slit-product-compactness}
\end{equation}
Every admissible binary \(u\) is recovered by the constant sequence.
\end{theorem}

\begin{proof}
\emph{Compactness.}
Threshold rounding is unchanged because the potential is integrated once
over physical volume and the pair measure has the same marginal bound.
Choose a finite collection of ordinary boxes, one-sided face boxes, and
front boxes with radii below \(\rho_*/8\).  On boxes compactly contained in
one bank, sufficiently short bonds are ordinary bonds; ordered-flow
smoothing and spanning give uniform \(BV\) bounds.  In a one-sided face box
whose closure avoids \(\Gamma\), glue the rounded physical field to the
fixed opposite-collar ghost across \(K\).  The resolved-crossing detector and pair
symmetry identify its energy, after adding the fixed ghost--ghost bonds, with
the ambient crossing-bond energy of Lemma
\ref{lem:boundary-ambient-completion}.  The added term is independent of the
sequence and uniformly bounded.  Theorem
\ref{thm:boundary-flow-gamma-limit} therefore gives compactness of the whole
glued field in the enlarged Lipschitz box.  We do not invoke compactness of the
trace operator.  The trace appearing in the limit is identified instead by
the \(BV\) gluing formula and lower semicontinuity of the glued derivative.
Applying this on a countable exhaustion of every face box and then taking one
diagonal closes the possible boundary layer at \(K\).

In a front box, cut
along one radial ray, apply the same estimate on the resulting Lipschitz
box, and let the omitted angular strip shrink.  Overlap signs agree, so the
local limits glue to one \(BV(X_K^\circ)\) function.  Finite-atlas
compactness and the vanishing volume of the omitted strips prove
\eqref{eq:slit-product-compactness}.  Equivalently, after projection the
only uncontrolled Euclidean jump lies on the fixed \(K\), whose variation
is bounded by \(2\Hh^{d-1}(K)\); ordinary \(BV(\Omega)\) compactness then
gives the same strong \(L^1\) conclusion.

\emph{Liminf.}
Fix one regular presentation, denote its auxiliary seam family by
\(\mathcal S\), choose regular \(r,\eta>0\) and \(\delta>0\), and apply
Lemma~\ref{lem:resolved-cell-localization}.  On its finite disjoint family,
ordinary cells give the intrinsic variation and the outer contact, while each
one-sided face cell gives one bank contact after ambient completion and
subtraction of the convergent fixed ghost term.  Thus
\eqref{eq:resolved-cell-finite-liminf} and
\eqref{eq:resolved-cell-exhaustion} yield
\begin{align}
 \liminf_{\eps\downarrow0}
 \mathcal G_{\eps,K}^{g_\partial,g^+,g^-}(u_\eps)
 \ge{}&
 \mathfrak m_u\bigl(
 (X_K^\circ\sqcup\partial\Omega\sqcup K^+\sqcup K^-)
 \setminus(T_r(\Gamma)\cup T_\eta(\mathcal S))\bigr)-\delta .
 \label{eq:slit-localized-liminf}
\end{align}
The coefficient of each bank is one by
Lemma~\ref{lem:boundary-halfspace-cell}; the two labelled copies are never
averaged.  The curvature and flow-freezing error is
\(O(\eps/\rho_*)\) by \eqref{eq:slit-curvature-flow-estimate}, and the seam
collar error is \(O(\eta+\eps)\) by
\eqref{eq:slit-chart-seam-estimate}.  Let first \(\eps\downarrow0\), then
\(\delta\downarrow0\), \(\eta\downarrow0\), and \(r\downarrow0\).
Continuity from below gives the full lower bound.  Indeed
\(|D_{X_K}u|\ll\Hh^{d-1}\) in the \(BV\)-null-set sense and
\(\Hh^{d-1}(\Gamma)=0\), while the face area in \(T_r(\Gamma)\) tends to zero.

\emph{Recovery.}
For binary \(u\in BV(X_K^\circ)\), use \(u_\eps=u\).  The upper-bound part of
Lemma~\ref{lem:resolved-cell-localization}, equivalently the global \(BV\)
flow-translation formula for the finitely many glued physical/ghost fields,
gives the intrinsic bulk, outer contact, and both face contacts outside
\(T_r(\Gamma)\) and the optional seam strips.  Identity
\eqref{eq:one-step-transparent-refinement} recombines their interior and
trace-jump terms into \(D_{X_K}u\), and subtracting the fixed ghost--ghost
translations leaves exactly the two own-face terms.  Formula
\eqref{eq:slit-curvature-flow-estimate} bounds all curvature, coefficient,
and Jacobian remainders by \(C\eps/\rho_*\).  The only change caused by the
front gate is supported on \(A_{i,\eps,t}(r)\); binary values and
\eqref{eq:slit-normalized-tip-bound} give
\begin{equation}
 \limsup_{\eps\downarrow0}|R_{\eps,r}|
 \le Cr\Hh^{d-2}(\Gamma).
 \label{eq:slit-recovery-tip-remainder}
\end{equation}
The intrinsic variation in \(T_r(\Gamma)\) tends to zero, and the two face
areas there are \(O(r\Hh^{d-2}(\Gamma))\).  Optional seams contribute
\(O(\eta)\).  Letting \(\eps\downarrow0\), then \(\eta\downarrow0\), then
\(r\downarrow0\), proves the limsup.  For the actual detector,
\(r_\eps=\sqrt\eps\), the front remainder is directly
\(O(\sqrt\eps+\eps)\).  The potential vanishes identically.
\end{proof}

\begin{corollary}[A Euclidean null slit remains invisible]
\label{cor:lifted-versus-null-cut}
Take \(g_\partial\equiv1\), \(g^+\equiv g^-\equiv0\), and \(u\equiv1\).
Then
\begin{equation}
 \lim_{\eps\downarrow0}\mathcal G_{\eps,K}^{1,0,0}(1)
 =2\int_Kh_{\rm fl}(x,\nu_K)\dd\Hh^{d-1}>0.
 \label{eq:slit-double-face-cost}
\end{equation}
The same constant state on the Euclidean null-slit representative
\(\Omega\setminus K\) has zero potential and zero pair energy for every
\(\eps\).  Thus the two coefficient-one face laws come from the intrinsic
finite-scale crossing rule, not from deleting a null set or appending a
sharp term after the limit.
\end{corollary}

\begin{proof}
Theorem~\ref{thm:lifted-cut-gamma-limit} gives the displayed limit, and
uniform spanning makes it positive.  Proposition
\ref{prop:null-cut-invisible} gives exact zero energy in the Euclidean
model because removing \(K\) changes neither Lebesgue classes nor the
absolutely continuous flow-pair measure.
\end{proof}

\subsection{A fixed finite sectorial junction}
\label{subsec:sectorial-junction}

The intrinsic slit theorem does not by itself determine how several faces
are to be completed at a common front.  We next establish the result for a
finite, fixed sectorial class.  This includes the complete \(Y\)-junction,
but not arbitrary branch graphs, moving junctions, or accumulating sheets.
Here \(K_\alpha\) denotes one sheet and
\(\mathcal K=\bigcup_\alpha K_\alpha\) the fan.  Accordingly,
\(\Psi_{\eps,K}^{i,t}\) is reserved for the single-slit detector, whereas
\(\Psi_{\eps,\mathcal K}^{i,t}\) denotes its gated sector extension and
agrees with it whenever the arc has one resolved crossing.

\paragraph{Sectorial geometry and the cut completion.}
Let \(N\ge3\), let \(\Gamma\Subset\Omega\) be a compact embedded
\(C^{1,1}\) submanifold of codimension two, and let
\(K_1,\ldots,K_N\Subset\Omega\) be compact embedded orientable
\(C^{1,1}\) hypersurfaces-with-boundary such that
\begin{equation}
 \partial K_\alpha=\Gamma,\qquad
 K_\alpha^\circ\cap K_\beta^\circ=\varnothing\quad(\alpha\ne\beta),
 \qquad \mathcal K:=\bigcup_{\alpha=1}^N K_\alpha .
 \label{eq:fan-incidence}
\end{equation}
Assume the outer separation in \eqref{eq:graph-slit-geometry} and a common
scale \(\rho_*>0\).  More precisely, a finite family of orientation-preserving
\(C^{1,1}\) charts, with chart and inverse \(C^{1,1}\) norms bounded by
\(C/\rho_*\), identifies a neighbourhood of every \(q\in\Gamma\) with
\((z,w)\in\R^{d-2}\times\R^2\), sends \(\Gamma\) to \(\{w=0\}\), and sends
the sheets to the labelled rays
\begin{equation}
 K_\alpha=\{(z,\varrho e_\alpha(z)):0\le\varrho<4\rho_*\},
 \qquad e_\alpha(z)=(\cos\vartheta_\alpha(z),\sin\vartheta_\alpha(z)).
 \label{eq:fan-ray-chart}
\end{equation}
The cyclic labels are compatible on overlaps and, with
\(\vartheta_{N+1}=\vartheta_1+2\pi\),
\begin{equation}
 \vartheta_{\alpha+1}(z)-\vartheta_\alpha(z)\ge\theta_*>0,
 \qquad \operatorname{Lip}e_\alpha\le\rho_*^{-1}.
 \label{eq:fan-angle-gap}
\end{equation}
There are no other intersections or higher-order strata.  No bound on the
number of raw intersections of a flow arc with one sheet is assumed.  The
detector below uses the gate \(r_\eps=\sqrt\eps\): outside that gate the angle
estimate proves that a short arc can meet at most one distinct sheet, and all
same-sheet recrossings are determined by its two endpoint sectors.  Inside
the gate the detector uses a bounded endpoint rule.  Thus no crossing-word
regularity hypothesis is needed.

Let \(S_\alpha(z)\) be the open sector between rays \(e_\alpha(z)\) and
\(e_{\alpha+1}(z)\).  Define
\begin{equation}
 X_{\mathcal K}:=\overline{(\Omega\setminus\mathcal K,
 d_{\Omega\setminus\mathcal K})},\qquad
 p:X_{\mathcal K}\to\overline\Omega,\qquad
 m_{\mathcal K}:=\jmath_\#
 \bigl(\mathcal L^d\mathbin{\llcorner}(\Omega\setminus\mathcal K)\bigr),
 \qquad p_\#m_{\mathcal K}=\mathcal L^d\mathbin{\llcorner}\Omega .
 \label{eq:fan-cut-completion}
\end{equation}
Here \(\jmath:\Omega\setminus\mathcal K\to X_{\mathcal K}\) is the canonical
isometric embedding.
This is a cut completion, not a global \(N\)-sheet cover.  The projection is
one-to-one over \(\Omega\setminus\mathcal K\).  Over
\(K_\alpha^\circ\) it has two face points, one for each adjacent sector.  Over
\(q\in\Gamma\) it has the \(N\) distinct points
\(q^{[1]},\ldots,q^{[N]}\), where the two boundary rays of
\(S_\alpha\) meet at \(q^{[\alpha]}\).  Thus every sector carries one copy of
its physical volume, and only the boundary topology is completed.  The
uniform angle gap makes each closed sector a Lipschitz wedge with constant
at most \(C\csc(\theta_*/2)\), so this description is independent of the
chosen fan chart.

For each sheet choose its two adjacent-sector labels
\(s_\alpha^-\) and \(s_\alpha^+\), traces
\(g_\alpha^- ,g_\alpha^+\in BV(K_\alpha;\{0,1\})\), and fixed opposite-collar
ghosts \(G_\alpha^- ,G_\alpha^+\).  The superscript names the physical sector
whose face datum is imposed; the ghost itself is evaluated in the collar on
the other side of that face.  We use the typed notation
\(G_\alpha^{s_\alpha^\pm}:=G_\alpha^\pm\) below.
The same half-collar construction just given supplies these binary ghosts
sheet by sheet; overlaps between prescribed ghost collars do not identify the
fields, and the junction gate removes their common codimension-two core.

\paragraph{Gated endpoint-sector crossing rule.}
Let \(\gamma=\gamma_{i,x,t}^\eps\) be the actual flow arc and retain the gate
\(r_\eps=\sqrt\eps\) from \eqref{eq:slit-tip-gate-width}.  We say that
\(\gamma\) has a \emph{resolved crossing} if \(x,y\in\Omega\),
\(\operatorname{dist}(\gamma,\Gamma)>r_\eps\), and its endpoints lie in the
two sectors adjacent to one unique labelled sheet \(K_{\alpha(\gamma)}\).
Write \(s_x,s_y\) for those endpoint sectors.  On the null set where an
endpoint lies on a sheet, use either stable one-sided value.  For all
sufficiently small \(\eps\), Lemma~\ref{lem:fan-junction-estimate} proves that
the labelled sheet is unique.  Multiple crossings of that same sheet have odd
parity exactly when the endpoint sectors differ, so no intermediate
intersection time or value of an \(L^1\) representative is used.  Denote the
unique label by \(\mathfrak a_\eps(\gamma)=\alpha(\gamma)\), and put
\(\mathfrak a_\eps(\gamma)=\varnothing\) otherwise.  This map is Borel:
distance from a continuous arc to the closed front is continuous in its
starting point, and the finitely many endpoint-sector membership tests are
Borel.  Set
\begin{equation}
 \Psi_{\eps,\mathcal K}^{i,t}[u](x,y)
 :=\begin{cases}
 |u(x)-G_{\alpha(\gamma)}^{s_x}(y)|^2+
 |u(y)-G_{\alpha(\gamma)}^{s_y}(x)|^2,
 &\gamma\text{ has a resolved crossing},\\
 |\bar u^{\rm p}(x)-\bar u^{\rm p}(y)|^2,&\text{otherwise}.
 \end{cases}
 \label{eq:fan-gated-pairing}
\end{equation}
Path reversal preserves the gate and the unique sheet and exchanges
\((x,s_x)\) with \((y,s_y)\).  Hence \eqref{eq:fan-gated-pairing} is
pair-symmetric and its first branch is exactly the two endpoint-to-own-face
rule \eqref{eq:slit-finite-epsilon-pairing}.  Every arc entering the gate
receives the bounded, representative-independent ordinary endpoint rule.
Any other bounded symmetric core rule has the same Gamma-limit by
Lemma~\ref{lem:fan-junction-estimate}.

Define \(\mathcal G_{\eps,\mathcal K}^{g}\) by
\eqref{eq:slit-prelimit-energy}, with
\(\Psi_{\eps,K}^{i,t}\) replaced by
\(\Psi_{\eps,\mathcal K}^{i,t}\).  No intermediate value and no
positive-volume sheet has been added.

To make the junction constants explicit, fix the finite fan atlas
\(\{\Upsilon_\gamma\}\) used above and set
\begin{align}
 m_{\rm fan}&:=\sup_{x\in\Omega}
 \#\{\gamma:x\in\operatorname{im}\Upsilon_\gamma\},\notag\\
 J_{\rm fan}&:=\max_\gamma\left(
 \|\det D\Upsilon_\gamma\|_\infty
 +\|\det D\Upsilon_\gamma^{-1}\|_\infty\right),\notag\\
 Q_{\rm fan}&:=\max_\gamma\rho_*\left(
 \|D^2\Upsilon_\gamma\|_\infty
 +\|D^2\Upsilon_\gamma^{-1}\|_\infty\right),\notag\\
 C_{\rm fan}&:=C_d\,m_{\rm fan}Nc^*
 \bigl(1+e^{\eps_0bD_0}\bigr)
 (1+J_{\rm fan})(1+Q_{\rm fan})
 (1+bV_0+bV_1).
 \label{eq:fan-explicit-constant}
\end{align}
Here \(C_d\) is the same type of dimensional constant as above.  Enlarging
it once also covers the finite seam-atlas multiplicity.  Dependence on the
angle and tubular radius remains visible in the estimates through
\(\csc(\theta_*/2)\) and \(\rho_*^{-1}\), rather than being hidden in
\(C_{\rm fan}\).

\begin{lemma}[Angle separation, face cells, and the junction estimate]
\label{lem:fan-junction-estimate}
There is \(r_0>0\), determined by the fixed atlas, such that, if
\(q_\alpha(z,\varrho)\) is the physical point represented
by \((z,\varrho e_\alpha(z))\), then for \(0<\varrho<r_0\) and
\(\alpha\ne\beta\),
\begin{equation}
 |q_\alpha(z,\varrho)-q_\beta(z,\varrho)|
 \ge2\varrho\sin(\theta_*/2)-C_{\rm fan}\varrho^2/\rho_*
 \label{eq:fan-angle-separation}
\end{equation}
on the transverse circle of radius \(\varrho\).  If
\(A_{i,\eps,t}^{\rm jun}(r)\) denotes the set of starting points whose arc
crosses a sheet at intrinsic sheet distance at most \(r\) from \(\Gamma\), or
meets two distinct labelled sheets, then
\begin{align}
 |A_{i,\eps,t}^{\rm jun}(r)|
 &\le C_{\rm fan}\eps t\left(r+\eps t\csc(\theta_*/2)\right)
      \Hh^{d-2}(\Gamma),
 \label{eq:fan-active-junction-volume}\\
 \frac1\eps\widehat Q_{i,\eps,t}
 \bigl(A_{i,\eps,t}^{\rm jun}(r)\times U
       \cup U\times A_{i,\eps,t}^{\rm jun}(r)\bigr)
 &\le C_{\rm fan}t\left(r+\eps t\csc(\theta_*/2)\right)
      \Hh^{d-2}(\Gamma).
 \label{eq:fan-normalized-junction-bound}
\end{align}
Every open face of every \(K_\alpha\) has coefficient one.  Moreover the
possible junction density, defined rather than suppressed by
\begin{equation}
 j_\Gamma[u,g]:=\lim_{r\downarrow0}\limsup_{\eps\downarrow0}
 \frac1{\Hh^{d-2}(\Gamma)}
 \sum_i\int_a^b\frac1\eps
 \iint_{(A_{i,\eps,t}^{\rm jun}(r)\times U)\cup
 (U\times A_{i,\eps,t}^{\rm jun}(r))}
 \Psi_{\eps,\mathcal K}^{i,t}\dd\widehat Q_{i,\eps,t}\lambda(\dd t),
 \label{eq:fan-junction-cell}
\end{equation}
satisfies \(j_\Gamma[u,g]=0\) for binary fields and ghosts.  With
\(r_\eps=\sqrt\eps\), its normalized remainder is
\(O(\sqrt\eps+\eps\csc(\theta_*/2))\).

In addition, curved fan coordinates and flow freezing cost
\begin{equation}
 C_{\rm fan}\eps\bigl(\rho_*^{-1}+V_1\bigr)
 \sum_{\alpha=1}^N\Hh^{d-1}(K_\alpha),
 \label{eq:fan-curvature-estimate}
\end{equation}
while deleting an \(\eta\)-strip around the finite hard-atlas seam union
\(\Sigma_{\mathcal K}\) costs
\begin{equation}
 C_{\rm fan}(\eta+\eps)\Hh^{d-2}(\Sigma_{\mathcal K}).
 \label{eq:fan-seam-estimate}
\end{equation}
Thus every constant is bounded by \eqref{eq:fan-explicit-constant}; the only
additional degeneracies are the displayed factors
\(\csc(\theta_*/2)\) and \(\rho_*^{-1}\).  No constant depends on
\(\eps,r,\eta\).
\end{lemma}

\begin{proof}
In a flat transverse plane, two rays with angle at least \(\theta_*\) are
separated on the radius-\(\varrho\) circle by
\(2\varrho\sin(\theta_*/2)\).  The \(C^{1,1}\) chart changes this by at most
\(C_{\rm fan}\varrho^2/\rho_*\), proving
\eqref{eq:fan-angle-separation}.  A flow arc of length \(\eps t\) that crosses
two distinct sheets must enter the transverse cylinder
\(\varrho\le C_{\rm fan}\eps t\csc(\theta_*/2)\).  Indeed, the two crossing
points differ by at most \(C\eps t\); replacing their tangential and radial
coordinates by one common \((z,\varrho)\) changes the ray separation by at
most \(C_{\rm fan}\eps t\) through the Lipschitz fan frame.  The angle estimate
then yields the displayed cylinder after enlarging \(C_{\rm fan}\).
Since \(\eps b=o(r_\eps)\), for all sufficiently small \(\eps\) an arc
avoiding the gate \(T_{r_\eps}(\Gamma)\) cannot meet two distinct sheets.
This proves that \(\mathfrak a_\eps\) is well-defined without counting raw
same-sheet intersections.
Prescribing the first crossing gives an interval of starting points of length
\(C_{\rm fan}\eps t\); the admissible crossing strip along a sheet has width
\(C_{\rm fan}(r+\eps t\csc(\theta_*/2))\).  Fubini, the fan-chart Jacobian, and
the finite sheet and atlas multiplicities prove
\eqref{eq:fan-active-junction-volume}.  The pair marginal gives
\eqref{eq:fan-normalized-junction-bound}.
If the full arc enters \(T_r(\Gamma)\) while a crossing occurs outside it,
the crossing point lies within \(C\eps t\) of that tube.  Hence every bond
whose classification is changed by the gate is contained in
\(A_{i,\eps,t}^{\rm jun}(r+C\eps t)\), which obeys the same bound after
enlarging \(C_{\rm fan}\).

 Away from this set an arc meets at most one distinct sheet.  Its endpoint
sectors decide the same-sheet intersection parity, so the resolved branch of
\eqref{eq:fan-gated-pairing} is used exactly for an odd crossing.  Flattening
that sheet leaves one physical endpoint and its own-face ghost on opposite
half-spaces.  Lemma
\ref{lem:boundary-halfspace-cell} gives \(1/2+1/2=1\) for the directed and
transposed halves.  Repeating on the other side proves coefficient one for
 each of the \(2N\) open faces.  On the active set the canonical core
 integrand is at most one; allowing either the ordinary or resolved branch
 gives the uniform bound two.  Integrating
\eqref{eq:fan-normalized-junction-bound}, first sending \(\eps\downarrow0\)
and then \(r\downarrow0\), gives \(j_\Gamma[u,g]=0\).  Thus zero junction cost is
a cell computation, not a codimension assertion imposed in advance.  The
choice \(r=\sqrt\eps\) gives the stated direct rate.

The signed-distance and Jacobian expansions used in
\eqref{eq:slit-curvature-flow-estimate} apply sheet by sheet and sum to
\eqref{eq:fan-curvature-estimate}.  Exact label compatibility removes overlap
error.  If hard cells are used, the same two-normal-variable argument as in
\eqref{eq:slit-chart-seam-estimate}, summed over the finite atlas, proves
\eqref{eq:fan-seam-estimate}.
\end{proof}

Use the sector presentation in
Remark~\ref{rem:slit-fan-resolved-presentations} and write
\(BV(X_{\mathcal K}^\circ):=BV_{\rm res}(X_{\mathcal K}^\circ)\).  At
\(\Gamma\), transparent traces are glued only within the same sector vertex;
differences across open sheet faces are not part of
\(D_{X_{\mathcal K}}u\).  Proposition~\ref{prop:resolved-bv-package} gives
atlas independence, compactness, anisotropic lower semicontinuity, and all
outer and face traces used in the energy below.  Define
\begin{align}
 \mathcal G_{0,\mathcal K}^{g}(u)
 :={}&\int_{X_{\mathcal K}^\circ}
 h_{\rm fl}(p(x),\sigma_{D_{X_{\mathcal K}}u})
 \dd|D_{X_{\mathcal K}}u|(x)\notag\\
 &+\int_{\partial\Omega}h_{\rm fl}(x,\nu_\Omega)
 |T_\partial u-g_\partial|\dd\Hh^{d-1}\notag\\
 &+\sum_{\alpha=1}^N\int_{K_\alpha}h_{\rm fl}(x,\nu_\alpha)
 \left(|T_\alpha^-u-g_\alpha^-|+|T_\alpha^+u-g_\alpha^+|\right)
 \dd\Hh^{d-1},
 \label{eq:fan-limit-energy}
\end{align}
with value \(+\infty\) outside that binary \(BV\) space.  There is no
compulsory term on \(\Gamma\), and the constant recovery has no junction
concentration, by Lemma~\ref{lem:fan-junction-estimate}.

\begin{theorem}[Gamma-limit for a fixed finite sectorial junction]
\label{thm:fan-junction-gamma-limit}
Assume \eqref{eq:ray-radial-law}, \eqref{eq:boundary-ambient-frame},
\eqref{eq:flow-ray-steepness}, and the fixed sectorial hypotheses
\eqref{eq:fan-incidence}--\eqref{eq:fan-angle-gap}.  Then
\(\mathcal G_{\eps,\mathcal K}^{g}\) Gamma-converges to
\(\mathcal G_{0,\mathcal K}^{g}\) in
\(L^1(X_{\mathcal K},m_{\mathcal K})=L^1(\Omega)\).  Every bounded-energy
sequence is, after threshold rounding and extraction, strongly compact in
that space with a limit in
\(BV(X_{\mathcal K}^\circ;\{0,1\})\).  Every admissible binary state is
recovered by the constant sequence.
\end{theorem}

\begin{proof}
\emph{Compactness.}
Every branch of \eqref{eq:fan-gated-pairing} is either one ordinary endpoint
pair or two endpoint-to-binary-ghost terms.  In both cases its rounding excess
is bounded by \(2d_*(u(x))+2d_*(u(y))\).  The two endpoint marginals are
therefore exactly those already estimated in
\eqref{eq:flow-ray-steepness}.  Threshold rounding
still uses one physical-volume potential.  Ordinary and
one-sided face boxes give local \(BV\) bounds as follows.  On each face box
away from \(\Gamma\), the detector has one stable resolved sheet.  Glue the physical
field in either adjacent sector to that face's fixed opposite-collar ghost.
After adding the fixed ghost--ghost bonds, the local energy is the ambient
crossing-bond energy, so Theorem~\ref{thm:boundary-flow-gamma-limit} gives
compactness of the glued \(BV\) fields.  The limiting one-sided traces are
then read from their gluing measures; no compactness of the trace operator is
used.  A diagonal over finite face-box
exhaustions treats all \(2N\) face copies.  In a junction box,
cut every sector along a bisector and exhaust the resulting Lipschitz wedges;
their extension and Poincare constants are bounded by
\(C\csc(\theta_*/2)\).  The sector labels agree on chart overlaps, so the
local limits glue at exactly the vertices prescribed by
\eqref{eq:fan-cut-completion}.  Alternatively, after projection to
\(\Omega\), the only uncontrolled Euclidean jumps lie on the fixed union
\(\mathcal K\), and their total variation is at most
\(2\sum_\alpha\Hh^{d-1}(K_\alpha)\).  Ordinary \(BV(\Omega)\) compactness and
the finite sector atlas give the claimed strong \(L^1\) convergence.

\emph{Liminf.}
Fix one regular presentation and call its auxiliary seam family
\(\mathcal S\).  Choose regular \(r,\eta>0\) and \(\delta>0\).  Outside
\(T_r(\Gamma)\cup T_\eta(\mathcal S)\), every sufficiently short arc has at
most one resolved sheet.  Lemma~\ref{lem:resolved-cell-localization} supplies a finite
disjoint family of ordinary and one-sided face cells and gives
\begin{align}
 \liminf_{\eps\downarrow0}\mathcal G_{\eps,\mathcal K}^{g}(u_\eps)
 \ge{}&
 \mathfrak m_u\bigl(
 (X_{\mathcal K}^\circ\sqcup\partial\Omega
  \sqcup\textstyle\bigsqcup_{\alpha}(K_\alpha^-\sqcup K_\alpha^+))
 \setminus(T_r(\Gamma)\cup T_\eta(\mathcal S))\bigr)-\delta .
 \label{eq:fan-localized-liminf}
\end{align}
Ambient completion identifies each one-sided trace, and
Lemma~\ref{lem:boundary-halfspace-cell} fixes every face coefficient at one.
The losses are \(O(\eps/\rho_*)\) from
\eqref{eq:fan-curvature-estimate}, \(O(\eta+\eps)\) from
\eqref{eq:fan-seam-estimate}, and
\(O(r+\eps\csc(\theta_*/2))\) from
\eqref{eq:fan-normalized-junction-bound}.  Send \(\eps\downarrow0\), then
\(\delta\downarrow0\), \(\eta\downarrow0\), and \(r\downarrow0\).
Continuity from below and \(\mathfrak m_u(\Gamma)=0\) prove
\eqref{eq:fan-limit-energy}.

\emph{Recovery.}
For \(u\in BV(X_{\mathcal K}^\circ;\{0,1\})\), take \(u_\eps=u\).
The upper-bound part of Lemma~\ref{lem:resolved-cell-localization} applies
sheetwise to the finitely many physical/ghost glued fields.  Identity
\eqref{eq:one-step-transparent-refinement} recombines the auxiliary-seam
interior and trace-jump terms into \(D_{X_{\mathcal K}}u\), and subtraction
of the fixed ghost--ghost terms gives the bulk, outer contact, and all \(2N\)
labelled face contacts off the deleted sets.  The same three
estimates bound the difference
on the deleted sets.  The intrinsic variation of \(u\) in
\(T_r(\Gamma)\) tends to zero and the sheet areas there are
\(O(r\Hh^{d-2}(\Gamma))\).  The same ordered limit proves the limsup; for
\(r_\eps=\sqrt\eps\) the junction remainder vanishes directly.  The
potential is zero.
\end{proof}

\begin{corollary}[The complete \(Y\)-junction and Euclidean invisibility]
\label{cor:fan-null-set-invisible}
For \(N=3\), Theorem~\ref{thm:fan-junction-gamma-limit} gives the full fixed
\(Y\)-junction theorem.  More generally, take \(g_\partial\equiv1\), all
\(g_\alpha^\pm\equiv0\), and \(u\equiv1\).  Then
\begin{equation}
 \lim_{\eps\downarrow0}\mathcal G_{\eps,\mathcal K}^{g}(1)
 =2\sum_{\alpha=1}^N\int_{K_\alpha}
 h_{\rm fl}(x,\nu_\alpha)\dd\Hh^{d-1}>0,
 \label{eq:fan-double-face-cost}
\end{equation}
whereas the ordinary Euclidean state on
\(\Omega\setminus\mathcal K\) has zero potential and zero pair energy for
every \(\eps\).  The Euclidean branched null set cannot see the sector
vertices, any of the \(2N\) face copies, or their coefficient-one costs.
\end{corollary}

\begin{proof}
The intrinsic identity is Theorem~\ref{thm:fan-junction-gamma-limit} applied
to the constant state; positivity follows from uniform spanning.  Since
\(\mathcal L^d(\mathcal K)=0\), Proposition~\ref{prop:null-cut-invisible}
applies to the finite union and gives exact Euclidean zero.  The junction
contributes neither side of the comparison by
\eqref{eq:fan-junction-cell}.
\end{proof}

\subsection{The master theorem for finite resolved cut complexes}
\label{subsec:finite-cut-master-theorem}

The preceding slit and fan geometries can coexist.  We record the resulting
statement because it is the form stable under finite localization and the one
used to organize the theory.

\begin{definition}[Admissible finite resolved cut complex]
\label{def:admissible-finite-cut-complex}
An admissible finite resolved cut complex \(\mathfrak C\Subset\Omega\) is a
finite union of compact orientable \(C^{1,1}\) hypersurfaces-with-boundary for
which the singular set is a finite disjoint union
\(\Sigma=\bigcup_{r=1}^R\Gamma_r\) of compact embedded \(C^{1,1}\)
codimension-two manifolds.  Every point of \(\Sigma\) has a uniform chart of
one of the following two types:
\begin{enumerate}
 \item one terminated sheet with the curved-slit model
 \eqref{eq:graph-slit-geometry};
 \item finitely many labelled sheets with the sectorial model
 \eqref{eq:fan-incidence}--\eqref{eq:fan-angle-gap}.
\end{enumerate}
The reach, chart, and angle bounds are uniform over the finite
complex, distinct singular components have positive separation, and there are
no higher-order strata.  The completion \(X_{\mathfrak C}\) resolves every
open sheet bank and uses the finite chamber gluing of
Definition~\ref{def:finite-resolved-presentation}.  A bond is lifted by the
mod-two bank rule in a slit chart, by the gated endpoint-sector rule
\eqref{eq:fan-gated-pairing} in a fan chart, and by ordinary physical pairing
elsewhere.  These rules agree on overlaps.
\end{definition}

For each open sheet \(K_a\) choose the two labelled data \(g_a^\pm\) and let
\(\mathcal G_{\eps,\mathfrak C}^g\) be the resulting single-volume lifted-bond
energy.  Set
\begin{align}
 \mathcal G_{0,\mathfrak C}^{g}(u):={}&
 \int_{X_{\mathfrak C}^\circ}
 h_{\rm fl}(p,\sigma_{D_{X_{\mathfrak C}}u})
 \dd|D_{X_{\mathfrak C}}u|
 +\int_{\partial\Omega}h_{\rm fl}(x,\nu_\Omega)
 |T_\partial u-g_\partial|\dd\Hh^{d-1}\notag\\
 &+\sum_a\int_{K_a}h_{\rm fl}(x,\nu_a)
 \bigl(|T_a^-u-g_a^-|+|T_a^+u-g_a^+|\bigr)\dd\Hh^{d-1},
 \label{eq:finite-cut-master-limit}
\end{align}
with value \(+\infty\) outside
\(BV_{\rm res}(X_{\mathfrak C}^\circ;\{0,1\})\).

\begin{theorem}[Master Gamma-limit on a finite resolved cut complex]
\label{thm:finite-cut-master-gamma}
Assume the flow-frame hypotheses
\eqref{eq:ray-radial-law}, \eqref{eq:boundary-ambient-frame}, and
\eqref{eq:flow-ray-steepness}, and let \(\mathfrak C\) be admissible in the
sense of Definition~\ref{def:admissible-finite-cut-complex}.  Then
\(\mathcal G_{\eps,\mathfrak C}^g\) is equicoercive and Gamma-converges in
\(L^1(X_{\mathfrak C},m_{\mathfrak C})=L^1(\Omega)\) to
\(\mathcal G_{0,\mathfrak C}^g\).  Every admissible binary state is recovered
by the constant sequence.  Each open bank has contact coefficient one, and
the Gamma-limit has no compulsory codimension-two term supported on
\(\Sigma\).  Constant recovery has no concentration there.
\end{theorem}

\begin{proof}
Fix one regular presentation, with auxiliary seam family \(\mathcal S\).
Choose regular \(r,\eta>0\), pairwise disjoint tubes
\(T_r(\Gamma_j)\), and \(\delta>0\).  Apply
Lemma~\ref{lem:resolved-cell-localization} to obtain a finite disjoint family
of ordinary and one-sided face boxes which exhausts the sharp measure outside
\(T_r(\Sigma)\cup T_\eta(\mathcal S)\) up to \(\delta\).  On each
terminated-sheet tube use Lemma~\ref{lem:slit-face-tip-estimate}; on each
sectorial tube use Lemma~\ref{lem:fan-junction-estimate}.  The constants can be
replaced by their maximum because the complex is finite.  Away from the
singular tubes, every sufficiently short lifted bond has either no letter or
one face letter, so the boundary half-space cell gives coefficient one.

Threshold compactness is chamberwise.  Proposition
\ref{prop:resolved-bv-package} glues transparent seams, preserves independent
bank traces, and gives the anisotropic lower semicontinuity needed after the
finite disjoint localization.  Since every interaction is nonnegative,
\eqref{eq:resolved-cell-finite-liminf} gives the complete sharp measure on the
retained cells minus \(\delta\).  Send first \(\eps\downarrow0\), then
\(\delta\downarrow0\), \(\eta\downarrow0\), and finally \(r\downarrow0\).
The seam terms recombine into the intrinsic derivative by
\eqref{eq:anisotropic-refinement-invariance}, while continuity from below and
the fact that the sharp Radon measure does not charge \(\Sigma\) prove the
lower bound.  For
the constant recovery, the sum of all terminated-sheet remainders is
\(O(r+\eps)\sum_j\Hh^{d-2}(\Gamma_j)\); the sectorial remainders are
\(O(r+\eps\csc(\theta_*/2))\sum_j\Hh^{d-2}(\Gamma_j)\).  Curvature and hard
seam errors are the finite sums of
\eqref{eq:slit-curvature-flow-estimate},
\eqref{eq:slit-chart-seam-estimate},
\eqref{eq:fan-curvature-estimate}, and \eqref{eq:fan-seam-estimate}.
Use the upper-bound part of Lemma~\ref{lem:resolved-cell-localization} on the
same retained cells and take the same ordered limit.  The displayed error
estimates vanish, proving recovery and also proving, rather than assuming, the
absence of a compulsory codimension-two term in the Gamma-functional.  The
argument concerns the constructed recovery sequence and does not rule out
positive excess defects for other bounded-energy sequences.
\end{proof}

\begin{theorem}[Uniform equivalence of direct and regular detectors]
\label{thm:regular-detector-equivalence-main}
Let \(\mathfrak C\) be an admissible finite resolved cut complex, and let
\(\mathcal H_{\eps,\Omega,\mathfrak C}^{g}\) be the direct energy
\eqref{eq:rectifiable-direct-energy} on its labelled open sheets.  Assume that
every collar ghost used in the slit and fan detector has the prescribed
one-sided trace \(g_a^\pm\) on its bank.  Then
\begin{equation}
 \lim_{\eps\downarrow0}\sup_{0\le u\le1}
 \left|
  \mathcal G_{\eps,\mathfrak C}^{g}(u)
  -\mathcal H_{\eps,\Omega,\mathfrak C}^{g}(u)
 \right|=0.
 \label{eq:regular-direct-uniform-equivalence}
\end{equation}
Consequently Theorems~\ref{thm:finite-cut-master-gamma} and
\ref{thm:rectifiable-cut-master-gamma} give the same Gamma-limit on every
regular complex.  The slit and fan construction is therefore a geometric
realization of the direct crossing rule, not a different finite-scale model.
\end{theorem}

\begin{proof}
Set \(K_{\mathfrak C}:=\bigcup_aK_a\) and
\(\Sigma_{\rm cr}:=K_{\mathfrak C}\cup\partial\Omega\).  Write the
absolute difference of the two interaction terms as
\begin{equation}
 R_\eps(u)\le R_\eps^{\rm face}(u)+R_\eps^{\rm gate}(u)
                   +R_\eps^{\rm mult}(u),
 \label{eq:regular-direct-error-decomposition}
\end{equation}
according as the relevant arc has one crossing outside the gates, enters a
front or junction gate, or has at least two crossings.  The potential terms
are identical.  We estimate the three terms uniformly for \(0\le u\le1\).

On a uniquely crossed sheet outside the gates, the regular detector evaluates
an opposite-collar ghost at the endpoint, whereas the direct rule evaluates
its trace at the crossing point.  For \(0\le r,p,q\le1\),
\[
 \bigl||r-p|^2-|r-q|^2\bigr|\le2|p-q|.
\]
For a labelled bank \((a,\sigma)\), let \(C_{a,r}^\sigma\) be its one-sided
collar of width \(r\), let \(\pi_a\) be the collar projection, and set
\begin{equation}
 \omega_{a,\sigma}(r):=
 \sup_{0<q\le r}\frac1q
 \int_{C_{a,q}^\sigma}
 |G_a^\sigma-g_a^\sigma\circ\pi_a|\,\dd x .
 \label{eq:regular-direct-trace-modulus}
\end{equation}
The strong one-sided \(BV\) trace theorem gives
\(\omega_{a,\sigma}(r)\to0\) as \(r\downarrow0\).  The collar maps have
uniformly bounded Jacobians, and every endpoint of a bond of flow time
\(\eps t\), \(t\in[a,b]\), lies in a collar of width at most \(C\eps b\).
Changing variables from endpoints to the crossing point and collar distance,
or equivalently applying the area formula
\eqref{eq:rectifiable-exact-crossing}, therefore gives
\begin{equation}
 \sup_{0\le u\le1}R_\eps^{\rm face}(u)
 \le C\sum_{a,\sigma}\omega_{a,\sigma}(C\eps b)
 \longrightarrow0.
 \label{eq:regular-direct-face-error}
\end{equation}
The same formula supplies the integrable majorant needed for radial
integration.

The crossing points of arcs entering a shrinking gate belong to a decreasing
surface neighbourhood of the front and junction set.  That set is
\(\Hh^{d-1}\)-null.  Let \(K_a[r]\) denote its intrinsic
\(r\)-neighbourhood in the sheet \(K_a\).  Then the
gate widths \(r_\eps\downarrow0\) and the exact crossing formula yield
\begin{equation}
 \sup_{0\le u\le1}R_\eps^{\rm gate}(u)
 \le C\sum_{i,a}\int_{K_a[Cr_\eps]}
 c_i|v_i\cdot\nu_a|\,\dd\Hh^{d-1}\longrightarrow0
 \label{eq:regular-direct-gate-error}
\end{equation}
by continuity from above.  Finally, boundedness of both detectors gives
\begin{equation}
 \sup_{0\le u\le1}R_\eps^{\rm mult}(u)
 \le \frac C\eps\sum_i\int_a^b
 \bigl|\{x:N_{i,\eps t}^{\Sigma_{\rm cr}}(x)\ge2\}\bigr|
 \lambda(\dd t)
 \longrightarrow0.
 \label{eq:regular-direct-multiple-error}
\end{equation}
Indeed, \(\one_{\{N\ge2\}}\le(N-P)/2\), and
\eqref{eq:rectifiable-parity-error}, dominated in \(t\) by the exact crossing
formula, proves the limit.  Combining
\eqref{eq:regular-direct-error-decomposition}--
\eqref{eq:regular-direct-multiple-error} proves
\eqref{eq:regular-direct-uniform-equivalence}.
\end{proof}

\begin{remark}[Scope of the path rule]
\label{rem:master-path-rule-scope}
The theorem is universal over the finite geometric class just defined, but it
does not quantify over arbitrary path-lifting algorithms.  Its detector is
the canonical homotopy-stable parity rule at a terminated sheet and the gated
endpoint-sector rule at a junction.  A broader axiomatic classification would
require proving that every admissible lifting rule reduces to these local
rules; that statement is neither used nor claimed here.
\end{remark}

\section{Noncompact Radial Laws and Their Thresholds}
\label{sec:long-range-radial-laws}

We now remove compact radial support.  To avoid mixing remote flow times with
an ambient extension of a bounded body, this section is stated on
\(\mathbb T^d\), with the boundary term omitted from
\eqref{eq:rectifiable-direct-limit}.  The cut datum still satisfies
\eqref{eq:rectifiable-cut-mass}.  Every flow is globally defined.  For a
binary state and physical time \(h>0\), let
\begin{equation}
 B_{i,\mathfrak K}^g(u;h)
 :=\int_{\mathbb T^d}c_i(x)
 \Psi_{i,h}^{\mathfrak K}[u]
       \bigl(x,\Phi_i^{-h}x\bigr)\dd x,
 \label{eq:long-range-physical-bond}
\end{equation}
where \(\Psi_{i,h}^{\mathfrak K}\) is the direct rule of
Section~\ref{sec:rectifiable-direct-cuts}: ordinary interaction for no
crossing, two bank contacts for one crossing, and zero for two or more
crossings.  Denote the boundary-free restriction of
\eqref{eq:rectifiable-direct-limit} by
\begin{align}
 \mathcal H_{0,\mathfrak K}^{g}(u)
 :={}&\int_{\mathbb T^d}h_{\rm fl}(x,\sigma_{D_{\mathfrak K}u})
       \dd|D_{\mathfrak K}u|\notag\\
 &+\sum_j\int_{K_j}h_{\rm fl}(x,\nu_j)
 \bigl(|T_j^+u-g_j^+|+|T_j^-u-g_j^-|\bigr)\dd\Hh^{d-1}.
 \label{eq:long-range-local-limit}
\end{align}

Let \(\lambda\) be a finite positive Borel measure on \((0,\infty)\).  Define
\begin{align}
 \mathcal H_{\eps,\mathfrak K}^{g,\lambda}(u)
 :={}&\frac M\eps\int_{\mathbb T^d}d_*(u)\dd x
 +\frac1\eps\sum_i\int_0^\infty
 B_{i,\mathfrak K}^g(u;\eps t)\lambda(\dd t),
 \label{eq:long-range-energy}
\end{align}
with the direct bond extended to \(0\le u\le1\) by the same quadratic
endpoint and bank expressions.  Assume
\begin{equation}
 m_1:=\int_0^\infty t\lambda(\dd t)\in(0,\infty),
 \label{eq:long-range-first-moment-main}
\end{equation}
and that for some \(0<a<b<\infty\) and \(w_*>0\),
\begin{equation}
 \lambda(A)\ge w_*\mathcal L^1(A\cap(a,b))
 \quad\text{for every Borel }A\subset(0,\infty).
 \label{eq:long-range-spreading-main}
\end{equation}
For diffuse-state rounding at remote times we also assume bounded compression,
\begin{equation}
 C_J:=\max_i\sup_{s\in\R}\|J\Phi_i^s\|_{L^\infty}<\infty,
 \qquad
 M\ge2\lambda((0,\infty))\sum_i c^*(1+C_J).
 \label{eq:long-range-compression-main}
\end{equation}
Volume-preserving flows have \(C_J=1\).

\begin{proposition}[Finite radial-window Gamma-limit]
\label{prop:finite-radial-window-main}
Let \(\rho\) be a finite positive Borel measure supported in
\([r_-,r_+]\Subset(0,\infty)\), and put
\begin{equation}
 m_\rho:=\int_0^\infty t\rho(\dd t)>0.
 \label{eq:finite-window-first-moment}
\end{equation}
Assume that for some \(0<a_0<b_0<\infty\) and \(w_0>0\),
\begin{equation}
 \rho(A)\ge w_0\mathcal L^1(A\cap(a_0,b_0))
 \quad\text{for every Borel }A\subset(0,\infty).
 \label{eq:finite-window-spreading}
\end{equation}
Under \eqref{eq:rectifiable-cut-mass} and
\eqref{eq:flow-ray-regularity}--\eqref{eq:flow-ray-spanning}, suppose the
flows are global and
\begin{equation}
 \begin{aligned}
 C_J&:=\max_i\sup_{s\in\mathbb R}\|J\Phi_i^s\|_{L^\infty}<\infty,\\
 M&\ge2\rho((0,\infty))\sum_i c^*(1+C_J).
 \end{aligned}
 \label{eq:finite-window-compression-steepness}
\end{equation}
Then \(\mathcal H_{\eps,\mathfrak K}^{g,\rho}\) is equicoercive in
\(L^1(\mathbb T^d)\) and
\begin{equation}
 \mathcal H_{\eps,\mathfrak K}^{g,\rho}
 \mathop{\longrightarrow}^{\Gamma}
 m_\rho\mathcal H_{0,\mathfrak K}^{g}.
 \label{eq:finite-window-gamma}
\end{equation}
Every binary state of finite limit energy is recovered by the constant
sequence.
\end{proposition}

\begin{proof}
The pointwise rounding inequalities
\eqref{eq:rectifiable-rounding-inequalities}, the two flow-marginal bounds,
and \eqref{eq:finite-window-compression-steepness} show that thresholding decreases the
energy and has \(L^1\) error at most \(\eps\mathcal H/M\).  It is therefore
enough to work with binary states.

For a binary state \(z\), the triangle inequality across a singly crossed arc
and boundedness on the multiple-crossing class give
\begin{equation}
 c_*\|z-z\circ\Phi_i^{-h}\|_1
 \le B_{i,\mathfrak K}^g(z;h)
 +C\int_{\mathbb T^d}c_i(x)
       \one_{\{N_{i,h}^{K}(x)\ge1\}}\dd x.
 \label{eq:finite-window-uncut-comparison}
\end{equation}
Integrate this inequality over the spreading submeasure
\(w_0\one_{(a_0,b_0)}(t)\dd t\), put \(h=\eps t\), and divide by \(\eps\).
The exact crossing formula bounds the second term uniformly by
\(C\sum_j\Hh^{d-1}(K_j)\).  Thus every bounded-energy binary sequence
satisfies the averaged increment hypothesis of
Lemma~\ref{lem:ordered-flow-smoothing}.  That lemma and
\eqref{eq:flow-ray-spanning} prove equicoercivity.

Let now \(z_n\to z\) in \(L^1\) and \(\eps_n\downarrow0\).  For every fixed
\(t\in[r_-,r_+]\), the periodic part of
Theorem~\ref{thm:global-resolved-translation-contact}, applied with
\(h_n=\eps_nt\), gives
\begin{equation}
 \liminf_{n\to\infty}\frac1{\eps_nt}
 \sum_iB_{i,\mathfrak K}^g(z_n;\eps_nt)
 \ge\mathcal H_{0,\mathfrak K}^{g}(z).
 \label{eq:finite-window-fixed-radius-liminf}
\end{equation}
If the unweighted quotient has infinite liminf the inequality is automatic;
otherwise it is precisely the sequential statement of that theorem.  Fatou's
lemma, positivity of the potential, and
\eqref{eq:finite-window-first-moment} yield
\begin{align}
 \liminf_{n\to\infty}
 \mathcal H_{\eps_n,\mathfrak K}^{g,\rho}(z_n)
 &\ge\int_0^\infty t
 \liminf_{n\to\infty}\frac1{\eps_nt}
 \sum_iB_{i,\mathfrak K}^g(z_n;\eps_nt)\rho(\dd t)\notag\\
 &\ge m_\rho\mathcal H_{0,\mathfrak K}^{g}(z).
 \label{eq:finite-window-liminf}
\end{align}

Finally fix binary \(z\) with finite limit energy.  The fixed-state part of
Theorem~\ref{thm:global-resolved-translation-contact} gives pointwise
convergence of the quotient in
\eqref{eq:finite-window-fixed-radius-liminf} to
\(\mathcal H_{0,\mathfrak K}^{g}(z)\).  Its translation majorant is uniform
for \(0<h<h_0\).  Since \(t\) belongs to a fixed compact interval, for all
small \(\eps\)
\begin{equation}
 \frac1\eps\sum_iB_{i,\mathfrak K}^g(z;\eps t)
 \le C_z t,
 \label{eq:finite-window-recovery-majorant}
\end{equation}
and \(t\in L^1(\rho)\).  Dominated convergence proves constant recovery and
completes the Gamma-convergence proof.
\end{proof}

\begin{theorem}[Finite-first-moment long-range limit]
\label{thm:long-range-first-moment-main}
Under \eqref{eq:rectifiable-cut-mass},
\eqref{eq:flow-ray-regularity}--\eqref{eq:flow-ray-spanning}, global
existence of the flows, and \eqref{eq:long-range-first-moment-main}--
\eqref{eq:long-range-compression-main}, the family
\(\mathcal H_{\eps,\mathfrak K}^{g,\lambda}\) is equicoercive and
\begin{equation}
 \mathcal H_{\eps,\mathfrak K}^{g,\lambda}
 \mathop{\longrightarrow}^{\Gamma}
 m_1\mathcal H_{0,\mathfrak K}^{g}
 \quad\text{in }L^1(\mathbb T^d).
 \label{eq:long-range-first-moment-gamma}
\end{equation}
Every binary finite-energy state is recovered by the constant sequence.
Thus compact radial support is not part of the local Griffith mechanism.
\end{theorem}

\begin{proof}
For \(R>\max\{b,a^{-1}\}\), let
\(\lambda_R=\lambda\mathbin{\llcorner}(R^{-1},R)\) and
\(m_R=\int t\lambda_R(\dd t)\).  Because
\((a,b)\subset(R^{-1},R)\), the measure \(\lambda_R\) satisfies
\eqref{eq:finite-window-spreading}; the steepness bound for \(\lambda\)
implies \eqref{eq:finite-window-compression-steepness} for every truncation.  Proposition
\ref{prop:finite-radial-window-main} therefore gives
\begin{equation}
 \mathcal H_{\eps,\mathfrak K}^{g,\lambda_R}
 \mathop{\longrightarrow}^{\Gamma}
 m_R\mathcal H_{0,\mathfrak K}^{g}.
 \label{eq:long-range-window-convergence}
\end{equation}
Equicoercivity of the full family follows already from the fixed compact
submeasure \(\lambda\mathbin{\llcorner}(a,b)\): the full energy dominates its
finite-window energy, with the same potential term.

Positivity therefore yields, for every convergent sequence,
\[
 \liminf_{\eps\downarrow0}
 \mathcal H_{\eps,\mathfrak K}^{g,\lambda}(u_\eps)
 \ge m_R\mathcal H_{0,\mathfrak K}^g(u).
\]
Let \(R\uparrow\infty\); monotone convergence gives
\(m_R\uparrow m_1\).

For recovery, fix binary \(u\) with finite limit energy.  The ordinary
\(BV\) translation estimate, the exact crossing formula, and
\eqref{eq:rectifiable-cut-mass} give constants \(h_0,C_u>0\) such that
\begin{equation}
 \sum_iB_{i,\mathfrak K}^g(u;h)\le C_u h
 \quad(0<h<h_0),
 \qquad
 \sum_iB_{i,\mathfrak K}^g(u;h)\le C_u
 \quad(h\ge h_0).
 \label{eq:long-range-bond-majorant}
\end{equation}
The radial complement of \((R^{-1},R)\) with \(\eps t\le h_0\) is bounded
by \(C_u\int t\lambda(\dd t)\) over that complement.  On
\(\{\eps t>h_0\}\), Markov's inequality gives
\[
 \frac{C_u}{\eps}\lambda((h_0/\eps,\infty))
 \le\frac{C_u}{h_0}\int_{h_0/\eps}^\infty t\lambda(\dd t)
 \longrightarrow0.
\]
First send \(\eps\downarrow0\) and then \(R\uparrow\infty\).  The window
limit and \eqref{eq:long-range-bond-majorant} prove constant recovery.
\end{proof}

The spreading assumption is a compactness condition, not a disguised moment
condition.

\begin{proposition}[Exact-shift resonance]
\label{prop:long-range-resonance-main}
On \(\mathbb T^d\), \(d\ge2\), take the empty cut, the constant spanning
frame \(v_i=e_i\), \(c_i=1\), \(i=1,\ldots,d\), and
\(\lambda=t_0^{-1}\delta_{t_0}\).  There are \(\eps_m\downarrow0\) and
binary states of zero interaction energy with no strongly convergent
\(L^1\) subsequence.  Consequently a finite first moment and uniform spanning
do not imply equicoercivity without radial anti-resonance.
\end{proposition}

\begin{proof}
Set \(\eps_m=(mt_0)^{-1}\) and let \(z_m(x)\) be the period-
\(m^{-1}\) binary stripe in the first coordinate that equals one on the first
half of each period.  Translation by \(\eps_mt_0e_1=e_1/m\) fixes \(z_m\),
and translation in every \(e_i\), \(i\ge2\), also fixes it because the state
is independent of those coordinates.  Thus every directional interaction and
the well term vanish.  The stripes converge weakly to \(1/2\).  A strongly
convergent subsequence of binary functions would have a binary almost-
everywhere limit, which is impossible.
\end{proof}

We next record the sharp threshold when the first moment diverges.  Assume for
the rest of the section that every flow preserves volume, and put
\begin{equation}
 \lambda_s(\dd t)=a_s\one_{[1,\infty)}(t)t^{-1-2s}\dd t,
 \qquad 0<s<1,\quad a_s>0.
 \label{eq:power-tail-main}
\end{equation}
For binary \(u\), define the one-scale quotient
\begin{equation}
 Q_{\mathfrak K}^g(u;h):=\frac1h\sum_iB_{i,\mathfrak K}^g(u;h).
 \label{eq:one-scale-quotient-main}
\end{equation}

\begin{lemma}[Uniform one-scale resolved lower bound]
\label{lem:one-scale-resolved-main}
If binary \(u_n\to u\) in \(L^1(\mathbb T^d)\) and \(h_n\downarrow0\), then
\begin{equation}
 \liminf_{n\to\infty}Q_{\mathfrak K}^g(u_n;h_n)
 \ge\mathcal H_{0,\mathfrak K}^g(u).
 \label{eq:one-scale-resolved-liminf-main}
\end{equation}
For every binary finite-energy \(u\),
\begin{equation}
 Q_{\mathfrak K}^g(u;h)\longrightarrow
 \mathcal H_{0,\mathfrak K}^g(u),
 \qquad
 \sup_{0<h<h_0}Q_{\mathfrak K}^g(u;h)<\infty.
 \label{eq:one-scale-resolved-recovery-main}
\end{equation}
\end{lemma}

\begin{proof}
Apply the periodic part of
Theorem~\ref{thm:global-resolved-translation-contact} with \(\zeta=1\).
Its sequential statement is exactly
\eqref{eq:one-scale-resolved-liminf-main}, while its fixed-state translation
formula and majorant are exactly
\eqref{eq:one-scale-resolved-recovery-main}.
\end{proof}

\begin{lemma}[Two-scale fractional smoothing by ordered flows]
\label{lem:two-scale-fractional-smoothing}
Assume \eqref{eq:flow-ray-regularity}--
\eqref{eq:flow-ray-spanning}, and assume that the flows preserve volume.  Let
\(0<\alpha<1\), let \(z_\eps:\mathbb T^d\to[0,1]\), and suppose that, for
some \(A<\infty\) independent of \(\eps\),
\begin{equation}
 \|z_\eps-z_\eps\circ\Phi_i^{-q}\|_{L^1}
 \le A\bigl(|q|^\alpha+\eps^\alpha\bigr)
 \quad\text{for every }i\text{ and }0<|q|<q_0.
 \label{eq:two-scale-fractional-modulus}
\end{equation}
There are \(r_0,C>0\), depending only on the flow atlas, the frame bounds,
and \(\alpha\), such that for every \(0<\eps\le r<r_0\) one can construct
\(Z_{\eps,r}\in BV(\mathbb T^d)\) satisfying
\begin{align}
 \|Z_{\eps,r}-z_\eps\|_{L^1}
 &\le CA\bigl(r^\alpha+\eps^\alpha\bigr),
 \label{eq:two-scale-fractional-distance}\\
 |DZ_{\eps,r}|(\mathbb T^d)
 &\le CA\left(r^{\alpha-1}+\frac{\eps^\alpha}{r}\right).
 \label{eq:two-scale-fractional-variation}
\end{align}
In particular, if \(\eps_n\downarrow0\) and
\eqref{eq:two-scale-fractional-modulus} holds with a common \(A\), then
\((z_{\eps_n})_n\) is relatively compact in \(L^1(\mathbb T^d)\).
\end{lemma}

\begin{proof}
We repeat the parameter-space part of the ordered-flow construction, but the
averaging radius is now \(r\), not the microscopic cutoff \(\eps\).  Retain
the finite atlas \(U_\beta^-\Subset U_\beta^+\), the selected ordered frames
\(I_\beta=(i_{\beta,1},\ldots,i_{\beta,d})\), the partition
\(\{\chi_\beta\}\), and the kernel
\(\eta\in C_c^1((-c_\rho,c_\rho))\) from the proof of
Lemma~\ref{lem:ordered-flow-smoothing}.  Decrease \(r_0\), if necessary, so
that \(r_0c_\rho<q_0\) and all ordered flow paths remain in the larger chart.
Set
\begin{align}
 T_{\beta,r}(x,s)
 &:={\Phi}_{i_{\beta,d}}^{-rs_d}\circ\cdots\circ
 {\Phi}_{i_{\beta,1}}^{-rs_1}(x),\notag\\
 A_{\beta,r}z(x)&:=\int_{\mathbb R^d}\eta_d(s)
 z(T_{\beta,r}(x,s))\dd s,
 \qquad \eta_d(s):=\prod_{k=1}^d\eta(s_k).
 \label{eq:two-scale-ordered-average}
\end{align}

The matrices \(\mathsf A_{\beta,r}=D_xT_{\beta,r}\) and
\(\mathsf B_{\beta,r}=-r^{-1}D_sT_{\beta,r}\) obey
\eqref{eq:ordered-parameter-freezing}--
\eqref{eq:ordered-parameter-matrix-bounds} with \(\eps\) replaced by \(r\).
Thus \(\mathsf B_{\beta,r}\) is uniformly invertible for \(r<r_0\), and the
coefficients
\(a_{\beta,\ell}=\mathsf B_{\beta,r}^{-1}
\mathsf A_{\beta,r}e_\ell\) have uniform \(W^{1,\infty}\) bounds in the
parameter variable.  Parameter-space integration by parts gives the exact
identity
\begin{equation}
 \partial_\ell A_{\beta,r}z(x)
 =\frac1r\sum_{k=1}^d\int
 \partial_{s_k}(\eta_da_{\beta,\ell k})(x,s)
 [z(T_{\beta,r}(x,s))-z(x)]\dd s.
 \label{eq:two-scale-parameter-ibp}
\end{equation}
As in \eqref{eq:ordered-path-telescoping}, insert the intermediate states
after the first \(k\) ordered flows.  Volume preservation makes every change
of the physical variable an \(L^1\) isometry, and hence
\begin{equation}
 \int_{U_\beta^-}|z_\eps(T_{\beta,r}(x,s))-z_\eps(x)|\dd x
 \le\sum_{k=1}^d
 \|z_\eps-z_\eps\circ\Phi_{i_{\beta,k}}^{-rs_k}\|_{L^1}.
 \label{eq:two-scale-telescoping}
\end{equation}
The right-hand side is bounded by
\(CA(r^\alpha\sum_k|s_k|^\alpha+d\eps^\alpha)\) by
\eqref{eq:two-scale-fractional-modulus}.  Integrating first with weight
\(\eta_d\), and then with the derivative weights in
\eqref{eq:two-scale-parameter-ibp}, proves
\begin{align}
 \|A_{\beta,r}z_\eps-z_\eps\|_{L^1(U_\beta^-)}
 &\le CA(r^\alpha+\eps^\alpha),
 \label{eq:two-scale-local-distance}\\
 |D(A_{\beta,r}z_\eps)|(U_\beta^-)
 &\le CA\left(r^{\alpha-1}+\frac{\eps^\alpha}{r}\right).
 \label{eq:two-scale-local-variation}
\end{align}
The computation is first valid for smooth \(z_\eps\).  Physical convolution,
continuity of composition by a bi-Lipschitz flow in \(L^1\), and lower
semicontinuity of variation remove that temporary regularization.

Define
\(Z_{\eps,r}=\sum_\beta\chi_\beta A_{\beta,r}z_\eps\).  The partition
identity \eqref{eq:partition-cancellation-identity}, finite atlas
multiplicity, and \eqref{eq:two-scale-local-distance}--
\eqref{eq:two-scale-local-variation} prove
\eqref{eq:two-scale-fractional-distance}--
\eqref{eq:two-scale-fractional-variation}.  The partition-gradient term is
bounded by the right-hand side of
\eqref{eq:two-scale-fractional-variation} after decreasing \(r_0\le1\).

It remains to justify the compactness conclusion, because the variation bound
is not uniform as \(r\downarrow0\).  Given \(\delta>0\), choose one fixed
\(r\in(0,r_0)\) so small that \(2CAr^\alpha<\delta/2\).  For all sufficiently
small \(\eps\le r\),
\eqref{eq:two-scale-fractional-distance} places \(z_\eps\) within
\(\delta/2\) of a set bounded in \(BV(\mathbb T^d)\) by the finite constant
\(2CAr^{\alpha-1}\).  That set has a finite \(\delta/2\)-net in \(L^1\) by
BV compactness.  Adding the finitely many initial members of the sequence
proves total boundedness, hence relative compactness in \(L^1\).
\end{proof}

For \(0<s<1/2\), define the flow-fractional resolved-cut functional
\begin{equation}
 \mathcal F_{s,\mathfrak K}^{g}(u)
 :=a_s\sum_i\int_0^\infty
 B_{i,\mathfrak K}^g(u;h)h^{-1-2s}\dd h
 \label{eq:fractional-resolved-limit-main}
\end{equation}
on binary states and give it value \(+\infty\) otherwise.

\begin{theorem}[Power-tail threshold trichotomy]
\label{thm:power-tail-main}
Assume \eqref{eq:rectifiable-cut-mass},
\eqref{eq:flow-ray-regularity}--\eqref{eq:flow-ray-spanning}, and that every
flow preserves volume.  Choose the steep potential so that
\begin{equation}
 M\ge4\lambda_s((0,\infty))\sum_i c^*.
 \label{eq:power-tail-steepness-main}
\end{equation}
\begin{enumerate}
 \item If \(s>1/2\), then
 \[
  \mathcal H_{\eps,\mathfrak K}^{g,\lambda_s}
  \mathop{\longrightarrow}^{\Gamma}
  \frac{a_s}{2s-1}\mathcal H_{0,\mathfrak K}^g.
 \]
 \item If \(s=1/2\), then
 \[
  \frac1{|\log\eps|}
  \mathcal H_{\eps,\mathfrak K}^{g,\lambda_{1/2}}
  \mathop{\longrightarrow}^{\Gamma}
  a_{1/2}\mathcal H_{0,\mathfrak K}^g.
 \]
 \item If \(0<s<1/2\), then
 \[
  \eps^{1-2s}\mathcal H_{\eps,\mathfrak K}^{g,\lambda_s}
  \mathop{\longrightarrow}^{\Gamma}
  \mathcal F_{s,\mathfrak K}^{g}.
 \]
\end{enumerate}
All three families are equicoercive in their displayed normalization.  Every
binary finite-energy target is recovered by the constant sequence.  The third
limit is nonlocal: it retains the complete physical-time crossing rule.
\end{theorem}

\begin{proof}
For binary states the substitution \(h=\eps t\) gives
\begin{equation}
 \frac{a_s}{\eps}\int_1^\infty
 B_{i,\mathfrak K}^g(u;\eps t)t^{-1-2s}\dd t
 =a_s\eps^{2s-1}\int_\eps^\infty
 B_{i,\mathfrak K}^g(u;h)h^{-1-2s}\dd h.
 \label{eq:power-tail-rescaling-main}
\end{equation}
If \(s>1/2\), the first moment of \(\lambda_s\) is
\(a_s/(2s-1)\), so Theorem~\ref{thm:long-range-first-moment-main} applies.

Let \(s=1/2\) and \(L_\eps=|\log\eps|\).  Threshold rounding is still
energy decreasing, and bounded normalized energy makes the thresholding
distance at most \(C\eps L_\eps\).  For compactness, partition
\([\eps^{3/4},\eps^{1/4}]\) into dyadic shells
\([r_k,2r_k]\).  There are at least \(cL_\eps\) such shells, while
\eqref{eq:power-tail-rescaling-main} bounds their total unnormalized energy by
\(CL_\eps\).  One shell therefore satisfies
\[
 \frac1{r_k}\sum_i\int_1^2
 B_{i,\mathfrak K}^g(u_\eps;r_kt)t^{-2}\dd t\le C.
\]
It is a finite-range direct energy at the vanishing scale \(r_k\), with a
radially spread law.  Moreover the well bound gives
\[
 \frac1{r_k}\int d_*(u_\eps)\dd x
 \le C\frac{\eps L_\eps}{r_k}
 \le C\eps^{1/4}L_\eps=o(1).
\]
After multiplying the shell law by a fixed normalization constant,
Corollary~\ref{cor:periodic-rectifiable-cut} gives compactness.  If
\(u_n\to u\) and \(\eps_n\downarrow0\), then for every
\(0<\delta<1\),
\[
 \frac{a_{1/2}}{L_{\eps_n}}
 \int_{\eps_n}^{\delta}Q_{\mathfrak K}^g(u_n;h)\frac{\dd h}{h}
 \ge a_{1/2}\frac{\log(\delta/\eps_n)}{L_{\eps_n}}
 \inf_{\eps_n\le h\le\delta}Q_{\mathfrak K}^g(u_n;h).
\]
Lemma~\ref{lem:one-scale-resolved-main} gives
\[
 \lim_{\delta\downarrow0}\liminf_{n\to\infty}
 \inf_{\eps_n\le h\le\delta}Q_{\mathfrak K}^g(u_n;h)
 \ge \mathcal H_{0,\mathfrak K}^g(u).
\]
Indeed, a failure would select a diagonal \(h_n\le\delta_n\downarrow0\)
contradicting the lemma.  Since
\(\log(\delta/\eps_n)/L_{\eps_n}\to1\) for fixed \(\delta\), this proves the
critical liminf.  For a fixed target,
\eqref{eq:one-scale-resolved-recovery-main} and logarithmic Cesaro convergence
give constant recovery.

Finally let \(0<s<1/2\).  Write
\(E_{\eps,s}:=\eps^{1-2s}
\mathcal H_{\eps,\mathfrak K}^{g,\lambda_s}\).  Thresholding is energy
decreasing and
\[
 \|u_\eps-\chi(u_\eps)\|_1
 \le M^{-1}\eps^{2s}E_{\eps,s}(u_\eps)\longrightarrow0.
\]
After thresholding, a bounded normalized direct energy controls
\begin{equation}
 \sum_i\int_\eps^\infty
 \|u-u\circ\Phi_i^{-h}\|_1h^{-1-2s}\dd h
 \label{eq:fractional-truncated-budget-main}
\end{equation}
up to a uniform constant.  Indeed, direct and uncut bonds differ only on arcs
meeting \(K\), and their difference is bounded by
\(C\min\{h,1\}\), which is integrable against \(h^{-1-2s}\dd h\).
Write the thresholded state as \(z_\eps\), set
\(f_{i,\eps}(q)=\|z_\eps-z_\eps\circ\Phi_i^{-q}\|_1\), and denote the
uniform bound in \eqref{eq:fractional-truncated-budget-main} by \(A\).
Volume preservation and the group property give
\begin{equation}
 f_{i,\eps}(q)\le f_{i,\eps}(r)+f_{i,\eps}(q+r)
 \qquad(q,r>0).
 \label{eq:fractional-flow-subadditivity}
\end{equation}
If \(q\ge\eps\), average \eqref{eq:fractional-flow-subadditivity} over
\(r\in[q,2q]\).  Since both resulting intervals lie in \([q,3q]\),
\begin{align}
 qf_{i,\eps}(q)
 &\le\int_q^{3q}f_{i,\eps}(h)\dd h\notag\\
 &\le(3q)^{1+2s}
 \int_q^{3q}f_{i,\eps}(h)h^{-1-2s}\dd h
 \le C A q^{1+2s}.
 \label{eq:fractional-modulus-large-q}
\end{align}
If \(0<q<\eps\), average over \(r\in[\eps,2\eps]\); both terms are then
contained in \([\eps,3\eps]\), with multiplicity at most two, and
\begin{equation}
 \eps f_{i,\eps}(q)
 \le2(3\eps)^{1+2s}
 \int_\eps^{3\eps}f_{i,\eps}(h)h^{-1-2s}\dd h
 \le CA\eps^{1+2s}.
 \label{eq:fractional-modulus-small-q}
\end{equation}
Consequently
\begin{equation}
 \|z_\eps-z_\eps\circ\Phi_i^{-q}\|_1
 \le CA(q^{2s}+\eps^{2s})
 \quad(0<q<q_0),
 \label{eq:fractional-uniform-flow-modulus}
\end{equation}
after fixing any sufficiently small \(q_0\).  Volume preservation also gives
\(f_{i,\eps}(-q)=f_{i,\eps}(q)\), so the same estimate holds for
\(0<|q|<q_0\).  Apply
Lemma~\ref{lem:two-scale-fractional-smoothing} with \(\alpha=2s\).  It gives
relative \(L^1\) compactness of every thresholded sequence
\(z_{\eps_n}\), \(\eps_n\downarrow0\); the vanishing threshold distance then
gives equicoercivity of the original normalized family.

For the lower bound, restrict
\eqref{eq:power-tail-rescaling-main}, after multiplication by
\(\eps^{1-2s}\), to \([\delta,H]\).  The direct bond is uniformly
\(L^1\)-continuous in the state on this interval because the crossing
classification is state independent and both flow marginals preserve volume.
Pass to the limit and then send \(\delta\downarrow0\), \(H\uparrow\infty\).
Monotone convergence gives \(\mathcal F_{s,\mathfrak K}^g\).  For a fixed
binary target, the same rescaling identity and monotone convergence give the
constant recovery sequence.
\end{proof}

\begin{remark}[Scope of the tail classification]
\label{rem:tail-classification-scope}
The theorem classifies the local first-moment regime and the complete power
family.  More general slowly varying or Tauberian tails require an explicit
anti-resonance condition: Proposition~\ref{prop:long-range-resonance-main}
shows that scalar moment growth alone cannot provide compactness.  No claim
for an arbitrary radius-direction coupled tail is made here.
\end{remark}

\section{Conclusion}

The sharp-interface limit of a nonlocal energy on a cut body is determined
only after the finite-scale model records how bonds cross the cut.  The direct
crossing construction does so on every finite-mass countably rectifiable cut,
without duplicating physical volume or introducing a smooth carrier.  Its
Gamma-limit is the resolved flow-frame variation plus the outer contact and
one coefficient-one contact on each labelled bank.  It contains no compulsory
term on surface-null branches, tips, or crack-boundary contact sets, and the
constant recovery sequence creates no concentration there.  This does not
exclude nonnegative excess defects along arbitrary bounded-energy sequences.

The proof separates the analytic mechanisms responsible for this conclusion.
The spacetime area formula controls the complete crossing multiplicity, while
rectifiable coarea makes multiple crossings negligible at first order.  The
split-line lower bound then supplies the sequential Gamma-liminf and the two
independent bank traces.  Ordered-flow smoothing supplies compactness for
noncommuting directions.  The original positive-reach slit and fan cells remain
useful as explicit realizations, and their detectors are uniformly equivalent
to the direct rule at Gamma scale.

Compact radial support is also unnecessary.  A finite first moment is enough
to recover the local resolved Griffith functional.  The power family shows
that this condition is sharp for surface scaling: the critical tail requires
logarithmic renormalization, and heavier tails converge to a genuinely
fractional functional that remembers crossings at every physical scale.
Exact-shift resonance explains why a scalar moment assumption alone cannot
replace radial spreading in the compactness argument.

In Griffith terminology, these limits identify the anisotropic surface and
bank-contact cost that can enter a total free energy.  They do not by
themselves compute a bulk configurational energy-release rate: that quantity
requires a coupled elastic energy and an admissible class of irreversible cut
variations.  Thus the result supplies the rigorous surface side of a Griffith
comparison, while crack evolution and propagation equality remain separate
questions.

The scope remains precise.  The cut is fixed, rectifiable, and of finite total
surface mass; its bank labels are compatible on surface overlaps.  Purely
unrectifiable sets require a different limiting object because classical
normals and two-sided \(BV\) traces may not exist.  Arbitrary path-lifting
algorithms, radius-direction coupled remote tails, and topology-changing
motions are not classified here.  Appendix
\ref{app:moving-sectorial-junction} proves stability under uniform
label-preserving transport with fixed incidence and identifies the geometric
degeneration that leaves that class.

\appendix

\section{Uniform transport of resolved junctions}
\label{app:moving-sectorial-junction}
\label{subsec:moving-sectorial-junction}

We next allow the preceding fan to move on a macroscopic time interval while
keeping its incidence graph fixed.  To distinguish evolution time from the
bond-length variable in the flow pairs, we write
\(\tau\in[0,T]\) for evolution time and \(\ell\in[a,b]\) for the latter.
The result below is a fixed-topology transport theorem.  It does
not cover nucleation, merging, splitting, or disappearance of faces.

\paragraph{Route through this subsection.}
First, the maps \(P_\tau\) identify every moving cut completion with one
reference state space.  Second, the gated detector and pair measures are
transported by exact conjugacy, so resolved-sheet and sector labels are independent of
time after pullback.  Third, uniform angle, reach, chart, and coefficient
bounds yield a time-seam estimate for the pulled-back energies
\(\widetilde{\mathcal G}_\eps^\tau\).  The fixed-time fan theorem and this
seam estimate then give uniform Gamma-convergence for arbitrary joint
sequences \((\eps_n,\tau_n)\), including
\(\eps_n\ll|\tau_n-\tau|\).  The closing corollary constructs one recovery
operator, and the final proposition identifies the obstruction to topology
change within this transport class.

\paragraph{Uniform transport class and common state space.}
Start with the fan \(\mathcal K(0)=\bigcup_{\alpha=1}^NK_\alpha(0)\) and
\(\Gamma(0)\) of \eqref{eq:fan-incidence}--\eqref{eq:fan-angle-gap}.  Let
\(\Phi_\tau:\overline\Omega\to\overline\Omega\) be orientation-preserving
bi-\(C^{1,1}\) homeomorphisms such that
\begin{align}
 &\Phi_0=\operatorname{id},\qquad
 \Phi_\tau(\partial\Omega)=\partial\Omega,\qquad
 K_\alpha(\tau)=\Phi_\tau(K_\alpha(0)),\quad
 \Gamma(\tau)=\Phi_\tau(\Gamma(0)),
 \label{eq:moving-fan-transport}\\
 &\sup_{\tau\in[0,T]}
 \left(\|\Phi_\tau\|_{C^{1,1}}+\|\Phi_\tau^{-1}\|_{C^{1,1}}\right)
 \le M_*,\qquad
 \Phi\in W^{1,\infty}(0,T;C^1(\overline\Omega)).
 \label{eq:moving-fan-uniform-flow}
\end{align}
Equivalently, one may take \(\Phi_\tau\) to be the flow of a field
\(b\in L^\infty(0,T;C^{1,1})\), tangent to \(\partial\Omega\), whose flow and
inverse satisfy \eqref{eq:moving-fan-uniform-flow}.  Put
\begin{equation}
 \omega_\Phi(\delta):=
 \sup_{|\tau-\sigma|\le\delta}
 \left(\|\Phi_\tau-\Phi_\sigma\|_{C^1}
 +\|\Phi_\tau^{-1}-\Phi_\sigma^{-1}\|_{C^1}\right).
 \label{eq:moving-fan-time-modulus}
\end{equation}
Thus \(\omega_\Phi(\delta)\le C\delta\).  The weaker assumption that the
right-hand side merely tends to zero gives the same conclusions without a
linear rate.

Transport the finite fan charts by \(\Phi_\tau\).  Their number, overlap
multiplicity, labels, and sector incidence are unchanged.  After decreasing
the reference radius once, their radii are bounded below by one
\(\rho_{\rm mov}>0\), their direct and inverse \(C^{1,1}\) norms are bounded
by \(C(M_*,\rho_*)\), and
\begin{equation}
 \inf_\tau\operatorname{reach}\Gamma(\tau)\ge r_\Gamma>0,
 \qquad
 \inf_{\tau,z,\alpha}
 \bigl(\vartheta_{\alpha+1}(\tau,z)-\vartheta_\alpha(\tau,z)\bigr)
 \ge\theta_{\rm mov}>0.
 \label{eq:moving-fan-reach-angle}
\end{equation}
These constants can be read quantitatively from the reference constants and
\(M_*\); for example
\(\sin(\theta_{\rm mov}/2)\ge C(M_*)^{-1}
\sin(\theta_*/2)\).  In a transported chart,
\begin{equation}
 |q_\alpha^\tau(z,\varrho)-q_\beta^\tau(z,\varrho)|
 \ge 2c_*\varrho\sin(\theta_{\rm mov}/2)
      -C_*\varrho^2/\rho_{\rm mov},
 \label{eq:moving-fan-uniform-angle-separation}
\end{equation}
uniformly in \(\tau\).  The linear term is the bi-Lipschitz image of the
reference ray separation, and the quadratic term is the common chart
remainder.

The transport lifts uniquely to a label-preserving bi-Lipschitz map
\begin{equation}
 P_\tau:X_{\mathcal K(0)}\longrightarrow X_{\mathcal K(\tau)},
 \qquad p_\tau\circ P_\tau=\Phi_\tau\circ p_0.
 \label{eq:moving-fan-lift}
\end{equation}
It sends each reference face copy and sector vertex to the copy with the same
label.  We therefore use the fixed reference state space
\(L^1(X_{\mathcal K(0)},m_0)\): a physical state \(u^\tau\) is represented by
\(v^\tau=u^\tau\circ P_\tau\), with the transported Jacobian included in the
energy.  Uniform upper and lower Jacobian bounds make this norm uniformly
equivalent to every physical \(L^1\) norm.

\paragraph{A time-uniform finite-\(\eps\) detector.}
The detector must be transported with the geometry.  If \(\gamma_{i,x,\ell}
^\eps\) is a reference interaction arc, define
\begin{equation}
 \gamma_{i,\Phi_\tau(x),\ell}^{\eps,\tau}(s)
 :=\Phi_\tau\bigl(\gamma_{i,x,\ell}^\eps(s)\bigr),
 \qquad 0\le s\le\eps\ell,
 \label{eq:moving-fan-conjugated-arc}
\end{equation}
This is the flow arc of the pushed field
\begin{equation}
 v_i^\tau(\Phi_\tau(q)):=D\Phi_\tau(q)v_i(q).
 \label{eq:moving-fan-pushed-field}
\end{equation}
and push forward the symmetric pair measure and its transpose by
\(\Phi_\tau\times\Phi_\tau\).  The reference gate
\(T_{r_\eps}(\Gamma(0))\) and its endpoint-sector labels are pushed forward
with the arc; in particular, the moving core gate is
\(\Phi_\tau(T_{r_\eps}(\Gamma(0)))\), rather than a newly selected Euclidean
tube.  Bounded positive time-dependent amplitudes may
be included as pulled-back multipliers \(\alpha_i^\tau\) of the directed bond
measures and \(\beta^\tau\) of the potential, provided
\begin{align}
 0<a_*&\le\alpha_i^\tau,\beta^\tau\le a^*<\infty,\qquad
 \sup_{\tau,i}\bigl(
 \|\alpha_i^\tau\|_{W^{1,\infty}}+
 \|\beta^\tau\|_{W^{1,\infty}}\bigr)\le A_*,
 \notag\\
 \sup_i\left\|
 \frac{\alpha_i^\tau}{\alpha_i^\sigma}-1
 \right\|_{L^\infty}
 &\le\omega_c(|\tau-\sigma|),\notag\\
 \left\|\frac{\beta^\tau}{\beta^\sigma}-1
 \right\|_{L^\infty}
 &\le\omega_c(|\tau-\sigma|),
 \label{eq:moving-fan-relative-amplitude-modulus}
\end{align}
uniformly in \(i\), where \(\omega_c(\delta)\to0\).  Fixed amplitudes
correspond to \(\omega_c=0\).  Write
\(J_\tau(q):=|\det D\Phi_\tau(q)|\).  After pullback, the potential density is
\(M\beta^\tau J_\tau d_*(v)/\eps\), while the directed pair measures carry
the multipliers \(\alpha_i^\tau\).  Threshold rounding is therefore uniform
only under the explicit steepness condition
\begin{equation}
 M a_*\inf_{\tau,q}J_\tau(q)
 \ge 2a^*\Lambda\sum_i c^*
       \bigl(1+e^{\eps_0bD_i}\bigr).
 \label{eq:moving-fan-uniform-steepness}
\end{equation}
This condition reduces to \eqref{eq:flow-ray-steepness} for unit amplitudes
and volume-preserving transport.  The resolved-sheet detector satisfies
\begin{equation}
 \mathfrak a_\eps^\tau
 \bigl(\Phi_\tau\circ\gamma^\eps\bigr)
 =\mathfrak a_\eps^0(\gamma^\eps),
 \label{eq:moving-fan-detector-conjugacy}
\end{equation}
including the value \(\varnothing\), and preserves both endpoint sector
labels when the value is nonempty.  Consequently the resolved branch is the
transported own-face rule and pair reversal exchanges its two endpoint
summands.  This exact conjugacy is stronger than pointwise stability and
is what permits an arbitrary common limit \(\tau_n\to\tau\),
\(\eps_n\to0\), without a condition relating the two rates.

Write \(\widetilde{\mathcal G}_\eps^\tau\) for the transported version of
\(\mathcal G_{\eps,\mathcal K(\tau)}^{g(\tau)}\), pulled back by \(P_\tau\),
and \(\widetilde{\mathcal G}_0^\tau\) for the pullback of
\eqref{eq:fan-limit-energy}.  At \(\tau\ne0\), ``transported'' means that the
gate is the pushed reference gate just defined; it need not equal a freshly
chosen Euclidean distance tube.  Uniform bi-Lipschitz bounds give constants
\(0<c<C<\infty\) such that
\[
 T_{c r_\eps}(\Gamma(\tau))
 \subset\Phi_\tau(T_{r_\eps}(\Gamma(0)))
 \subset T_{C r_\eps}(\Gamma(\tau)).
\]
Therefore the fixed-time proof and its junction estimate apply unchanged,
while exact pullback conjugacy is retained.  We impose the following explicit datum
hypothesis: the outer datum, every labelled face datum, and every chosen
opposite-collar ghost extension are the pushforwards by \(\Phi_\tau\) of
fixed binary \(BV\) data on the reference geometry.  Thus all prescribed
fields are independent of \(\tau\) after pullback.  This label-preserving
transport is essential for the arbitrary joint-diagonal statement below;
mere \(L^1\) continuity of binary ghosts does not control a
\(1/\eps\)-scaled contact layer uniformly in \(\eps\).  For a reference
normal \(n\), the pushed starting-point coefficient is
\begin{equation}
 c_i^\tau(\Phi_\tau(q))
 :=\frac{\alpha_i^\tau(q)c_i(q)}{J_\tau(q)},
 \qquad
 h_{\rm fl}^\tau(y,\xi):=
 \sum_i c_i^\tau(y)|v_i^\tau(y)\cdot\xi|.
 \label{eq:moving-fan-physical-density}
\end{equation}
Set
\begin{align}
 n_\tau(q,n)&:=
 \frac{\operatorname{cof}D\Phi_\tau(q)n}
 {|\operatorname{cof}D\Phi_\tau(q)n|},\notag\\
 \overline h_\tau(q,n)&:=
 |\operatorname{cof}D\Phi_\tau(q)n|\,
 h_{\rm fl}^\tau(\Phi_\tau(q),n_\tau(q,n)).
 \label{eq:moving-fan-pulled-density}
\end{align}
Because the physical pair measure is the pushforward of the reference
directed measure, the area and volume Jacobians cancel exactly:
\begin{equation}
 \overline h_\tau(q,n)
 =\sum_i\alpha_i^\tau(q)c_i(q)|v_i(q)\cdot n|.
 \label{eq:moving-fan-pulled-density-cancellation}
\end{equation}
Uniform spanning is preserved, and the coefficient/geometry assumptions give
\begin{equation}
 \|\overline h_\tau-\overline h_\sigma\|_
 {C(\overline\Omega\times\Sph^{d-1})}
 \le C\bigl(\omega_\Phi(|\tau-\sigma|)
             +\omega_c(|\tau-\sigma|)\bigr).
 \label{eq:moving-fan-density-time-continuity}
\end{equation}
Every transported face still has the directed-plus-transposed coefficient
\(1/2+1/2=1\).  Hence the numerical face coefficient is constant in time,
while its complete surface density \(\overline h_\tau\) depends continuously
on time by \eqref{eq:moving-fan-density-time-continuity}.

We collect the uniform constants used below.  Let
\begin{align}
 J_-&:=\inf_{\tau,q}|\det D\Phi_\tau(q)|,&
 J_+&:=\sup_{\tau,q}|\det D\Phi_\tau(q)|,\notag\\
 V_{0,{\rm mov}}&:=\sup_{\tau,i}\|v_i^\tau\|_\infty,&
 V_{1,{\rm mov}}&:=\sup_{\tau,i}\|Dv_i^\tau\|_\infty,&
 D_{0,{\rm mov}}&:=\sup_{\tau,i}\|\operatorname{div}v_i^\tau\|_\infty,
 \label{eq:moving-fan-constant-data}
\end{align}
and let \(D_g\) be one plus the sum of the \(BV\) norms of the fixed
reference outer datum and all labelled opposite-collar ghost extensions.
All are finite by the hypotheses above.  If \(m_{\rm fan}\) is the reference
atlas multiplicity, define the non-optimal but explicit bound
\begin{align}
 C_{\rm mov}:={}&C_{\rm fan}
 \frac{a^*}{a_*c_*}\left(1+\frac{J_+}{J_-}\right)
 (1+M_*)^{4d+6}
 \bigl(1+e^{\eps_0bD_{0,{\rm mov}}}\bigr)\notag\\
 &\times(1+A_*+V_{0,{\rm mov}}+V_{1,{\rm mov}})
 (1+D_g).
 \label{eq:moving-fan-explicit-constant}
\end{align}
The angle, reach, front area, sheet area, and seam area remain displayed
separately in the estimates below.  Thus \(C_{\rm mov}\) does not conceal
any additional degeneration as
\(\theta_{\rm mov}\downarrow0\) or \(\rho_{\rm mov}\downarrow0\).

\begin{lemma}[Uniform junction and time-seam estimates]
\label{lem:moving-fan-uniform-estimates}
The following estimates hold with \(C=C_{\rm mov}\), independently of
\(\tau,\eps,\varrho,\eta\).  Here
\(E_{\rm chart}^\tau\) denotes the absolute interaction difference produced
by flattening the transported fan charts and freezing the transported flows,
and \(E_{\rm atlas}^\tau\) the interaction discarded in the hard-atlas
\(\eta\)-strips.  Both are evaluated on binary states, so their integrands are
bounded by two:
\begin{equation}
 B_{i,\eps,\ell}^{\rm jun}(\varrho;\tau)
 :=(A_{i,\eps,\ell}^{\rm jun}(\varrho;\tau)\times U)
 \cup(U\times A_{i,\eps,\ell}^{\rm jun}(\varrho;\tau)).
 \label{eq:moving-fan-active-pair-set}
\end{equation}
\begin{align}
 \sup_{\tau\in[0,T]}|A_{i,\eps,\ell}^{\rm jun}
       (\varrho;\tau)|
 &\le C_{\rm mov}\eps\ell\left(\varrho+\eps\ell
       \csc(\theta_{\rm mov}/2)\right)
       \sup_\tau\Hh^{d-2}(\Gamma(\tau)),
 \label{eq:moving-fan-active-volume}\\
 \sup_\tau\frac1\eps\widehat Q_{i,\eps,\ell}^\tau
 \bigl(B_{i,\eps,\ell}^{\rm jun}(\varrho;\tau)\bigr)
 &\le C_{\rm mov}\ell\left(\varrho+\eps\ell
       \csc(\theta_{\rm mov}/2)\right).
 \label{eq:moving-fan-normalized-junction}\\
 \sup_\tau E_{\rm chart}^\tau(\eps)
 &\le C_{\rm mov}\eps\left(\rho_{\rm mov}^{-1}
       +\sup_{\tau,i}\operatorname{Lip}v_i^\tau\right)
       \sup_\tau\sum_\alpha\Hh^{d-1}(K_\alpha(\tau)),
 \label{eq:moving-fan-chart-estimate}\\
 \sup_\tau E_{\rm atlas}^\tau(\eps,\eta)
 &\le C_{\rm mov}(\eta+\eps)
       \sup_\tau\Hh^{d-2}(\Sigma_{\mathcal K(\tau)}).
 \label{eq:moving-fan-atlas-estimate}
\end{align}
For binary \(v\) with finite energy, the finite-scale time seam satisfies
\begin{equation}
 |\widetilde{\mathcal G}_\eps^\tau(v)-
   \widetilde{\mathcal G}_\eps^\sigma(v)|
 \le C_{\rm mov}\omega_*(|\tau-\sigma|)
       \bigl(1+\widetilde{\mathcal G}_\eps^\sigma(v)\bigr),
 \quad
 \omega_*:=\omega_\Phi+\omega_c,
 \label{eq:moving-fan-time-seam}
\end{equation}
uniformly in \(\eps\).  Finally, the uniform junction cell
uses the localized interaction
\begin{equation}
 E_{\eps,\tau}^{\rm jun}(\varrho;v):=
 \sum_i\int_a^b\frac1\eps
 \iint_{B_{i,\eps,\ell}^{\rm jun}(\varrho;\tau)}
 \Psi_{\eps,\mathcal K(\tau)}^{i,\ell}[v](x,y)
 \widehat Q_{i,\eps,\ell}^\tau(\dd x,\dd y)\lambda(\dd\ell).
 \label{eq:moving-fan-localized-energy}
\end{equation}
For binary states the potential vanishes, so this is the whole contribution
that can be forced by the moving junction detector.  Define
\begin{equation}
 j_{\rm mov}[v,g]:=
 \lim_{\varrho\downarrow0}\limsup_{\eps\downarrow0}
 \sup_{\tau\in[0,T]}E_{\eps,\tau}^{\rm jun}(\varrho;v)
 \label{eq:moving-fan-junction-cell}
\end{equation}
is zero.  The direct choice \(\varrho_\eps=\sqrt\eps\) gives the uniform
remainder
\begin{equation}
 C_{\rm mov}\left(\sqrt\eps+\eps\csc(\theta_{\rm mov}/2)
 +\eps/\rho_{\rm mov}+\eta+\eps\right).
 \label{eq:moving-fan-total-estimate}
\end{equation}
\end{lemma}

\begin{proof}
Estimate \eqref{eq:moving-fan-uniform-angle-separation} replaces
\eqref{eq:fan-angle-separation}; the Fubini argument of Lemma
\ref{lem:fan-junction-estimate} then has uniform chart Jacobians and gives
\eqref{eq:moving-fan-active-volume}--
\eqref{eq:moving-fan-normalized-junction}.  The resolved branch and the
canonical endpoint core rule are both bounded by two.
Taking
first \(\eps\downarrow0\), uniformly in \(\tau\), and then
\(\varrho\downarrow0\) proves \eqref{eq:moving-fan-junction-cell}.

The signed-distance, Jacobian, and flow expansions are uniform under
\eqref{eq:moving-fan-uniform-flow} and give
\eqref{eq:moving-fan-chart-estimate}; finite atlas multiplicity gives
\eqref{eq:moving-fan-atlas-estimate}.

For the time seam, set
\(\delta_{\tau\sigma}:=\omega_*(|\tau-\sigma|)\).  The determinant
inequality
\begin{equation}
 |\det A-\det B|
 \le d\max\{\|A\|,\|B\|\}^{d-1}\|A-B\|
 \label{eq:moving-fan-determinant-difference}
\end{equation}
and the lower bound \(J_->0\) imply
\begin{equation}
 \left\|\frac{J_\tau}{J_\sigma}-1\right\|_\infty
 +\sup_{|n|=1}\left\|
 \frac{|\operatorname{cof}D\Phi_\tau n|}
 {|\operatorname{cof}D\Phi_\sigma n|}-1
 \right\|_\infty
 \le C_J\omega_\Phi(|\tau-\sigma|),
 \label{eq:moving-fan-relative-jacobian}
\end{equation}
where one may take
\begin{equation}
 C_J:=C_d\left(1+\frac{J_+}{J_-}\right)(1+M_*)^{2d+2}.
 \label{eq:moving-fan-jacobian-constant}
\end{equation}

Index the finitely many pulled-back ordinary, face, and junction bond types
by \(c\).  Exact detector conjugacy in
\eqref{eq:moving-fan-detector-conjugacy} gives a reference pair measure
\(\mu_{\eps,c}\) and a nonnegative binary integrand
\(\Xi_{\eps,c}(v)\), both independent of time.  The corresponding cell has
the form
\begin{equation}
 E_{\eps,c}^\tau(v)
 =\frac1\eps\int a_c^\tau(q)
 \Xi_{\eps,c}(v;q)\,\mu_{\eps,c}(\dd q),
 \label{eq:moving-fan-reference-cell}
\end{equation}
with all geometry and amplitudes contained in the positive weight
\(a_c^\tau\).  Equations
\eqref{eq:moving-fan-relative-amplitude-modulus} and
\eqref{eq:moving-fan-relative-jacobian}, together with
\(c_i\ge c_*\), yield
\begin{equation}
 \left\|\frac{a_c^\tau}{a_c^\sigma}-1\right\|_\infty
 \le C_{\rm rel}\delta_{\tau\sigma},
 \qquad
 C_{\rm rel}:=C_d\frac{a^*c^*}{a_*c_*}
 \left(1+\frac{J_+}{J_-}\right)(1+M_*)^{2d+2}.
 \label{eq:moving-fan-relative-cell-weight}
\end{equation}
Consequently,
\begin{equation}
 |E_{\eps,c}^\tau(v)-E_{\eps,c}^\sigma(v)|
 \le C_{\rm rel}\delta_{\tau\sigma}
 E_{\eps,c}^\sigma(v).
 \label{eq:moving-fan-cell-time-seam}
\end{equation}
The same calculation applies to the potential weight.  The pulled-back
binary data are time independent.  Any fixed-data terms introduced in the
local ambient-completion comparison are bounded by \(D_g\) using the \(BV\)
flow-translation estimate.
Summing \eqref{eq:moving-fan-cell-time-seam} over at most
\(m_{\rm fan}N\) active labelled cell types therefore gives
\begin{equation}
 |\widetilde{\mathcal G}_\eps^\tau(v)-
   \widetilde{\mathcal G}_\eps^\sigma(v)|
 \le C_{\rm mov}\delta_{\tau\sigma}
 \bigl(1+\widetilde{\mathcal G}_\eps^\sigma(v)\bigr),
 \label{eq:moving-fan-time-seam-summed}
\end{equation}
which is \eqref{eq:moving-fan-time-seam}.
The factor \(1/\eps\) in
\eqref{eq:moving-fan-reference-cell} multiplies the same
\(\Xi_{\eps,c}\) at both times.  No state, detector label, or node is shifted
when \(\tau\) is compared with \(\sigma\), so no term of size
\(|\tau-\sigma|/\eps\) is created.  Adding the four geometric estimates and
choosing \(\varrho=\sqrt\eps\) proves
\eqref{eq:moving-fan-total-estimate}.
\end{proof}

\paragraph{Why no ratio between time and interaction scales is needed.}
Let \(\delta_n:=|\tau_n-\tau|\) and suppose that
the binary states \(v_n\) satisfy
\(\widetilde{\mathcal G}_{\eps_n}^{\tau_n}(v_n)\le M\).  Taking
\(\sigma=\tau_n\) in \eqref{eq:moving-fan-time-seam} gives the quantitative
comparison
\begin{equation}
 \left|\widetilde{\mathcal G}_{\eps_n}^{\tau}(v_n)
       -\widetilde{\mathcal G}_{\eps_n}^{\tau_n}(v_n)\right|
 \le C_{\rm mov}(1+M)\omega_*(\delta_n)\longrightarrow0,
 \label{eq:moving-fan-rate-free-diagonal}
\end{equation}
with no occurrence of \(\delta_n/\eps_n\).  Hence the estimate permits
\(\eps_n/\delta_n\to0\), \(\delta_n/\eps_n\to0\), or no limiting ratio at
all.  This rate independence uses all three transported structures above:
the detector and pair measure are exactly conjugated, prescribed data are
fixed after pullback, and only positive weights vary relatively by
\(O(\omega_*(\delta_n))\).

The conclusion would generally fail for a detector or contact layer moved in
physical coordinates without conjugation.  The one-dimensional layer
\(I_\eps=(0,\eps)\) already gives
\begin{equation}
 \frac1\eps\int_{\R}
 \left|\one_{I_\eps}(x-\delta)-\one_{I_\eps}(x)\right|\dd x
 =2\min\left\{\frac{|\delta|}{\eps},1\right\}.
 \label{eq:unconjugated-time-layer-counterexample}
\end{equation}
Thus \(\delta_n\to0\) alone does not control an unconjugated layer when
\(\eps_n\ll\delta_n\).  In that different model one would need a scale
condition such as \(\delta_n=o(\eps_n)\), or additional trace regularity that
removes the shifted layer.  The rate-free theorem below is therefore a
consequence of exact transport, not of time continuity alone.

\begin{theorem}[Uniform equicoercivity and Gamma-convergence in time]
\label{thm:moving-fan-uniform-gamma}
Under \eqref{eq:moving-fan-transport}--
\eqref{eq:moving-fan-reach-angle}, the conjugated detector
\eqref{eq:moving-fan-conjugated-arc}, and the coefficient and datum hypotheses
above, including \eqref{eq:moving-fan-uniform-steepness}, the family
\(\{\widetilde{\mathcal G}_\eps^\tau:\eps>0,\tau\in[0,T]\}\) is
equicoercive on \(L^1(X_{\mathcal K(0)},m_0)\).  Moreover it
Gamma-converges to \(\widetilde{\mathcal G}_0^\tau\) uniformly in the
sequential sense:
for every \(\tau_n\to\tau\), \(\eps_n\downarrow0\), and \(v_n\to v\) in
\(L^1\),
\begin{equation}
 \widetilde{\mathcal G}_0^\tau(v)
 \le\liminf_{n\to\infty}
 \widetilde{\mathcal G}_{\eps_n}^{\tau_n}(v_n),
 \label{eq:moving-fan-common-liminf}
\end{equation}
and every binary \(v\in BV(X_{\mathcal K(0)}^\circ)\) has the common
recovery \(v_n=v\), for which
\begin{equation}
 \lim_{n\to\infty}\widetilde{\mathcal G}_{\eps_n}^{\tau_n}(v)
 =\widetilde{\mathcal G}_0^\tau(v).
 \label{eq:moving-fan-common-recovery}
\end{equation}
In particular, the conclusion holds for arbitrary joint diagonals; no rate
condition between \(|\tau_n-\tau|\) and \(\eps_n\) is needed.  This includes
both \(\eps_n/|\tau_n-\tau|\to0\) and
\(|\tau_n-\tau|/\eps_n\to0\), whenever the quotients are defined.  At every time
the limit has the bulk and outer terms of \eqref{eq:fan-limit-energy}, the
two coefficient-one face terms on every \(K_\alpha(\tau)\), and no junction
term.
\end{theorem}

\begin{proof}
Uniform ordinary-box and sector-wedge estimates use the common angle,
reach, atlas multiplicity, and Jacobian bounds.  Threshold rounding therefore
does not increase the energy by
\eqref{eq:moving-fan-uniform-steepness} and gives one \(BV\) bound on the
fixed reference completion; alternatively, the
projected uncontrolled jumps are bounded by
\(C\sup_\tau\sum_\alpha\Hh^{d-1}(K_\alpha(\tau))\).  This proves
equicoercivity.

For the lower bound it is enough to treat a subsequence with bounded energy.
Threshold-round \(v_n\) at time \(\tau_n\), obtaining binary \(z_n\to v\)
without increasing that energy.  For large \(n\), the time-seam estimate and
\(C_{\rm mov}\omega_*(|\tau_n-\tau|)<1/2\) first imply
\[
 \widetilde{\mathcal G}_{\eps_n}^{\tau}(z_n)
 \le 2\widetilde{\mathcal G}_{\eps_n}^{\tau_n}(z_n)
      +2C_{\rm mov}\omega_*(|\tau_n-\tau|),
\]
so the fixed-time energies are bounded.  Applying the seam estimate once
more shows
\(\widetilde{\mathcal G}_{\eps_n}^{\tau}(z_n)
-\widetilde{\mathcal G}_{\eps_n}^{\tau_n}(z_n)\to0\).
 Work entirely on the reference completion.  Exact conjugacy leaves the
reference \(W^{2,\infty}\) arcs unchanged; only the positive
\(W^{1,\infty}\) weights \(\alpha_i^\tau c_i\) and
\(\beta^\tau J_\tau\) occur.  The rounding, ordered-smoothing, weighted
directional-liminf, and constant-recovery proof of
Theorem~\ref{thm:fan-junction-gamma-limit} therefore applies with these
weights.  This reference-space argument is essential: a bi-\(C^{1,1}\)
change of coordinates need not preserve \(W^{2,\infty}\) regularity of the
pushed-forward vector fields.  The uniform estimates of
Lemma~\ref{lem:moving-fan-uniform-estimates} permit the same ordered deletion limit
for every sequence of times.  This proves
\eqref{eq:moving-fan-common-liminf}.  The constant binary recovery from the
fixed theorem, the continuous density estimate
\eqref{eq:moving-fan-density-time-continuity}, and the uniform total estimate
\eqref{eq:moving-fan-total-estimate} give
\eqref{eq:moving-fan-common-recovery}.  The uniform junction cell is zero,
so no time-dependent line density is left unaccounted for.
\end{proof}

\begin{lemma}[Joint recovery under chamberwise strict convergence]
\label{lem:moving-fan-joint-strict-recovery}
Let \(\tau_n\to\tau\), \(\eps_n\downarrow0\), and let
\(v_n,v\in BV_{\rm res}(X_{\mathcal K(0)}^\circ;\{0,1\})\).  Assume that,
on one fixed finite resolved presentation,
\begin{equation}
 v_n\to v\quad\text{in }L^1,\qquad
 |Dv_{n,a}|(C_a)\to|Dv_a|(C_a)\quad\text{for every chamber }C_a,
 \label{eq:moving-fan-chamberwise-strict}
\end{equation}
and that every fixed outer and resolved-face trace converges in \(L^1\).
Then
\begin{equation}
 \widetilde{\mathcal G}_{\eps_n}^{\tau_n}(v_n)
 \longrightarrow \widetilde{\mathcal G}_0^\tau(v).
 \label{eq:moving-fan-joint-strict-energy}
\end{equation}
\end{lemma}

\begin{proof}
The lower bound is \eqref{eq:moving-fan-common-liminf}.  For the upper bound,
delete a junction tube and the finite atlas seams as in the proof of
Theorem~\ref{thm:fan-junction-gamma-limit}.  On every retained chamber, glue
\(v_{n,a}\) to the fixed neighbouring chamber or ghost prescribed by that
cell.  Chamberwise strict convergence and convergence of the fixed-face
traces make these glued fields converge strictly in ordinary \(BV\), by
\eqref{eq:resolved-bv-gluing}.  Reshetnyak continuity therefore gives
convergence of their anisotropic variations.

We record the required joint translation step explicitly.  Let \(w_n\to w\)
strictly in \(BV(U)\), let \(a_n\to a\) uniformly with uniformly bounded
positive \(W^{1,\infty}\) norms, and let \(s_n\downarrow0\).  For the flow of
one reference field \(v_i\),
\begin{equation}
 \frac1{s_n}\int_U a_n(x)
 |w_n(x)-w_n(\Phi_i^{-s_n}(x))|\dd x
 \longrightarrow\int_Ua\dd|v_i\cdot Dw|,
 \label{eq:joint-strict-flow-translation}
\end{equation}
after restricting to a protected inner cell so that every flow segment stays
in \(U\).  For the lower bound, the signed measures
\[
 a_n\frac{w_n-w_n\circ\Phi_i^{-s_n}}{s_n}\,\mathcal L^d
 \stackrel{*}{\rightharpoonup}a\,v_i\cdot Dw
\]
by testing, changing variables along the flow, and using \(w_n\to w\) in
\(L^1\).  Lower semicontinuity of total variation gives the liminf in
\eqref{eq:joint-strict-flow-translation}.  For the reverse inequality, the
one-dimensional \(BV\) estimate along flow lines gives
\begin{equation}
 \frac1s\int_Ua_n|w_n-w_n\circ\Phi_i^{-s}|\dd x
 \le\int_0^1\int_U
 a_n(\Phi_i^{\theta s}(y))J\Phi_i^{\theta s}(y)
 \dd|v_i\cdot Dw_n|(y)\dd\theta .
 \label{eq:joint-strict-flow-upper}
\end{equation}
The multiplier on the right converges uniformly to \(a\).  Strict convergence
and Reshetnyak continuity imply
\(\int a\dd|v_i\cdot Dw_n|\to\int a\dd|v_i\cdot Dw|\), proving the limsup.
The same estimate, with the multiplier bounded, also yields
\begin{equation}
 \frac1{|s|}\int |w-w\circ\Phi_i^{-s}|\dd x
 \le C|Dw|(U),\qquad0<|s|\le b\eps_0,
 \label{eq:joint-strict-flow-majorant}
\end{equation}
uniformly on the protected cell.  Apply
\eqref{eq:joint-strict-flow-translation} with \(s_n=\eps_nt\); the majorant
and \(t\in[a,b]\) permit integration in the radial variable.  This gives the
required limsup on every retained cell.  The chart, seam, and
junction errors are uniform by
\eqref{eq:moving-fan-total-estimate}.  Sending first \(n\to\infty\), then the
seam width and junction radius to zero, proves
\eqref{eq:moving-fan-joint-strict-energy}.
\end{proof}

\begin{corollary}[Measurable and continuous recovery families]
\label{cor:moving-fan-recovery-family}
Let \(\tau\mapsto v(\tau)\) be a binary finite-energy family on the reference
completion and define
\begin{equation}
 \mathcal R_\eps(\tau)v(\tau)
 :=v(\tau)\circ P_\tau^{-1}
 \quad\hbox{on }X_{\mathcal K(\tau)}.
 \label{eq:moving-fan-recovery-operator}
\end{equation}
If \(v\) is strongly measurable in \(L^1\), then the physical recovery family
is strongly measurable after the natural embedding by \(p_\tau\).  If
\(v\) is continuous in \(L^1\), then the family is continuous.  If, in
 addition, \(v([0,T])\) is compact in the chamberwise strict topology
\eqref{eq:moving-fan-chamberwise-strict} and all its pulled-back outer and
resolved-face traces form compact subsets of \(L^1\), then
\begin{equation}
 \sup_{\tau\in[0,T]}
 \left|\widetilde{\mathcal G}_\eps^\tau(v(\tau))
       -\widetilde{\mathcal G}_0^\tau(v(\tau))\right|\longrightarrow0.
 \label{eq:moving-fan-uniform-recovery-family}
\end{equation}
Thus for every \(\tau_n\to\tau\), \(\eps_n\to0\), and every such family
with \(v(\tau_n)\to v(\tau)\) chamberwise strictly as in
\eqref{eq:moving-fan-chamberwise-strict}, the same recovery
operator gives a common diagonal recovery:
\begin{equation}
 \mathcal R_{\eps_n}(\tau_n)v(\tau_n)\to
 \mathcal R_0(\tau)v(\tau),\qquad
 \widetilde{\mathcal G}_{\eps_n}^{\tau_n}(v(\tau_n))\to
 \widetilde{\mathcal G}_0^\tau(v(\tau)).
 \label{eq:moving-fan-common-diagonal}
\end{equation}
This gives one recovery operator for time-dependent coefficients and
fixed-incidence moving junctions.
\end{corollary}

\begin{proof}
 The map \((\tau,x)\mapsto P_\tau x\) is Borel under the
\(W^{1,\infty}\)-in-time flow hypothesis and continuous under
\eqref{eq:moving-fan-time-modulus}.  Composition on \(L^1\), with uniformly
bounded Jacobians, preserves strong measurability and continuity.  The
constant recovery contains no timewise selection.  If
\eqref{eq:moving-fan-uniform-recovery-family} failed, there would be
\(\eps_n\downarrow0\) and \(\tau_n\to\tau\), after extraction, for which the
displayed difference stays bounded away from zero.  Compactness of the image
gives chamberwise strict convergence of \(v(\tau_n)\) and convergence of all
fixed-face traces.  Lemma~\ref{lem:moving-fan-joint-strict-recovery} then makes
that difference tend to zero, a contradiction.  This proves the uniform
statement; evaluation on an arbitrary joint sequence gives
\eqref{eq:moving-fan-common-diagonal}.
\end{proof}

\begin{proposition}[Fixed-incidence obstruction to topology change]
\label{prop:moving-fan-topology-obstruction}
No family satisfying \eqref{eq:moving-fan-transport} can change the number of
sheets, face copies, sector vertices, or their incidence graph.  If a new
sheet is created through a limiting smooth neck, that family cannot be
represented by label-preserving identifying maps for which the direct and
inverse bi-\(C^{1,1}\) norms, the relevant reaches, and the nonzero incidence
angles remain uniformly controlled.  Equivalently, along every attempted
fixed-reference representation at least one of those controls degenerates.
If instead a positive-area sheet appears instantaneously, there are bounded
data for which the sharp energies are not continuous in time.  Namely, with
outer state and physical state one and both new face ghosts zero, the energy
jumps from zero to
\begin{equation}
 2\int_{K_{\rm new}}h_{\rm fl}(x,\nu_{K_{\rm new}})\dd\Hh^{d-1}>0.
 \label{eq:moving-fan-birth-energy-jump}
\end{equation}
Hence neither the common state-space identification nor the time-seam estimate
\eqref{eq:moving-fan-time-seam} survives such a change.  Any topology-changing
extension therefore requires a different state-space identification and a
new cell problem.  The jump exhibits failure of continuity for one state; it
is not a state-independent creation cost.
\end{proposition}

\begin{proof}
A label-preserving homeomorphism preserves connected components, boundary
incidence, and the cyclic order of the labelled sectors, proving the first
assertion.  For the limiting assertion, suppose to the contrary that a
sequence of fixed-reference identifications and their inverses has uniformly
bounded \(C^{1,1}\) norms and that the reaches and nonzero incidence angles
stay bounded away from zero.  Arzela--Ascoli, applied to the maps, their first
derivatives, and their inverses, gives after extraction a \(C^1\) bi-Lipschitz
limit and inverse.  The labelled sheets converge in their uniform tubular
charts, so the limit map carries every reference stratum and incidence
relation to its labelled limit.  Its incidence graph is therefore unchanged,
contradicting neck creation.  For instantaneous creation,
apply \eqref{eq:fan-double-face-cost} to the new sheet.  Before creation no
corresponding face copies or contact integral exist, whereas uniform spanning
makes the displayed post-creation contact cost positive.  This contradicts any
vanishing time seam and proves the obstruction.
\end{proof}

\begin{remark}[Zero-cost topology change outside the transport class]
\label{rem:topology-zero-cost-path}
Proposition \ref{prop:moving-fan-topology-obstruction} does not assert that
every topology-changing family has a positive energy gap.  Two explicit
families show the limitation.  First, let a fixed positive-area sheet
\(K_{\rm new}\) be absent for \(\tau<\tau_0\) and present for
\(\tau\ge\tau_0\).  Set the outer state, physical state, and both new face
ghosts equal to one.  The potential, intrinsic variation, and every contact
mismatch then vanish, so the sharp energy is identically zero across
\(\tau_0\), although the incidence graph changes.  The spatial functional
therefore does not assign a state-independent cost to the birth of a cut.

Second, even the incompatible data in
\eqref{eq:moving-fan-birth-energy-jump} need not produce a positive limiting
gap if uniform geometry is abandoned.  In a coordinate ball about
\(x_0\Subset\Omega\), let
\[
 K_r:=\{x_0+(y,0):y\in\R^{d-1},\ |y|<r\},
 \qquad
 H_*:=\sup_{\overline\Omega\times\Sph^{d-1}}h_{\rm fl}<\infty.
\]
Write \(\omega_{d-1}:=\mathcal L^{d-1}(B_1^{d-1})\).
With physical state one and both face ghosts zero, the new-face cost obeys
\begin{equation}
 0<2\int_{K_r}h_{\rm fl}(x,e_d)\dd\Hh^{d-1}
 \le 2H_*\omega_{d-1}r^{d-1}\longrightarrow0.
 \label{eq:shrinking-sheet-vanishing-cost}
\end{equation}
Taking \(r=r(\tau)\downarrow0\) as \(\tau\downarrow\tau_0\) gives a
nontrivial topology-changing family whose sharp energy converges to the
pre-birth value.  There is no contradiction: the front
\(\partial K_r\) has curvature of order \(r^{-1}\), so its tubular radius
and every uniform \(C^{1,1}\) identification degenerate.  The proposition
excludes topology change only inside the uniformly controlled transport
class; it does not exclude zero- or vanishing-cost paths in a weaker geometric
closure.
\end{remark}

\begin{proposition}[Censoring deletes the contact law]
\label{prop:boundary-censoring-counterexample}
Replace \(B_\Omega\) in \eqref{eq:boundary-crossing-flow-energy} by
\(\Omega\times\Omega\).  If \(g\equiv0\) on \(\partial\Omega\) and
\(u_\eps\equiv1\) in \(\Omega\), then every censored bond and the potential
have zero energy.  Nevertheless
\begin{equation}
 \int_{\partial\Omega}h_{\rm fl}(x,\nu_\Omega)
 |\operatorname{Tr}u-g|\dd\Hh^{d-1}
 \ge\kappa\Hh^{d-1}(\partial\Omega)>0.
 \label{eq:boundary-censored-contradiction}
\end{equation}
Thus imposed-exterior and censored interactions have different Gamma-limits;
the latter cannot be used to derive the contact density of the former.
\end{proposition}

\begin{proof}
The constant interior state has no interior difference and lies in a true
well.  The strict inequality follows from uniform spanning in
\eqref{eq:boundary-ambient-frame}.  This contradicts any positive contact
term and proves the claim.
\end{proof}

\section{Explicit constants in ordered-flow smoothing}
\label{app:ordered-flow-constants}

For completeness we record the constants used in
Lemma~\ref{lem:ordered-flow-smoothing}.  With \(C_d\ge1\) dimensional and the
data of \eqref{eq:ordered-flow-constant-data}--
\eqref{eq:ordered-flow-radial-constant}, one may take
\begin{align}
 \eps_{\rm sm}&:=\min\left\{\eps_0,
 \frac{\delta_{\rm atl}}
 {4dc_\rho(1+V_0)e^{d\eps_0c_\rho V_1}},
 \frac1{C_d d c_\rho K_{\rm fr}(1+V_0)(1+V_1)
 e^{C_d d\eps_0c_\rho V_1}}\right\},
 \notag\\
 C_{\rm sm}&:=C_dm_{\rm atl}(1+G_{\rm atl})(1+K_{\rm fr})^2R_\rho
 (1+V_0)(1+V_1+V_2)\notag\\
 &\qquad{}\times
 \exp\!\bigl(C_d d\eps_0b(V_1+D_0)\bigr)
 \left(1+\frac{c_\rho V_0}{\delta_{\rm atl}}\right).
 \label{eq:ordered-flow-explicit-constant}
\end{align}

\section*{Statements and Declarations}

\paragraph{Funding.}
This work was supported by NSFC Grant 12501602, Hunan Provincial Education
Department Grant 24C0055, Hunan Provincial Science and Technology Department
Grant 2025JJ60052, and Xiangtan University Start-up Fund Grant KZ0810769.

\paragraph{Competing interests.}
The author declares that there are no competing interests.

\paragraph{Use of AI-assisted tools.}
OpenAI Codex was used for language editing, structural revision, and
consistency checks.  The author reviewed and verified the mathematical
content, references, and final text and assumes full responsibility for the
manuscript.

\begingroup
\scriptsize
\raggedright
\setlength{\bibsep}{0pt plus 0.2ex}
\setlength{\bibhang}{0pt}
\bibliographystyle{unsrtnat}
\bibliography{paper_I_references}
\endgroup

\end{document}